\documentclass[10pt]{article}
\usepackage[dvips]{graphicx}
\usepackage{amsmath,amsfonts,amssymb,latexsym,epsfig}
\include{references}

\usepackage{multirow}
\usepackage{multicol}
\usepackage{amsthm}    
\usepackage{mathabx}    
\usepackage{mathtools}
\usepackage{verbatim}   
\usepackage{color}      
\usepackage{caption,subcaption}
\usepackage{epstopdf}

\newtheorem{theorem}{Theorem}[section]

\newtheorem{corollary}[theorem]{Corollary}
\newtheorem{prop}[theorem]{Proposition}

\numberwithin{equation}{section}

\newcommand{\E}{{\mathbb E}}

\newcommand{\KK}{{\mathcal K}}
\newcommand{\cL}{\mathcal L}

\newcommand{\M}{{\mathcal M}}
\newcommand{\R}{{\mathbb R}}

\newcommand{\Y}{\mathcal{Y}}
\newcommand{\Z}{{\mathcal Z}}

\newcommand{\bx}{{\bf x}}
\newcommand{\by}{{\bf y}}

\newcommand{\eps}{\epsilon}

\begin{document}
\setlength{\baselineskip}{10pt}
\title{Mean Field Limits for Interacting Diffusions in a Two-Scale Potential}
\author{S.~ N. Gomes and G.~A. Pavliotis \\   Department of Mathematics\\
    Imperial College London \\
        London SW7 2AZ, UK}
\maketitle

\begin{abstract}
In this paper we study the combined mean field and homogenization limits for a system of weakly interacting diffusions moving in a two-scale, locally periodic  confining potential, of the form considered in~\cite{DuncanPavliotis2016}. We show that, although the mean field and homogenization limits commute for finite times, they do not, in general, commute in the long time limit. In particular, the bifurcation diagrams for the stationary states can be different depending on the order with which we take the two limits. Furthermore, we construct the bifurcation diagram for the stationary McKean-Vlasov equation in a two-scale potential, before passing to the homogenization limit, and we analyze the effect of the multiple local minima in the confining potential on the number and the stability of stationary solutions. 
\end{abstract}

\section{Introduction}
\label{sec:intro}

Systems of interacting particles, possibly subject to thermal noise, arise in several applications, ranging from standard ones such as plasma physics and galactic 
dynamics~\cite{BinneyTremaine2008} to dynamical density functional theory~\cite{GodPavlal-2012, GodPavlKall11}, mathematical biology~\cite{Farkhooi2017,Lucon2016} and 
even in mathematical models in the social sciences~\cite{GPY2017,Motsch2014}.  As examples of models of interacting ``agents" in a noisy environment that appear in the social 
sciences -- which has been the main motivation for this work -- we mention the modeling of cooperative behavior~\cite{Dawson1983}, risk management~\cite{GPY2012} and 
opinion formation~\cite{GPY2017}. Another recent application that has motivated this work is that of global optimization~\cite{Pinnau_al2017}. 

In this work we will consider a system of interacting particles in one dimension, moving in a confining potential, that interact through their mean, i.e. a Curie-Weiss type interaction~\cite{Dawson1983}:
\begin{equation}\label{e:inter_diff}
d X^i_t = \left(-V'(X_t^i) - \theta \left( X_t^i  - \frac{1}{N} \sum_{j=1}^N X_t^j \right) \right) \, dt + \sqrt{2 \beta^{-1}} \, dB_t^i.
\end{equation}
Here ${\bf X}_t:= \{  X_t^i \}_{i=1}^N$ denotes the position of the interacting agents, $V(\cdot)$ a confining potential, $\theta$ the strength of the interaction between 
the agents, $\{  B_t^i \}_{i=1}^N$ standard independent one-dimensional Brownian motions and $\beta$ denotes the inverse temperature. The total energy (Hamiltonian) 
of the system of interacting diffusions~\eqref{e:inter_diff} is
\begin{equation}\label{e:energy}
W_N({\bf X}) = \sum_{\ell =1}^N V(X^{\ell}) + \frac{\theta}{4 N} \sum_{n=1}^N \sum_{\ell =1}^N (X^n - X^{\ell})^2.
\end{equation}
Passing rigorously to the mean field limit as $N\rightarrow\infty$ using, for example, martingale techniques~\cite{Dawson1983,Gartner1988,Oelschlager1984}, 
and under appropriate assumptions on the confining potential and on the initial conditions (propagation of chaos), is a well-studied problem. 
Formally, using the law of large numbers we deduce that 
$$
\lim_{N \rightarrow +\infty} \frac{1}{N} \sum_{j=1}^N X^j_t = \E X_t,
$$ 
where the expectation is taken with respect to the ``1-particle" distribution function $p(x,t)$.\footnote{This corresponds to the mean field ansatz for the $N-$particle 
distribution function, $p_N(x_1, \dots x_N,t ) = \prod_{n=1}^N p(x_n,t)$ and passing to the limit as $N\rightarrow \infty$. See~\cite{MartzelAslangul2001,balescu97}.} 
Passing, formally, to the limit as $N\rightarrow \infty$ in the stochastic differential equation~\eqref{e:inter_diff} we obtain the McKean SDE
\begin{equation}\label{e:mckean-sde}
d X_t = -V'(X_t) \, dt - \theta (X_t - \E X_t) \, dt + \sqrt{2 \beta^{-1}} \, dB_t.
\end{equation}
The Fokker-Planck equation corresponding to this SDE is the McKean-Vlasov equation~\cite{frank04,McKean1966, McKean1967}
\begin{equation}\label{e:mckean-vlasov}
\frac{\partial p}{\partial t} = \frac{\partial}{\partial x} \left(V'(x) p + \theta \left(x - \int_{\R} x p(x,t) \, dx \right) p + \beta^{-1} \frac{\partial p}{\partial x} \right).
\end{equation} 
The McKean-Vlasov equation is a nonlinear, nonlocal Fokker-Planck type equation that we will sometimes refer to as the McKean-Vlasov-Fokker-Planck equation. It is a gradient flow, with respect to the Wasserstein metric, for the free energy functional
\begin{equation}\label{e:free-energy}
\mathcal{F}[\rho] = \beta^{-1} \int \rho \ln \rho \, dx + \int V \rho \, dx + \frac{\theta}{2} \int \int F(x-y) \rho(x) \rho(y) \, dx dy,
\end{equation}
with $F(x) = \frac{1}{2} x^2$. Background material on the McKean-Vlasov equation can be found in, e.g.~\cite{frank04, CMV2006,Villani2003}.

The finite dimensional dynamics~\eqref{e:inter_diff} has a unique invariant measure. 
Indeed, the process ${\bf X}_t$ defined in~\eqref{e:inter_diff} with $V$ being a confining potential is always 
ergodic, and in fact reversible, with respect to the Gibbs measure~\cite[Ch. 4]{Pavl2014}, 
\begin{equation}\label{e:gibbs-N}
\mu_N (dx) = \frac{1}{Z_N} e^{-\beta W_N(x^1, \dots x^N)} \, dx^1 \dots dx^N, \quad Z_N = \int_{\R^N} e^{-\beta W_N(x^1, \dots x^N)} \, dx^1 \dots dx^N
\end{equation}
where $W_N(\cdot)$ is given by~\eqref{e:energy}.

On the other hand, the McKean dynamics~\eqref{e:mckean-sde} and the corresponding McKean-Vlasov-Fokker-Planck equation~\eqref{e:mckean-vlasov} can have more than one invariant measures, for nonconvex confining potentials and at sufficiently low temperatures~\cite{Dawson1983,Tamura1984}. This is not surprising, since the McKean-Vlasov equation is a nonlinear, nonlocal PDE and the standard uniqueness of solutions for the linear (stationary) Fokker-Planck equation does not apply~\cite{BKRS2015}. 

The density of the invariant measure(s) for the McKean dynamics~\eqref{e:mckean-sde} satisfies the stationary nonlinear Fokker-Planck equation
\begin{equation}\label{e:mckean-station}
\frac{\partial}{\partial x} \left(V'(x) p_{\infty} + \theta \left(x - \int_{\R} x p_{\infty}(x) \, dx \right) p_{\infty} + \beta^{-1} \frac{\partial p_{\infty}}{\partial x} \right) =0.
\end{equation}
Based on earlier work~\cite{Dawson1983,Tamura1984}, it is by now well understood that the number of invariant measures, i.e. the number of solutions 
to~\eqref{e:mckean-station}, is related to the number of metastable states (local minima) of  the confining potential -- see~\cite{Tugaut2014} and the references therein.  

For the Curie-Weiss (i.e. quadratic) interaction potential a one-parameter family of solutions to the stationary McKean-Vlasov equation~\eqref{e:mckean-station} can be obtained:
\begin{subequations}\label{e:inv-meas-mckean}
\begin{eqnarray}
p_{\infty}(x ; \theta, \beta, m) &=& \frac{1}{Z(\theta, \beta ; m )} e^{- \beta \left( V(x) + \theta \left(\frac{1}{2}x^2 - x m \right) \right)}, 
\\  Z(\theta, \beta ; m ) &=& \int_{\R} e^{- \beta \left( V(x) + \theta \left(\frac{1}{2}x^2 - x m \right) \right)} \, dx.
\end{eqnarray}
\end{subequations}
This one-parameter family of probability densities is subject, of course, to the constraint that it provides us with the correct formula for the first moment:
\begin{equation}\label{e:self-consist}
m = \int_{\R} x p_{\infty}(x ; \theta, \beta, m) \, dx =: R(m; \theta, \beta).
\end{equation}
We will refer to this as the {\bf selfconsistency} equation and it will be the main object of study of this paper. Once a solution to~\eqref{e:self-consist} has been obtained, 
substitution back into~\eqref{e:inv-meas-mckean} yields a formula for the invariant density $p_{\infty}(x ; \theta, \beta, m)$. 

Clearly, the number of invariant measures of the McKean-Vlasov dynamics is determined by the number of solutions to the selfconsistency equation~\eqref{e:self-consist}. It is well known and not difficult to prove that for symmetric nonconvex confining potentials a unique invariant measure exists at sufficiently high temperatures, whereas more than one invariant measures exist below a critical temperature $\beta^{-1}_c$~\cite[Thm. 3.3.2]{Dawson1983},~\cite[Thm. 4.1, Thm. 4.2]{Tamura1984}, see also~\cite{shiino1987}. In particular, for symmetric potentials, $m = 0$ is always a solution to the selfconsistency equation~\eqref{e:self-consist}. Above $\beta_c$, i.e. at sufficiently low temperatures, the zero solution loses stability and a new branch bifurcates from the $m = 0$ solution~\cite{shiino1987}. This second order phase transition is similar to the one familiar from the theory of magnetization and the study of the Ising model. In Figure~\ref{fig:Shiino_intro} we present the solution to the selfconsistency equation and the bifurcation diagram for stationary solutions of the McKean-Vlasov equation for the standard bistable -- Landau -- potential  $V(x) = \frac{x^4}{4}-\frac{x^2}{2}$.
\begin{figure}[h!]
\begin{subfigure}{0.5\textwidth}
\includegraphics[width=\textwidth]{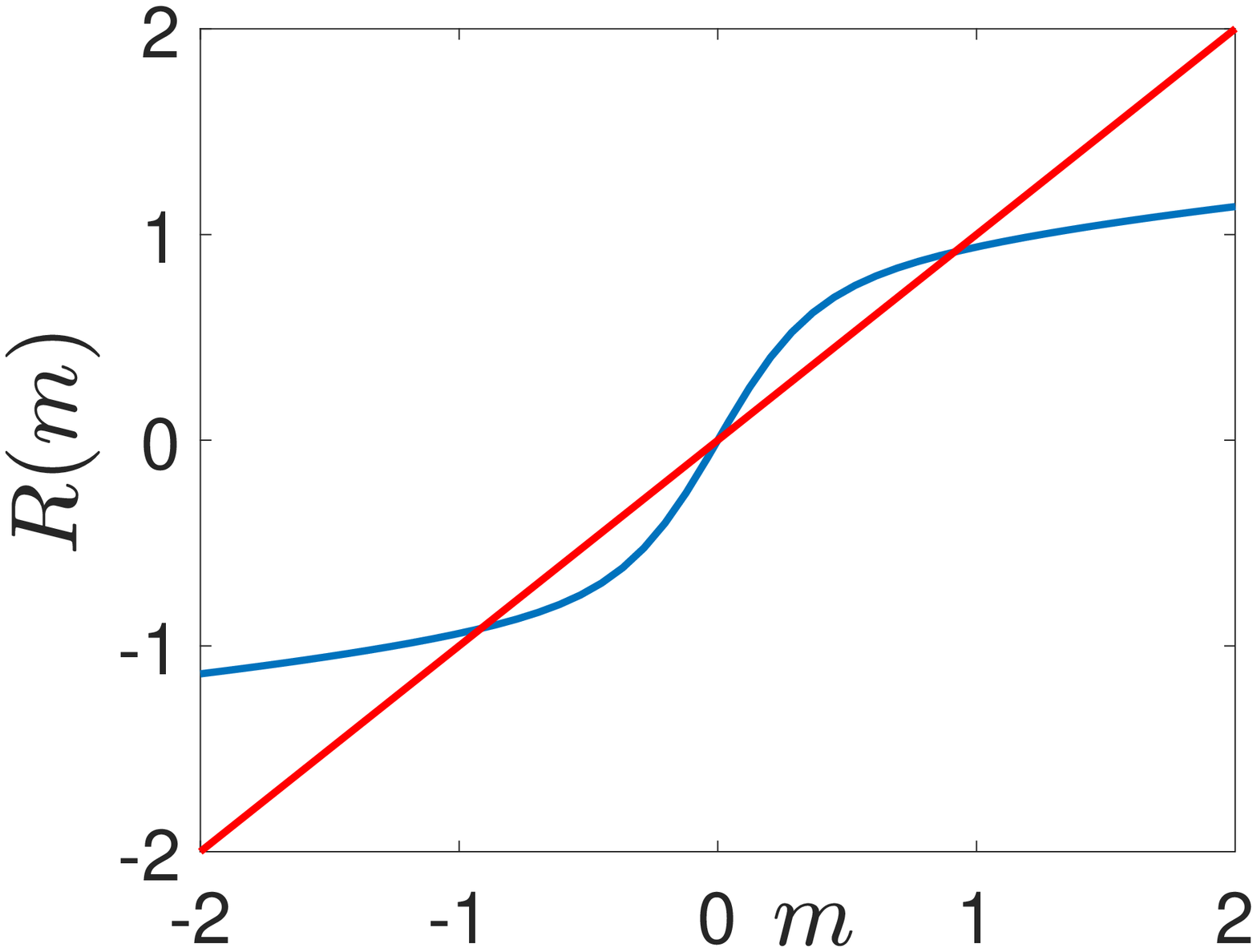}  
\caption{$R(m;\theta,\beta)=m$}  \label{fig:Shiino}
\end{subfigure}
\begin{subfigure}{0.5\textwidth}
\includegraphics[width=\textwidth]{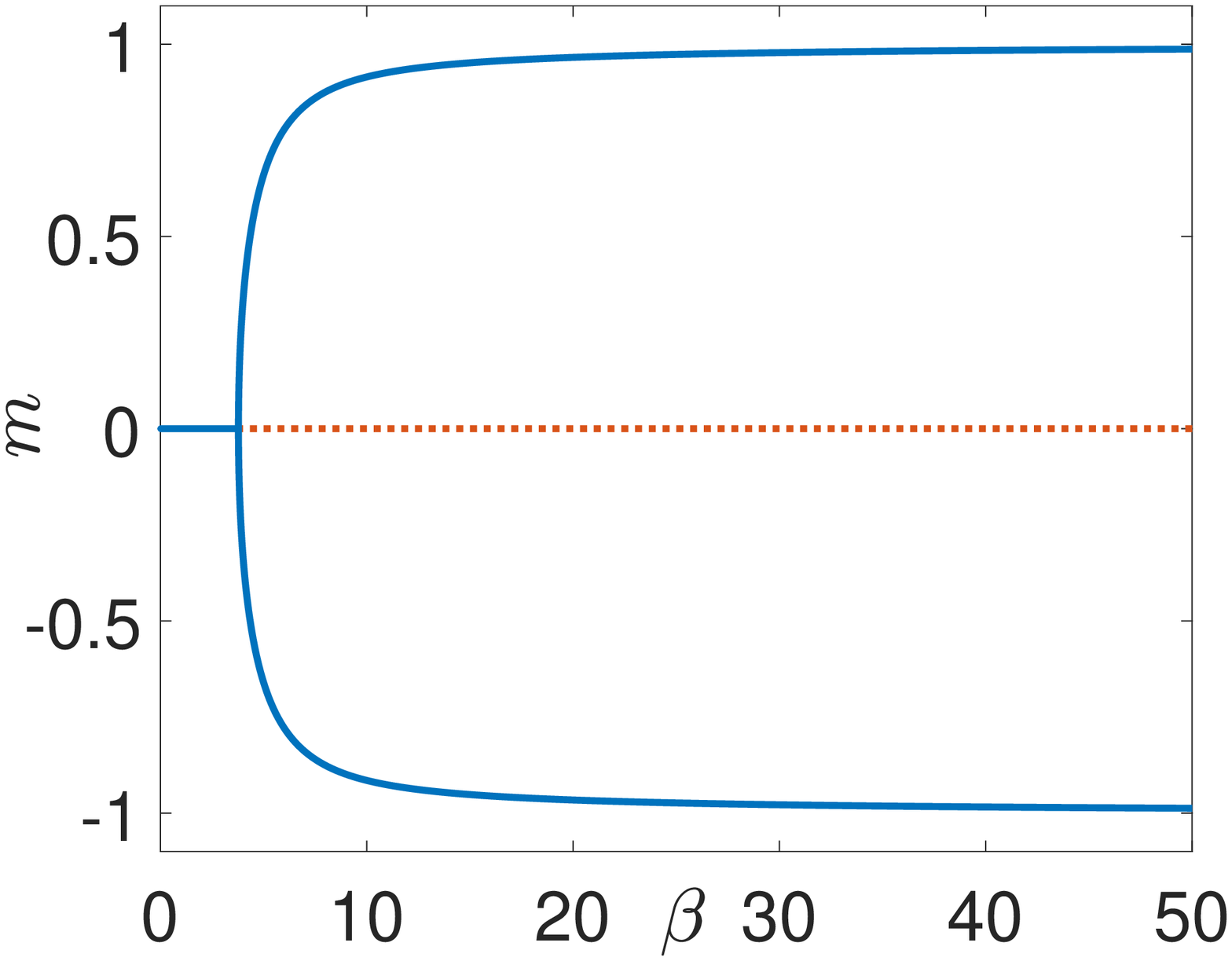}
\caption{Bifurcation diagram}\label{fig:bif_bistable_intro}
\end{subfigure}
    \caption{\eqref{fig:Shiino} Plot of $R(m;\theta,\beta)$ and of the straight line $y =x$ for $\theta = 0.5$, $\beta = 10$, and \eqref{fig:bif_bistable_intro} bifurcation diagram of $m$ as a function of $\beta$ for $\theta = 0.5$ for the bistable potential $V(x) = \frac{x^4}{4}-\frac{x^2}{2}$ and interaction potential $F(x) = \frac{x^2}{2}$.}\label{fig:Shiino_intro}
\end{figure}

To compute the critical temperature we need to solve the equation obtained by differentiating the selfconsistency equation with respect to the order parameter $m$ at $m=0$: 
\begin{equation}\label{e:critical_temp}
\operatorname{Var}_{p_{\infty}}(x)\Big|_{m=0} := \int x^2 p_\infty(x;m=0, \beta, \theta) \ dx - \left( \int x p_\infty(x;m=0, \beta, \theta) \ dx )\right)^2 = \frac{1}{\beta\theta}.
\end{equation} 

The main purpose of this paper is to study the dynamics and, in particular, bifurcations and phase transitions for a system of interacting diffusions moving in a rugged energy landscape, 
coupled through the Curie-Weiss interaction. We are particularly interested in understanding the combined effect of the presence of several local minima (metastable states) 
in the confining potential and of the passage to the mean field limit. 
We will study the problem for a system of interacting diffusions of the form~\eqref{e:inter_diff} moving in a two-scale, locally periodic confining potential 
\begin{equation}\label{e:pot-multisc}
V^{\eps}(\bx) = V\left(\bx, \frac{\bx}{\eps} \right),
\end{equation}
where $V: \, (x,y)\in \R\times\Y\rightarrow \R$, $\Y$ denotes a periodic box in $\R^d$, $\Y = [0,L]^d$:
\begin{equation}\label{e:Y_defn}
V(x,y+kLe_i) = V(x,y), \quad k\in \mathbb{Z}, \quad i \in \left\{1,\dots,d\right\},
\end{equation}
and $\left\{e_1,\dots,e_d\right\}$ is the canonical basis of $\R^d$. Throughout this paper, $L = 2\pi$. The particles $\{ X_t^i, \,  i=1,\dots,N \}$ are 
interacting through the Curie-Weiss interaction, $F(x) = \frac{x^2}{2}$. 
This class of potentials provides us with a natural testbed for testing several techniques and methodologies for the study of multiscale diffusions such as maximum likelihood 
estimation~\cite{PapPavSt08,PavlSt06}, particle filters and filtering~\cite{papav-2007,Imkeller2012}, importance sampling and large deviations~\cite{spiliop2013} and 
optimal  control~\cite{HLZP2014}.

Of particular relevance to us is the multiscale analysis presented in~\cite{Pavliotis_al2016, DuncanPavliotis2016}.\footnote{In fact, in these papers a potential with $N$ microscales 
and one macroscopic scale of the form $V^{\eps}(\bx) = V\left(\bx, \frac{\bx}{\eps}, \frac{\bx}{\eps^2}, \dots \frac{\bx}{\eps^N} \right)$, where $V$ is periodic in all the 
microvariables is studied. For the purposes of this work it is sufficient to consider a potential with two characteristic, widely separated, length scales. } 
In these works, the homogenized SDE for a Brownian particle moving in a two-scale potential in $\R^d$, valid in the limit of infinite scale separation $\eps \rightarrow 0$ was obtained and the effect of the multiscale structure on noise-induced 
transitions was investigated. It was shown, in particular, that the homogenized SDE is characterized by multiplicative noise. 
For a single Brownian particle in $\R^d$ moving in a two-scale potential~\eqref{e:pot-multisc} (or, equivalently, for a system of $d$ non-interacting Brownian particles in a $2$--scale potential) the homogenized equation reads
\begin{equation}\label{e:homog_intro}
dX_t = -\mathcal{M}(X_t) \nabla \Psi(X_t) \, dt + (\nabla \cdot \mathcal{M})(X_t) \, dt + \sqrt{2 \mathcal{M}(X_t)} \, dB_t,
\end{equation}
where $\mathcal{M}(\cdot)$ denotes the diffusion tensor and $\Psi (\cdot)$ the free energy--see Section~\ref{sec:epsilon_first}. 
It is important to note that, in addition to the presence of multiplicative noise, the potential energy driving the dynamics is not simply the average of the two-scale potential over 
its period, but, rather, the free energy  $\Psi = - \beta^{-1} \ln \left(\int e^{-\beta V(x,y)} \, dy \right)$. Since the dynamics~\eqref{e:homog_intro} is finite dimensional, 
no phase transitions can occur. In fact, the homogenized dynamics is reversible with respect to the thermodynamically consistent Gibbs measure, 
see the discussion in Section~\ref{sec:epsilon_first}. It is well known, however, that multiplicative noise can lead to noise-induced transitions, i.e. to changes in the topological 
structure of the invariant measure~\cite{HorsLef84},~\cite[Sec. 5.4]{Pavl2014}. Such phenomena, including multiscale-induced hysteresis effects, for a one-dimensional Brownian particle moving in a multiscale potential, were studied in detail in~\cite{Pavliotis_al2016}. 

Our goal is to study mean field limits for multiscale interacting diffusions of the form
\begin{equation}
\label{eq:system_of_sdes}
d X^{\eps,i}_t = -\nabla V^\epsilon(X_t^{\eps,i})\,dt - \frac{\theta}{N}\sum_{j=1}^{N} \nabla F(X_t^{\eps,i} - X_t^{\eps,j})\,dt + \sqrt{2\beta^{-1}}dB_t^i,  
\end{equation} 
where the two-scale potential is given by~\eqref{e:pot-multisc}. The interaction potential $F(\cdot)$ is assumed to be a smooth even function, with $F(0) = 0$ and $F'(0)= 0$. 
All of the numerical experiments that we will present will be for the Curie-Weiss quadratic interaction potential $F(x) = \frac{1}{2}x^2$. Although we will mainly consider the case of interacting particles moving in one dimension, (at least parts of) the multiscale analysis that we will present is also valid in arbitrary dimensions.

The main issues that we address in this work are:

\begin{enumerate}

\item What is the effect of the presence of (infinitely) many local minima in the locally periodic confining potential on the bifurcation diagram? In other words, how do the bifurcation diagrams for $\eps \ll 1$ but finite and $\eps \rightarrow 0$ differ?

\item Do the homogenization and mean field limits commute, in particular when also passing to the long time limit $T \rightarrow+\infty$? In other words: are the bifurcations diagrams corresponding to the $N\rightarrow\infty, \, T\rightarrow\infty, \, \epsilon\rightarrow 0$ and $\epsilon\rightarrow 0, \, N\rightarrow\infty, \, T\rightarrow\infty$ limits the same? 

\end{enumerate}

Two typical examples of the type of locally-periodic potentials that we will study in this paper are shown in Figure~\ref{fig:pot}:
\begin{figure}[t!]
\begin{subfigure}{0.5\textwidth}
\includegraphics[width=0.95\textwidth]{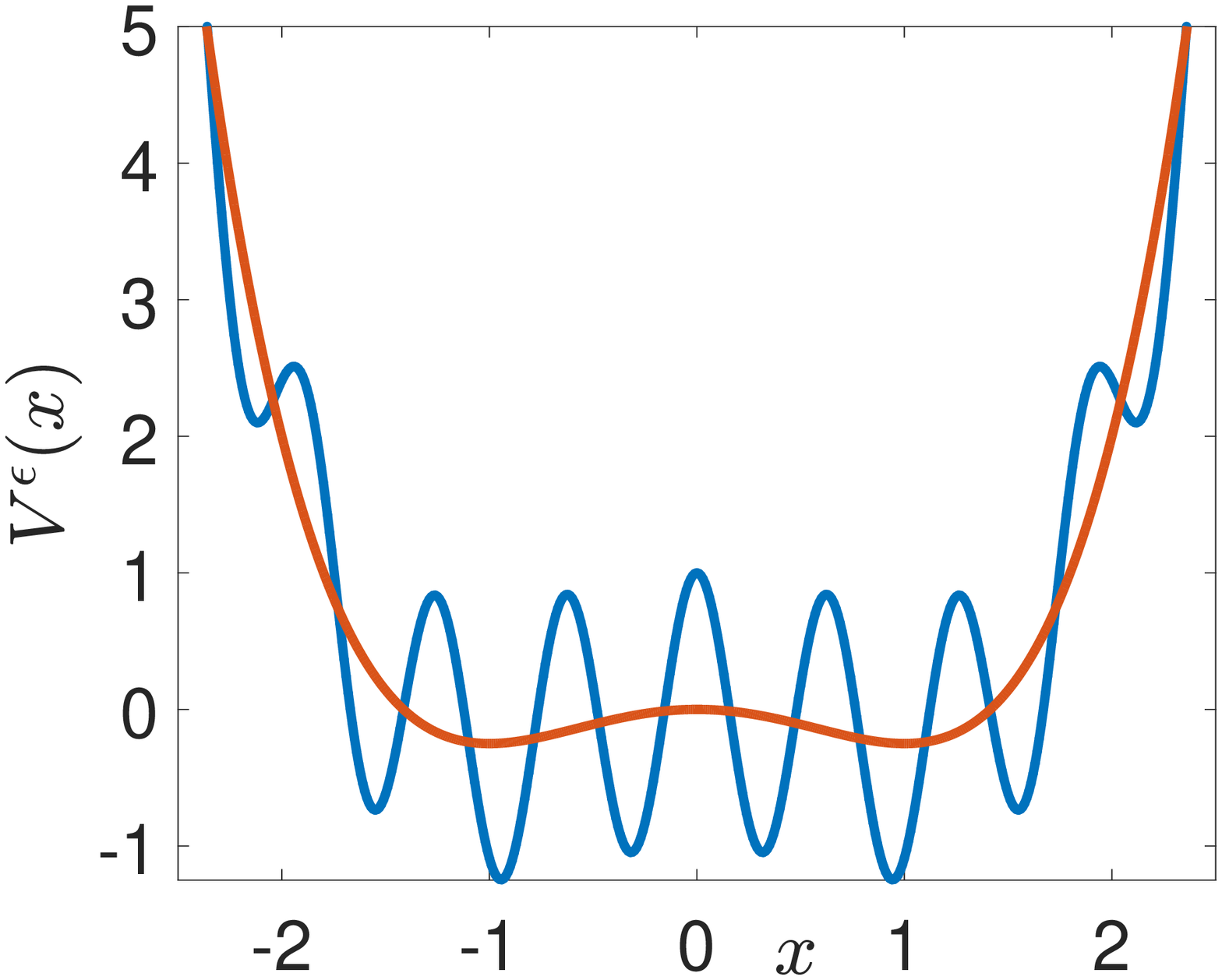}
\caption{Separable fluctuations}\label{fig:sep_pot}
\end{subfigure}
\begin{subfigure}{0.5\textwidth}
\includegraphics[width=0.95\textwidth]{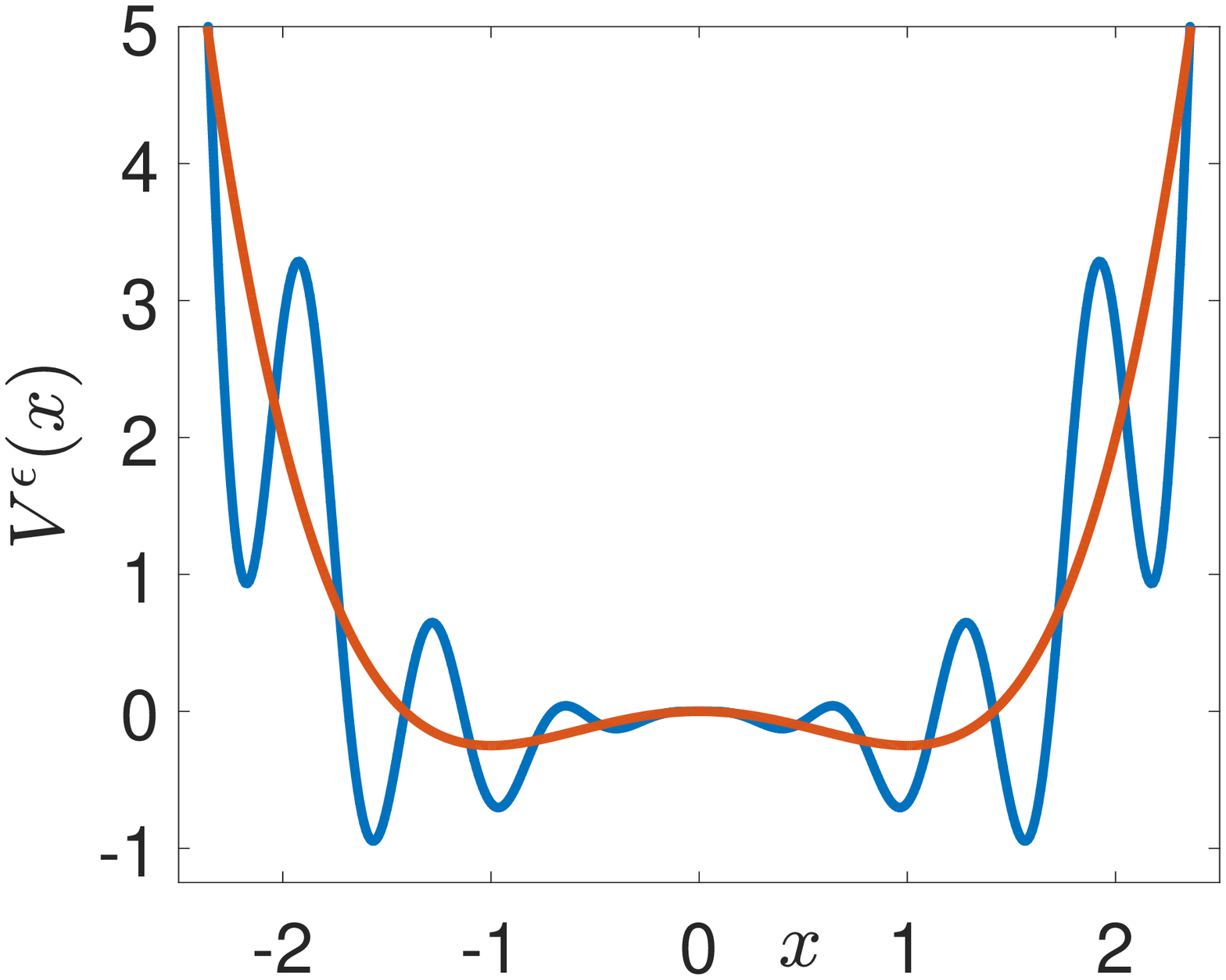}
\caption{Nonseparable fluctuations}\label{fig:non_sep_pot}
\end{subfigure}
\caption{Bistable potential with (left) separable  and (right) nonseparable  fluctuations.}\label{fig:pot}
\end{figure}
\begin{equation}
V^{\eps}(x) = \frac{x^4}{4} - \frac{x^2}{2} + \delta \cos \left(\frac{x}{\eps} \right) \quad \mbox{and} \quad V^{\eps}(x) = \frac{x^4}{4} - \left( 1 - \delta \cos \left(\frac{x}{\eps} \right) \right) \frac{x^2}{2}.
\end{equation}
It should be clear from these two figures that the homogenization and mean field limits, when also combined with the long time limit, do not necessarily commute. 
First, the homogenization process tends to smooth out local minima and to even ``convexify" the confining potential -- think of a quadratic potential perturbed by 
fast periodic fluctuations. This implies, in particular, that even though many additional stationary solutions, i.e. branches in the bifurcation diagram may appear for all finite values of $\epsilon$, most, if not all, of them may not be present in the bifurcation diagram for the homogenized dynamics. Furthermore, multiplicative/nonseparable fluctuations of the type presented in Figure~\ref{fig:pot}(b) tend to flatten the potential around $x=0$. As we will see in Section~\ref{sec:numerics}, this phenomenon is very much related to the lack of commutativity of the limits $N\rightarrow\infty, \, T\rightarrow\infty, \, \epsilon\rightarrow 0$ and $\epsilon\rightarrow 0, \, N\rightarrow\infty, \, T\rightarrow\infty$.

We will study these problems using a combination of formal multiscale calculations, (some) rigorous analysis and extensive numerical simulations. There are many technical issues that we do not address, such as the rigorous homogenization study of the McKean-Vlasov equation and the rigorous study of bifurcations in the presence of infinitely many local minima. We will address these in future work.

The rest of the paper is organized as follows. In Section~\ref{sec:epsilon_first} we study the mean field limit for a system of homogenized interacting diffusions, i.e. the first $\epsilon\rightarrow 0$, then $N\rightarrow\infty$ limit. In Section~\ref{sec:N_first} we study the homogenization problem for the McKean-Vlasov equation in a two-scale potential. In Section~\ref{sec:numerics} we present extensive numerical simulations. Section~\ref{sec:conclusions} is reserved for conclusions. 

%
%
\section{Mean field limit of the homogenized interacting diffusions: first $\epsilon\rightarrow 0$, then $N\rightarrow\infty$}\label{sec:epsilon_first}
In this section we consider the one dimensional version of the system of SDEs~\eqref{eq:system_of_sdes}. We consider the first homogenization, then mean field order of limits. The homogenization theorem for a system of finite dimensional interacting diffusions moving in a two-scale confining potential is presented in~\cite{DuncanPavliotis2016}. The mean field limit of the homogenized SDE system can be obtained by using the results of~\cite{Gartner1988,Oelschlager1984}.
%
%
\subsection{Homogenization for finite system of interacting diffusions in a two-scale potential}
We consider the system of interacting diffusions
\begin{equation}
\label{eq:system_of_sdes_in_1D}
dX_t^i = -\partial_x V^\epsilon(X_t^i)\,dt - \frac{\theta}{N}\sum_{j=1}^{N} \partial_x F(X_t^i - X_t^j)\,dt + \sqrt{2\beta^{-1}}dB_t^i,
\end{equation} 
where $F$ is a smooth even function, with $F(0) = 0$ and $F'(0)= 0$ and $V^\epsilon$ is a smooth locally periodic potential of the form~\eqref{e:pot-multisc} or~\eqref{eq:multiscale_potential}. We introduce the notation $x_t =  (X_t^1, \dots, X_t^N)$, so that we have
\begin{equation}\label{e:System_1eq}
dx_t = -\nabla V^\epsilon(x_t) \, dt - \theta \nabla F(x_t-\bar{x}_t) \, dt +\sqrt{2\beta^{-1}}dB_t,
\end{equation}
where $\bar{x}_t = \frac{1}{N}\sum_{i=1}^N X_t^i$ and $B_t$ is a standard Brownian motion in $\R^N$. This equation is of the same form as \cite[Eqn.(1)]{DuncanPavliotis2016} and \cite[Eqn.(1)]{Pavliotis_al2016}, with $V^\epsilon(X_t^\epsilon)$ replaced by $V^\epsilon(x_t)+\theta F(x_t-\bar{x}_t)$.

Since $F$ does not depend on the fast scale, we can proceed as in \cite{DuncanPavliotis2016} and obtain the homogenized SDE:
\begin{equation}
\label{eq:homogenized_SDE}
dx_t = -\left[\M(x_t)\nabla\Psi_N(x_t) -  \beta^{-1}\nabla\cdot\M(x_t)\right] \, dt + \sqrt{2\beta^{-1}\M(x_t)}dB_t,
\end{equation} 
 where
\begin{equation}
\label{eq:Free_energy}
\Psi_N (\bx) = -\beta^{-1}\ln \Z_N(\bx),
\end{equation}
for
\begin{equation}
\label{eq:partition_function}
\Z_N(\bx) = \int_{\Y} e^{-\beta W_N(\bx,\frac{\bx}{\eps})} \ dy,
\end{equation}
where $W_N(\bx,\by)$ is defined as in Eqn.~\eqref{e:energy},
\begin{equation}\label{e:energy_multiscale}
W_N({\bf X,Y}) = \sum_{\ell =1}^N V(X^{\ell},Y^\ell) + \frac{\theta}{4 N} \sum_{n=1}^N \sum_{\ell =1}^N F(X^n - X^{\ell}).
\end{equation}
It will be useful to decompose the two-scale potential into its large-scale confining part and the modulated, mean-zero, fluctuations:
\begin{equation}
	\label{eq:multiscale_potential}
	V^\epsilon(x)  = V_0(x) + V_1 \left( x,\frac{x}{\eps} \right), \quad V_0(x) = \int_{\Y} V(x, y) \, dy. 
\end{equation}
Notice that this decomposition is not unique, since we can define the average of the two-scale potential over the unit cell with respect to a different, e.g. Gibbs, weight. However, the choice of the weight does not affect our results. See, e.g., the proof of Proposition~\ref{prop:commutative}.
 
We note that the free energy $\Psi_N$ is of the form
\begin{equation}\label{eq:Free_energy_expanded}
\Psi_N(\bx) =  -\left(\sum_{\ell =1}^N V_0(x^{\ell}) + \frac{\theta}{2 N} \sum_{n=1}^N \sum_{\ell =1}^N F(x^n - x^{\ell}) \right)+\psi(\bx),
\end{equation}
where
\begin{equation}\label{e:small_psi_def}
\psi(\bx) = -\beta^{-1}\ln\int_{\Y} \prod_{\ell=1}^N e^{-\beta  V_1(x^{\ell},y^{\ell})}  \ dy^{\ell}.
\end{equation}
Finally, $\M:\mathbb{R}^{d}\rightarrow \mathbb{R}^{d\times d}_{sym}$ is defined by
\begin{equation}
\label{eq:diff_coeff}
\M(\bx) = \frac{\KK(\bx)}{\Z_N(\bx)},
\end{equation}
where
\begin{equation}\label{eq:K(x)}
\KK(\bx) = \int_{\Y} (I + \nabla_y \Phi(\bx, \by))e^{-\beta W_N(\bx,\by)}\,dy,\quad \bx\in\mathbb{R}^d, 
\end{equation}
and,  for fixed $\bx\in\mathbb{R}^d$, $\Phi$ is the unique solution to
\begin{equation}
\label{eq:poisson}
	\nabla_\by\cdot\left(e^{-\beta V_1(\bx,\by)}(I + \nabla_\by\Phi(\bx,\by))\right) = 0,\quad \by\in\Y,
\end{equation}
such that $\int_\Y \Phi(x,y)e^{-\beta V_1(x,y)}\,dy = 0$, for all $x\in\mathbb{R}^d$.

To compute the diffusion tensor, we observe that 
\begin{equation}
\begin{array}{rl}
\nonumber \KK_{ij}(\bx) &= \displaystyle{\delta_{ij} + \frac{1}{\Z(\bx)}\int_{\Y}\frac{\partial \Phi_i}{\partial y_j}(\bx,\by)e^{-\beta \left(\sum_{\ell=1}^N V_0(x^\ell)+\sum_{\ell=1}^N V_1(x^\ell,y^\ell)+\theta F(x^\ell-\bar{x})\right)} \ dy}\\
\label{e:dif_tensor_1} &\displaystyle{=\delta_{ij} + \frac{1}{\bar{\Z}(\bx)} \int_0^L \cdots \int_0^L \frac{\partial \Phi_i}{\partial y_j}(\bx,\by)\prod_{m=1}^N e^{-\beta V_1(x^m,y^m)} \ dy^m,}
\end{array}
\end{equation}
where 
\begin{equation}\label{e:barZ}
\bar{\Z}(\bx) = \prod_{m=1}^N \int_0^L e^{-\beta  V_1(x^m,y^m)} \ dy^m.
\end{equation}

By manipulating the Poisson equation that the function $\Phi(\bx,\by)$ solves, one can conclude that $\Phi(\bx,\by) = (\phi(x_1,y_1),\phi(x_2,y_2),\dots,\phi(x_N,y_N))$, 
where $\phi(x,y)$ solves
\begin{equation}\label{e:poisson-1d}
-\mathcal{L}_0\phi(x,y) = -\frac{\partial V_1}{\partial y}(x,y), \qquad \mathcal{L}_0  = -\partial_y V_1 \partial_y +\beta^{-1}\partial^2_y,
\end{equation}
and therefore $\Phi_i(x,y) = \phi(x_i,y_i)$ and
\[
\frac{\partial \Phi_i}{\partial y_j}(x,y) = \frac{\partial \phi(x^i,y^i)}{\partial y^j} = \delta_{ij}\frac{\partial \phi}{\partial y^j}(x^i,y^i).
\]
Substituting in \eqref{e:dif_tensor_1}, we obtain
\begin{equation}
 \KK_{ij}(\bx) = \displaystyle{=\delta_{ij} + \frac{1}{\bar{\Z}(\bx)} \int_0^L \cdots \int_0^L \delta_{ij} \frac{\partial \phi}{\partial y_j}(x_i,y_i)\prod_{m=1}^N e^{-\beta V_1(x^m,y^m)} \ dy^m,}
\end{equation}
and the diffusion tensor is diagonal, with
\begin{eqnarray*}
\KK_{ii}(\bx) &=& 1 + \frac{1}{\prod_{m=1}^N \int_0^L e^{-\beta V_1(x^m,y^m)} \ dy^m} \left(\int_0^L  \frac{\partial \phi}{\partial y^i}(x^i,y^i) e^{-\beta V_1(x^i,y^i)} \ dy^i\right) \times \\ && \int_0^L \prod_{m=1, m\neq i}^N e^{-\beta V_1(x^m,y^m)} \ dy^m \\
& =&  \displaystyle{1 + \frac{1}{ \int_0^L e^{-\beta V_1(x^i,y^i)} \ dy^i} \int_0^L  \frac{\partial \phi}{\partial y^i}(x^i,y^i) e^{-\beta V_1(x^i,y^i)} \ dy^i.}
\end{eqnarray*}
As it is well known~\cite[Sec 13.6.1]{PavlSt08}, the one dimensional Poisson equation~\eqref{e:poisson-1d} can be solved explicitly, up to quadratures. We can then obtain formulas for the diagonal elements of the diffusion tensor $\KK_{ii}$ and of $\M(x)$: 
\begin{equation}
\label{eq:Diffusion_coeff}
\M(x) = \frac{1}{\left(\frac{1}{L}\int_0^L e^{-\beta V_1(x,y)} \, dy\right)\left(\frac{1}{L}\int_0^L e^{\beta V_1(x,y)} \, dy\right)}.
\end{equation}

We can write the system of stochastic differential equations for the homogenized system of interacting particles:
\begin{equation}
\label{eq:system_of_homogenized_sdes}
dX_t^i = -\left[\M(X_t^i) \partial_{x_i}\Psi(X_t^1, \dots X_t^N) -  \beta^{-1}\M'(X_t^i)\right] \, dt + \sqrt{2\beta^{-1}\M(X_t^i)}dB_t^i,
\end{equation} 
for $i=1,\dots,N$, where $\M$ is defined in Eqn.~\eqref{eq:Diffusion_coeff} above, prime denotes differentiation with respect to $x$ and $\Psi$ is given by
\begin{equation}\label{e:psi_homog_first}
\Psi(x^1,\dots,x^N) = \sum_{\ell =1}^N V_0(x^{\ell}) +  \frac{\theta}{4 N} \sum_{n=1}^N \sum_{\ell =1}^N F(x^n - x^{\ell})-\beta^{-1}\ln\int_0^L\prod_{\ell = 1}^N e^{-\beta V_1(x^{\ell},y^{\ell})} \ dy^\ell.
\end{equation}

We note that the homogenized system of SDEs~\eqref{eq:system_of_homogenized_sdes} is characterized by multiplicative noise.\footnote{In fact, the noise in this SDE can be interpreted in the Klimontovich sense:
\begin{equation*}
dX_t^i = - \M(X_t^i) \partial_{x_i}\Psi'(X_t^1, \dots X_t^N) \, dt + \sqrt{2\beta^{-1}\M(X_t^i)} \circ^{K} dB_t^i,
\end{equation*}
where $\circ^{K}$ denotes the Klimontovich stochastic integral; see~\cite{DuncanPavliotis2016}. 
} Furthermore, the diffusion coefficient of the $i-$th particle depends only on the position of the particle itself, and not of the other particles. The dynamics~\eqref{eq:system_of_homogenized_sdes} is reversible with respect to the Gibbs measure
\begin{equation}\label{e:Gibbs_measure}
p_{\infty}(dx) = \frac{1}{\bar{Z}} e^{-\beta\Psi(x)} \, dx, \quad \bar{Z} = \int_\R e^{-\beta\Psi(x)} \ dx.
\end{equation}

%
%
\subsection{Mean field limit for the homogenized SDE}\label{sec:epsilonThenN}
We can now pass to the mean field limit $N\rightarrow\infty$. 
The system of SDEs~\eqref{eq:system_of_homogenized_sdes} is of the form
\[
dX_t^i = b\left(X_t^i,\frac{1}{N}\sum_{j=1}^N X_t^j\right) dt + \sigma(X_t^i) dB_t^i,
\]
which is in the same form to the one considered~\cite{Gartner1988,Oelschlager1984}, with slightly different drift and diffusion coefficients.\footnote{In fact, 
these papers consider the more general case, where the diffusion coefficient, $\sigma$, also depends on the empirical measure, 
$\sigma\left(X_t^i, \frac{1}{N}\sum_{j=1}^N X_t^j\right)$.} 
It is straightforward to check that the homogenized equation satisfies the conditions in the aforementioned papers.\footnote{These are variants of boundedness and Lipschitz continuity assumptions for the drift and diffusion coefficients. The estimates on the homogenized coefficients that are obtained in~\cite{Abdulle2017}, are sufficient in order to invoke the results of~\cite{Gartner1988,Oelschlager1984}. For the purposes of this paper it is sufficient to pass formally to the mean field limit. The rigorous analysis will be presented elsewhere.} Taking the mean field limit of~\eqref{eq:system_of_homogenized_sdes} we obtain the following nonlinear Fokker-Planck equation:
\begin{equation}
\label{eq:FP_for_homogenized}\frac{\partial p}{\partial t} =\frac{\partial}{\partial x}\left[ \beta^{-1}\frac{\partial\left(\M(x) p\right)}{\partial x} + \M(x)\left( V'_0(x) +\psi'(x) + \theta\left( F' \star p\right)(x)\right)p  + \beta^{-1}\frac{\partial \M(x)}{\partial x}p\right], 
\end{equation}
where
\begin{equation}
\label{eq:small_psi}
\psi(x) = -\beta^{-1}\ln \left( \int_0^L e^{-\beta V_1(x,y)} \ dy \right),
\end{equation}
and $\M(x)$ is defined in~\eqref{eq:diff_coeff}.

The McKean stochastic differential equation corresponding to~\eqref{eq:FP_for_homogenized} is
\begin{equation}\label{e:sde-homog}
dX_t = -\M (X_t)  (V'_0(X_t) + \psi'(X_t) +\theta F'(X_t-\bar{X})) \, dt + \beta^{-1}  \M'(X_t) \, dt + \sqrt{2 \beta^{-1} \M(X_t)} \,  dB_t.
\end{equation}
We reiterate that the correction to the drift, $\beta^{-1} \M'(X_t) \, dt$ is not the Stratonovich correction, but, rather the Klimontovich (kinetic) one. 
This interpretation of the stochastic integral ensures that the homogenized dynamics is reversible with respect to the (thermodynamically consistent) Gibbs measure(s) that we can calculate by solving the stationary Fokker-Planck equation.

The (one or more) stationary distributions $p_\infty(x)$ are solutions to the stationary Fokker-Planck equation
\begin{equation}
\label{eq:steady_FP_homog}
\cL^* p_{\infty}: = \frac{\partial}{\partial x}\left(\M(x)\left( V'_0(x) + \psi'(x) +  (F'\star p_\infty) p_\infty + \beta^{-1}p_\infty\right) + \beta^{-1}\frac{\partial(\M(x)p_\infty)}{\partial x}\right) = 0.
\end{equation}
The detailed balance condition implies that
\[
\beta^{-1}\M(x)\frac{\partial p_\infty}{\partial x} = -\M(x)\left(V'_0(x)+ (F'\star p_\infty)(x) + \psi'(x)\right)p_\infty,
\]
And since $\M(x)$ is strictly positive, a simple variant of~\cite[Lemma 4.1]{Tamura1984} enables us to obtain an integral equation for the invariant distribution:
\begin{equation}
\label{eq:invariant_measure_homog_mean_field}
p_\infty(x) = \frac{1}{Z}e^{-\beta\left(V_0(x)+\theta (F\star p_\infty)(x)+\psi(x)\right)},
 \quad
Z = \int_\R e^{-\beta\left(V_0(x)+\theta (F\star p_\infty)(x)+\psi(x)\right)} \ dx,
\end{equation}
where $\psi(x)$ is given by Eqn.~\eqref{eq:small_psi}. In particular, $p_{\infty}$ is independent of the diffusion tensor $\M(x)$.

For the particular case of a quadratic interaction potential $F(x) = \frac{x^2}{2}$, which is the case that we will study here, all stationary solutions are given by the one parameter family of Gibbs states of the form~\eqref{e:inv-meas-mckean} and the integral equation~\eqref{eq:invariant_measure_homog_mean_field} reduces to a nonlinear equation, the selfconsistency equation~\cite{shiino1987}
\begin{equation}\label{e:selfconsist}
m = R(m;\theta,\beta) := \frac{1}{Z}\int_\R x e^{-\beta\left(V_0(x)+\theta \left(\frac{x^2}{2}-m x\right)+\psi(x)\right)} \ dx.
\end{equation}
By solving this equation we can construct the full bifurcation diagram of the stationary Fokker-Planck equation. This will be done in Section~\ref{sec:numerics}.

We are also interested in the equation for the critical temperature~\eqref{e:critical_temp}, which in this case is given by
\begin{equation}\label{eq:variance_with_Rm}
\frac{1}{Z}\int_\R x^2 e^{-\beta\left(V_0(x)+\theta \left(\frac{x^2}{2}\right)+\psi(x)\right)}\ dx -\left(\frac{1}{Z}\int_\R x e^{-\beta\left(V_0(x)+\theta \left(\frac{x^2}{2} \right)+\psi(x)\right)} \ dx\right)^2=  \frac{1}{\beta\theta}.
\end{equation}
Assuming that the large scale part of the potential is symmetric, we have that $\int x p_\infty(x;m=0, \beta, \theta) \ dx )  = 0$ and the equation above simplifies to
\begin{equation}\label{eq:variance}
\frac{1}{Z} \int_\R x^2 e^{-\beta\left(V_0(x)+\theta \left(\frac{x^2}{2}\right)+\psi(x)\right)}\ dx=  \frac{1}{\beta\theta}.
\end{equation}

From the definition of $\psi(x)$ in Eqn.~\eqref{eq:small_psi}), we can conclude that for separable potentials, i.e. when $V_1(x,y)$ is independent of $x$, 
then $\psi(x)$ becomes a constant. This, in turn, means that the stationary solutions to the homogenized McKean-Vlasov equation are the same to the ones for the system  without fluctuations ($V_1(x,y)=0$) -- see Corollary~\ref{prop:additive_fluctuations} in Section~\ref{sec:N_first} below.
For example, when the large scale part of the potential $V_0(x)$ is convex, there are no phase transitions for the homogenized dynamics. We will show in Sections~\ref{sec:N_first} and~\ref{sec:numerics} 
that this is not the case if we take the limits in different order.
%
%
%
\section{Multiscale Analysis for the McKean-Vlasov Equation in a Two-scale Potential}\label{sec:N_first}
In this section we consider the homogenization problem for the McKean-Vlasov equation in a locally periodic potential for the case of a quadratic (Curie-Weiss) interaction. In particular, we first pass to the mean field limit (i.e., send $N\rightarrow\infty$) in Eqn.~\eqref{eq:system_of_sdes} with $F(x) = \frac{x^2}{2}$ and study the effects of finite (but small) $\epsilon$ on the bifurcation diagram, before sending $\epsilon\rightarrow 0$.
%
%
\subsection{Mean field limit for interacting diffusions in a two-scale potential: $N\rightarrow\infty$, $\epsilon>0$ finite}
We start with the system of interacting diffusions
\begin{equation}
\label{e:sde-2scale}
dX_t^i = -\partial_x V \left( X_t^i, \frac{X_t^i}{\eps} \right) \,dt - \theta \left( X_t^i - \frac{1}{N}  \sum_{j=1}^{N} X_t^j \right)\,dt + \sqrt{2\beta^{-1}} \, dB_t^i.
\end{equation} 
The notation is the same as in Section~\ref{sec:epsilon_first}, i.e. $V^{\eps}(x):=V \left( x,\frac{x}{\eps} \right)$ is a smooth confining potential that is $L-$periodic in its second argument, $\theta >0$ is the interaction strength, $\beta$ the inverse temperature and $\{ B^i_t, \; i=1, \dots, N \}$ are standard independent one-dimensional Brownian motions.

Taking the limit as $N\rightarrow\infty$, we obtain the McKean-Vlasov-Fokker-Planck equation:
\begin{equation}
\label{eq:FP_for_multiscale}
\frac{\partial p}{\partial t} = \frac{\partial}{\partial x}\left(\beta^{-1}\frac{\partial p}{\partial x} + \partial_x V^\epsilon(x) p + \theta \left(x - \int x p(x,t) \, dx \right) p\right).
\end{equation}
The equilibrium solutions, i.e. stationary states, of this equation are given by a one parameter family of two-scale Gibbs distributions -- see Eqn~\eqref{e:inv-meas-mckean}:
\begin{subequations}\label{e:inv-meas-mckean-two-scale}
\begin{eqnarray}
p^{\eps}_{\infty}(x ; \theta, \beta, m^{\eps}) &=& \frac{1}{Z^{\eps}(\theta, \beta ; m^{\eps} )} e^{- \beta \left( V^{\eps}(x) + \theta \left(\frac{1}{2}x^2 - x m^{\eps} \right) \right)}, 
\\  Z^{\eps}(\theta, \beta ; m^{\eps} ) &=& \int_{\R} e^{- \beta \left( V^{\eps}(x) + \theta \left(\frac{1}{2}x^2 - x m^{\eps} \right) \right)} \, dx.
\end{eqnarray}
\end{subequations}
Our goal now is to study the $\eps \rightarrow 0$ limit of the selfconsistency equation -- see Eqn.~\eqref{e:self-consist}
\begin{equation}\label{e:self-consist-eps}
m^{\eps} = \int_{\R} x p^{\eps}_{\infty}(x ; \theta, \beta, m^{\eps}) \, dx =:  R^{\eps} (m^{\eps}; \theta, \beta),
\end{equation}
and also the equation for the critical temperature,
\begin{equation}\label{e:variance-eps}
\int_{\R} x^2 p^{\eps}_{\infty}(x ; \theta, \beta, m^{\eps}) \, dx =\frac{1}{\beta\theta}.
\end{equation}
	
\begin{prop}\label{prop:commutative}
The limits $\epsilon\rightarrow 0, \, N\rightarrow\infty, \, T\rightarrow\infty$ and $N\rightarrow\infty, \, T\rightarrow\infty, \, \epsilon\rightarrow 0$ do not commute. In particular, the $\eps \rightarrow 0$ limits of the selfconsistency equation~\eqref{e:self-consist-eps} and of the equation for the critical temperature~\eqref{e:variance-eps} are {\bf different} from \eqref{e:selfconsist} and~\eqref{eq:variance}.
\end{prop}
\begin{proof}
The proof of this result follows from properties of periodic functions \cite[Thm. 2.28]{PavlSt08}. Consider $u\in L^2(\R^d;C_{per}(\Y)), \, \epsilon > 0$ and define $u^\epsilon(x,y) = u\left(x,\frac{x}{\epsilon}\right)$. Then 
\begin{equation}\label{e:period}
u^\epsilon\rightharpoonup \int_\Y u(x,y) \ dy \quad \mbox{weakly in} \;\; L^2(\R^d).
\end{equation}

We will use this fact to identify the limits as $\epsilon\rightarrow 0$ of $p^\epsilon_\infty(x;\theta,\beta,m^\epsilon)$ and $Z^\epsilon(m^\epsilon;\theta,\beta)$, in order to obtain the limits of the first and second moments. 
First, we note that both the invariant density $p^\epsilon$ and the first moment $m^\epsilon$ depend on $\epsilon$. For a fixed $\epsilon>0$, it is straightforward to check  that the two-scale potentials verify the conditions presented in \cite[Eqns. (3.1), (3.2)]{Arnold1996}, as long as the nonseparable fluctuations are truncated outside the interval $[-a, a]$ -- this is the case for us; see Table~\ref{tab:pot} in Section~\ref{sec:numerics}. This, by estimates~\cite[Eqn.(3.7), Eqn.(3.8)]{Arnold1996}, implies uniform boundedness of the first moment, $m^\epsilon$, as well as existence of a unique global weak solution for the McKean-Vlasov equation. 
We can therefore extract a converging subsequence that converges to some $m \in \R$. We use the notation $V_{eff}(x;m,\theta) = V_0(x) + \theta \left( \frac{1}{2} x^2 - m x\right)$ with $V(x,y) = V_0(x) +V_1(x,y)$ -- see Eqn.~\eqref{eq:multiscale_potential}. We note that $V_{eff}$ depends smoothly on $m$. We use the convergence of $m^{\eps}$ to $m$ and~\eqref{e:period} to deduce:
\begin{eqnarray}
Z^\epsilon(m^{\eps};\theta,\beta) &=& \int_\R e^{-\beta\left(V_{eff}(x;m^{\eps},\theta)+V_1 \left(x, \frac{x}{\eps} \right)   \right)}\ dx \nonumber \\ &\rightarrow& \int_0^L\int_\R e^{-\beta\left(V_{eff}(x;m,\theta) +  V_1(x,y)\right)} \ dx \ dy 
 =: \bar{Z}(m;\theta,\beta). \label{eq:Zbar}
\end{eqnarray}
Similarly,
\begin{equation}\label{eq:Mlimit}
\int_\R x \  e^{-\beta\left(V_{eff}(x;m,\theta) + V_1\left(x,\frac{x}{\epsilon}\right)\right)}\ dx \rightarrow \int_0^L\int_\R x \ e^{-\beta\left(V_{eff}(x;m,\theta) +  V_1(x,y)\right)} \ dx \ dy.
\end{equation}
Combining \eqref{eq:Zbar} and \eqref{eq:Mlimit}, we obtain
\begin{equation}\label{eq:xiLimit}
m = \frac{1}{\int_0^L\int_\R e^{-\beta\left(V_{eff}(x;m^{\eps},\theta) +  V_1(x,y)\right)} \ dx \ dy} \int_0^L\int_\R x \ e^{-\beta\left(V_{eff}(x;m,\theta) +  V_1(x,y)\right)} \ dx \ dy.
\end{equation}
Arguing in a similar way for the variance, we conclude that
\begin{equation}\label{eq:VarianceLimit}
1 = \frac{\beta\theta}{\bar{Z}(m;\theta,\beta)} \int_0^L\int_\R x^2 \ e^{-\beta\left(V_{eff}(x;m,\theta) +  V_1(x,y)\right)} \ dx \ dy.
\end{equation}

We conclude that equations \eqref{eq:xiLimit} and \eqref{eq:VarianceLimit} are {\bf different}, from \eqref{e:selfconsist} and~\eqref{eq:variance}.
\end{proof}

The two limits, $\epsilon\rightarrow 0, \, N\rightarrow\infty, \, T\rightarrow\infty$ and $N\rightarrow\infty, \, T\rightarrow\infty, \, \epsilon\rightarrow 0$ commute in the case where the fluctuations in the potential are independent of the macroscale $x$, $V_1 =V_1(y)$ in~\eqref{eq:multiscale_potential}. An immediate corollary of the above proposition is the following.

\begin{corollary}\label{prop:additive_fluctuations}
Separable fluctuations do not affect the bifurcation diagram in the mean field limit.
\end{corollary}

\begin{proof}
When the fluctuations are separable (i.e., $V_1(x,y)$ does not depend on $x$), $\psi(x,\beta)$ in \eqref{e:selfconsist}, \eqref{eq:variance} becomes a constant that we can ignore since it also appears in the partition function and they cancel out. Similarly, the terms of the form $\int_\R  e^{-\beta V_1(x,y)} \ dy$ in equations \eqref{eq:xiLimit} and \eqref{eq:VarianceLimit} become constants independent of $x$ and cancel with the corresponding terms in the partition function~\eqref{eq:Zbar}.
\end{proof}

To illustrate the fact that the two limits do commute when the fluctuations are independent of the macroscale, we present in Figures~\ref{fig:compare_epsilon_additive} and~\ref{fig:compare_epsilon_multiplicative} below the plots of $R(m^\epsilon;\theta,\beta)$ for various values of $\epsilon$ and fixed 
$\beta$ and $\theta$, which we compare with the solution of the homogenized selfconsistency equation $R(m;\theta,\beta)=m$. We present results both for a convex and nonconvex confining potential, with periodic fluctuations. More details about the two-scale potentials that we use for the numerical simulations will be given in Section~\ref{sec:numerics}.

\begin{figure}[h!]
\begin{subfigure}{0.5\textwidth}
	\includegraphics[width=6cm,height=6cm]{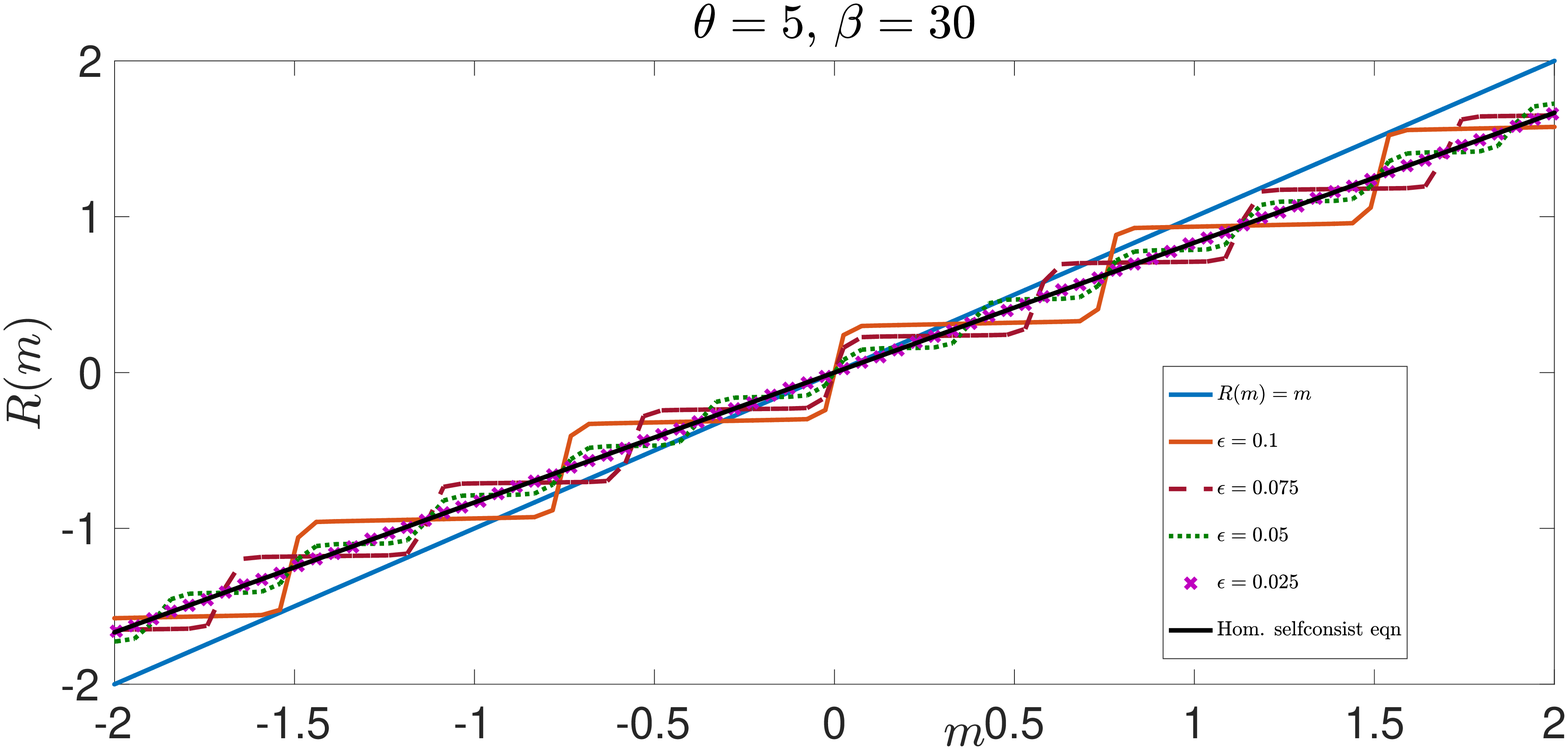}
	\caption{ $V^{\eps}(x) = \frac{x^2}{2} + \delta \cos \left(\frac{x}{\eps} \right)$} \label{fig:convex_add}
\end{subfigure}
\begin{subfigure}{0.5\textwidth}
	\includegraphics[width=6cm,height=6cm]{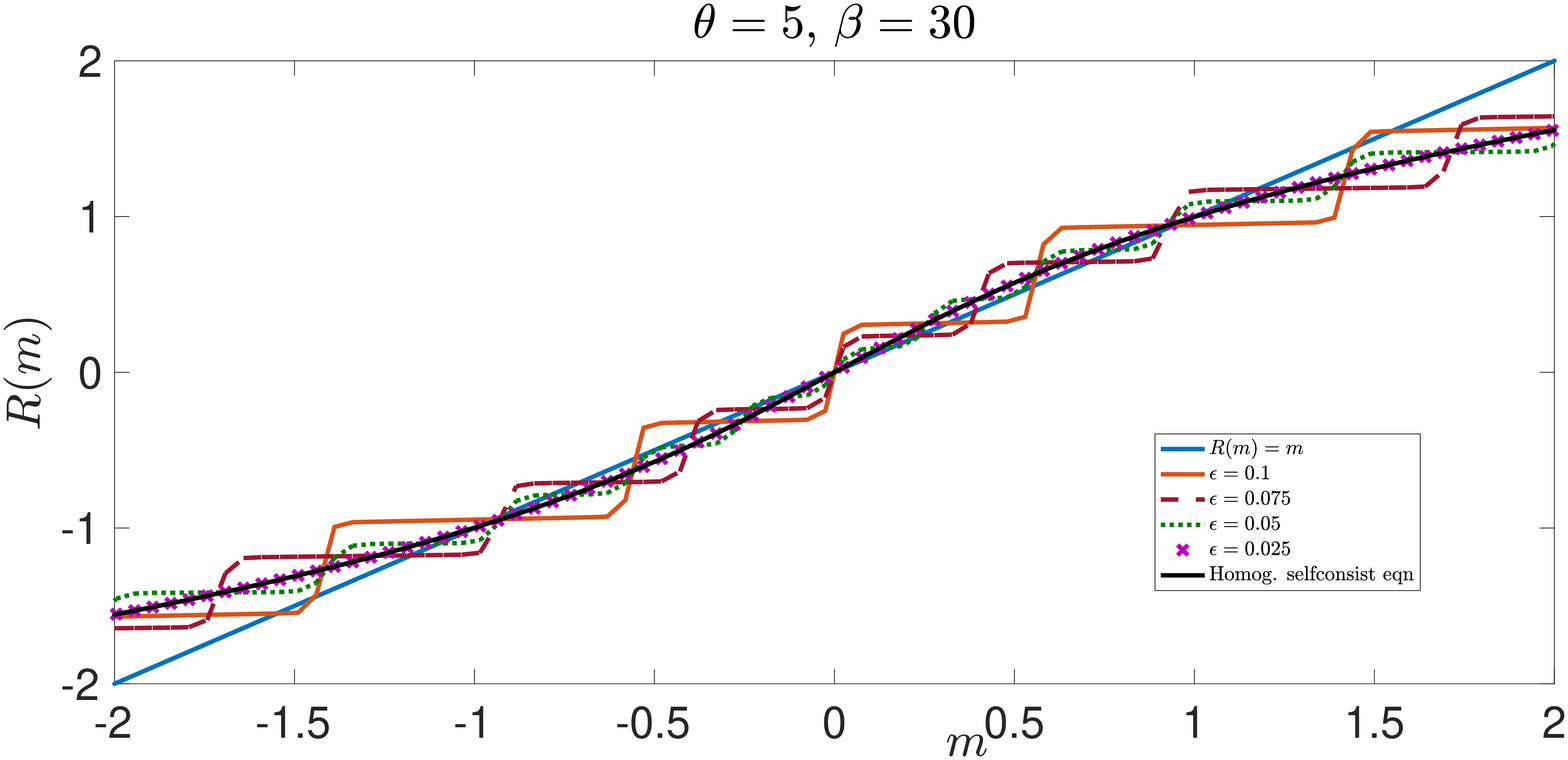}
\caption{$V^{\eps}(x) = \frac{x^4}{4} - \left( 1 + \delta \cos \left(\frac{x}{\eps} \right) \right) \frac{x^2}{2}$}\label{fig:bistable_add}
\end{subfigure}
		 \caption{Plot of $R(m;\theta,\beta) = m$ and $R(m^\epsilon;\theta,\beta)$ for $\theta = 5, \, \beta = 30, \delta = 1$ and various values of $\epsilon$ for separable potentials. \eqref{fig:convex_add} Convex potential $V_0(x)$ and \eqref{fig:bistable_add} Bistable potential $V_0(x)$.}
	 \label{fig:compare_epsilon_additive}
 \end{figure}
As is evident from Figure~\ref{fig:sep_pot}, the oscillatory part of the potential introduces (infinitely many) additional local minima. Consequently,~\cite{Tugaut2014}, the selfconsistency equation $R(m^\epsilon;\theta,\beta)=m^\eps$ has multiple solutions. Furthermore, at shown in Figure~\ref{fig:bistable_add}, in the limit $\epsilon\rightarrow 0$, the curves $R(m^\epsilon;\theta,\beta)$ (various dashed lines) approach those given by $R(m;\theta,\beta)$ computed from Eqn.~\eqref{e:selfconsist} (full black line), in accordance with Corollary~\ref{prop:additive_fluctuations}, showing the commutativity of the two limits.

\begin{figure}[h!]
\begin{subfigure}{0.5\textwidth}
	\includegraphics[width=6cm,height=6cm]{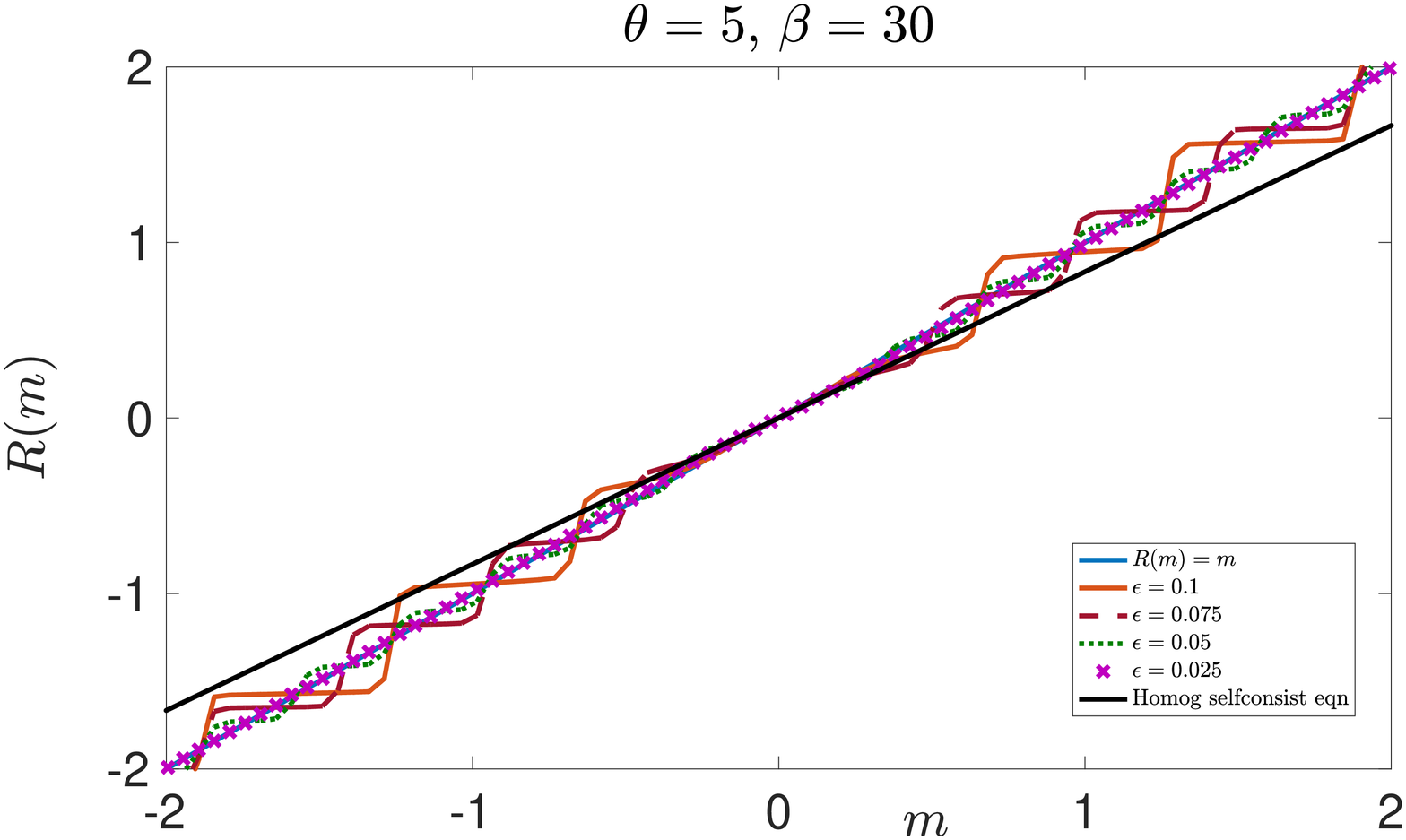}
	\caption{$V^{\eps}(x) = \left(1 + \delta \cos \left(\frac{x}{\eps} \right) \right) \frac{x^2}{2} $} \label{fig:convex_mult}
\end{subfigure}
\begin{subfigure}{0.5\textwidth}
	\includegraphics[width=6cm,height=6cm]{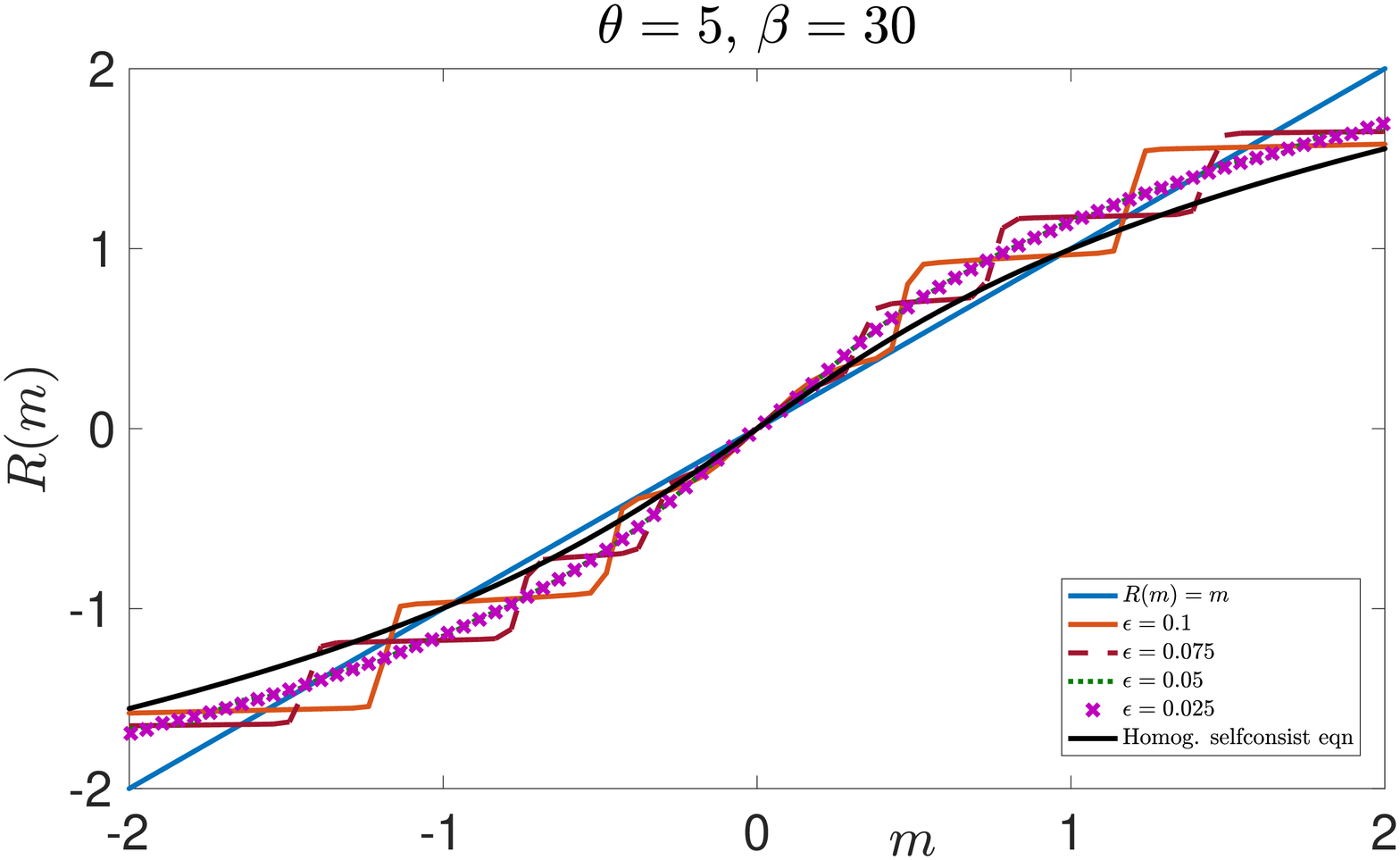}
\caption{$\frac{x^4}{4} - \left( 1 + \delta \cos \left(\frac{x}{\eps} \right) \right) \frac{x^2}{2}$}\label{fig:bistable_mult}
\end{subfigure}		
 \caption{Plot of $R(m;\theta,\beta) = m$ and $R(m^\eps;\theta,\beta)$ for $\theta = 5, \, \beta = 30, \delta = 1$ and various values of $\epsilon$ for nonseparable potentials. \eqref{fig:convex_add} Convex potential $V_0(x)$ and \eqref{fig:bistable_add} Bistable potential $V_0(x)$.}
	 \label{fig:compare_epsilon_multiplicative}
 \end{figure}

Let us consider now the case of nonseparable fluctuations. As we have already discussed, see Figure~\ref{fig:non_sep_pot} and also the inside panels of Figures~\eqref{fig:conv_mult_mean} and \eqref{fig:bistab_mult_mean}, the resulting two-scale potential does not only contain many additional local minima, it is also flattened around $x=0$. In Figure~\ref{fig:compare_epsilon_multiplicative} we present curves $R(m^\epsilon;\theta,\beta)$ for nonseparable fluctuations, 
compared with the line $R(m;\theta, \beta) = m$ (or $y=x$). 
We observe that in the limit $\epsilon\rightarrow 0$ the curves $R(m^\epsilon;\theta,\beta)$ (various dashed lines) {\bf do not} converge to $R(m;\theta,\beta)$ corresponding 
to the homogenized problem (full black line), in accordance with Prop.~\ref{prop:commutative}. 
Notice also the flatness of $R(m^\epsilon;\theta,\beta)$ around $m=0$ for smaller values of $\epsilon$, which follows from the flatness of  the corresponding 
potentials $V^\epsilon$ around $x=0$.
%
%
\subsection{Multiscale analysis for the McKean-Vlasov equation in a two-scale confining potential}

In this section we study the problem of periodic homogenization for the McKean-Vlasov equation in a locally periodic confining potential, for the Curie-Weiss quadratic interaction and in one dimension. We only present formal arguments. The rigorous analysis of this problem will be presented elsewhere.

We consider the nonlinear Fokker-Planck equation~\eqref{eq:FP_for_multiscale} with $F(x) = \frac{x^2}{2}$:
\begin{equation}
\label{eq:FP_particular}
\frac{\partial p^{\eps}}{\partial t} = \beta^{-1}\frac{\partial^2p^\eps}{\partial x^2} + \frac{\partial}{\partial x} \left(V^\prime_0(x)p^{\eps} + V_1^\prime \left(x,\frac{x}{\epsilon}\right) p^{\eps} + \theta(x-m^\epsilon)p^{\eps}\right),
\end{equation}
with initial conditions $p^{\eps}(x,0) = p_{in}(x)$, independent of $\eps$ and where the prime denotes differentiation with respect to $x$. The PDE~\eqref{eq:FP_particular} is coupled to the selfconsistency equation
\begin{equation}
\label{eq:selfconsistency_1}
m^\epsilon (t) = \int_\R x \ p^{\eps}(x,t)\ dx.
\end{equation}
This homogenization problem is (slightly) different from the standard one for the Fokker-Planck equation in a two-scale potential that was studied in~\cite{Pavliotis_al2016, DuncanPavliotis2016} due to the self-consistency equation~\eqref{eq:selfconsistency_1}. In particular, in addition to the standard two-scale expansion for the solution of the Fokker-Planck equation~\eqref{eq:FP_particular}, we also need to expand the solution of~\eqref{eq:selfconsistency_1} into a power series in $\eps$:
\begin{subequations}\label{e:expans-two-scale}
\begin{eqnarray}
\label{eq:p_expansion} p^{\eps}(x,t) &=& p_0 \left(x, \frac{x}{\eps},t \right) + \epsilon p_1\left(x, \frac{x}{\eps},t \right) + \epsilon^2 p_2\left(x, \frac{x}{\eps},t \right) + \dots, \\
\label{eq:xi_expansion} m^\epsilon & = & m_0 + \epsilon m_1 + \epsilon^2m_2+ \dots,
\end{eqnarray}
\end{subequations}
where, as usual~\cite{PavlSt08}, we take $\{ p_j = p_j\left(x ,\cdot, t \right), \; j=0, 1, \dots \}$ to be $L-$periodic in their second argument. Substituting~\eqref{e:expans-two-scale} into~\eqref{eq:FP_particular} and~\eqref{eq:selfconsistency_1} and using the standard tools from the theory of periodic homogenization, e.g. Fredholm's alternative, we obtain the homogenized equation~\eqref{eq:FP_for_homogenized}, satisfied by the marginal of the first term in the two-scale expansion $p(x,t) = \int_0^L p(x,y,t) \, dy$ and with the partial free energy $\psi(x)$ given by~\eqref{eq:small_psi} and with 
$$
m(t):=m_0(t) = \int_{\R} \int_0^L x p_0(x,y,t) \, dx dy.
$$
The convergence of $m^{\eps}(t)$ to $m(t)$ can be justified using the a priori estimates on moments of the solution to the McKean-Vlasov equation that were derived in~\cite{Arnold1996}, in particular~\cite[Eqns (3.1), (3.2)]{Arnold1996}).
 
Alternatively, we can work with the backward Kolmogorov equation: we recall that Eqn.~\eqref{eq:FP_particular} corresponds to the McKean SDE
\begin{equation}
\label{eq:SDE_for_FP_multiscale}
dx_t =-\left[ {V^{\epsilon}}^{\prime} \left(x_t\right) + \theta(x_t-m^\epsilon)\right]\ dt + \sqrt{2\beta^{-1}} dB_t,
\end{equation}
with $V^{\eps}(x) = V \left(x, \frac{x}{\eps} \right)$.
We introduce the auxiliary variable $y_t = \frac{x_t}{\epsilon}$, see, e.g.~\cite{PavlSt06}, and using the chain rule, we can write \eqref{eq:SDE_for_FP_multiscale} as a system of interacting diffusions across scales, driven by the same Brownian motion,
\begin{align}
dx_t & = -\left[ \partial_xV\left(x_t,y_t\right) +\frac{1}{\epsilon}\partial_y V\left(x_t,y_t\right) + \theta(x_t-m^\epsilon)\right]\ dt + \sqrt{2\beta^{-1}} dB_t,\\
dy_t & = -\left[ \frac{1}{\epsilon}\partial_xV \left(x_t,y_t\right) +\frac{1}{\epsilon^2}\partial_yV \left(x_t,y_t\right) + \frac{\theta}{\epsilon}(x_t-m^\epsilon)\right]\ dt + \sqrt{\frac{2\beta^{-1}}{\epsilon^2}} dB_t.
\end{align}
We start by expanding the first moment $m^\eps$ in powers of $\eps$ as in~\eqref{eq:xi_expansion}. The backward Kolmogorov equation for the observable $u^{\eps}(x,y,t) = \E (f(x^{\eps}_t, y^{\eps}_t) | x^{\eps}_0 = x, y^{\eps}_0 = y)$ reads (neglecting terms of $O(\eps)$ that are due to the expansion of $m^\eps$)
\begin{subequations}
\begin{eqnarray}\label{e:back_kolm}
\frac{\partial u^{\eps}}{\partial t} &=& \left(\frac{1}{\epsilon^2}\mathcal{L}_0+\frac{1}{\epsilon}\mathcal{L}_1+\mathcal{L}_2\right ) u^{\eps},
    \\  u^{\eps}(x,y,0) & = & f(x,y),
\end{eqnarray}
\end{subequations}
with 
\begin{eqnarray*}
\cL_0 &=& - \partial_y V \partial_y - \beta^{-1} \partial_y^2, 
\\
 \cL_1 &=& -\left( \partial_x V -\theta (x - m_0) \right) \partial_y -  \partial_y V \partial_x  - 2 \beta^{-1} \partial_x \partial_y, 
\\
 \cL_2 & = & - \left( \partial_x V -\theta (x - m_0) \right) \partial_x - \theta m_1 \partial_y - \beta^{-1} \partial_x^2,
\end{eqnarray*}
We can now proceed with the analysis of~\eqref{e:back_kolm}, first for the choice $f(x) = x$, i.e. the evolution of the first moment, and then for arbitrary observables. We obtain, thus, the homogenized backward Kolmogorov equation, from which we can read off the homogenized McKean SDE and the corresponding Fokker-Planck equation:  
\begin{equation}
\label{eq:FP_for_homogenized-1}\frac{\partial p}{\partial t} =\frac{\partial}{\partial x}\left[ \beta^{-1}\frac{\partial\left(\M(x) p\right)}{\partial x} + \M(x)\Big( V'_0(x) +\psi'(x) + \theta\left( x - m(t)\right)\Big)p  + \beta^{-1}\frac{\partial \M(x)}{\partial x}p\right], 
\end{equation}
where $\psi(x) = -\beta^{-1}\ln \left( \int_0^L e^{-\beta V_1(x,y)} \ dy \right)$ 
and $\M(x)$ is defined in~\eqref{eq:diff_coeff}. For the sake of brevity we will omit the details.
%
%
\section{Numerical Simulations}\label{sec:numerics}
In this section we construct the bifurcation diagram  for the stationary McKean-Vlasov equation (both for finite values of $\eps$ and in the homogenization limit), present the results of Monte Carlo (MC) simulations based on the numerical solution of the particle/SDE approximation, and we also solve the time-dependent McKean-Vlasov PDE. Our goal is to investigate numerically the issue of (lack of) commutativity of the mean field and homogenization limits. We consider interacting diffusions (and the corresponding McKean-Vlasov) in one dimension and we study two types of large-scale and fluctuating parts of the potential. We consider both convex and nonconvex potentials, and both additive (separable) and multiplicative (nonseprarable) fluctuations. The four potentials that we use for our simulations are tabulated in Table~\ref{tab:pot}. 
\begin{table}[h!]
\centering
\begin{tabular}{|c|c|c|}
\hline
Confining potential $V_0(x)$ & Fluctuating potential $V_1(x)$ & Case \\
\hline
\multirow{2}{*}{$V_0^c(x) = \frac{x^2}{2}$} & $V_1^+(x) = \delta\cos\left(\frac{x}{\epsilon}\right)$   & $1$\\
\cline{2-3}
							 & $V_1^\times(x) = \delta\chi_{[-a,a]}(x)\frac{x^2}{2}\cos\left(\frac{x}{\epsilon}\right)$   & $2$\\
\hline
\multirow{2}{*}{$V_0^b(x) = \frac{x^4}{4} - \frac{x^2}{2}$} & $V_1^+(x) = \delta\cos\left(\frac{x}{\epsilon}\right) $  & $3$\\
\cline{2-3}
										 & $V_1^\times(x) = \delta\chi_{[-a,a]}(x)\frac{x^2}{2}\cos\left(\frac{x}{\epsilon}\right)$   & $4$\\
 \hline 
\end{tabular}
\caption{Potentials used for the numerical simulations.\label{tab:pot}}
\end{table}
We remark that the nonseparable fluctuations $V_1^\times(x)$ are truncated outside the interval $[-a,a]$ in order to prevent the oscillations from growing as $|x| \rightarrow +\infty$.\footnote{In Table~\ref{tab:pot} we denote by $\chi_A$ the characteristic function of the set $A$.} We note that this is necessary for the proof of the homogenization theorem in~\cite{DuncanPavliotis2016} and that, furthermore, it ensures that the a priori estimates on the moments from~\cite{Arnold1996} hold.\footnote{The moment bounds in~\cite{Arnold1996} were obtained for confining potentials with no oscillatory terms. However, it can be checked that they are also valid for the class of fluctuating potentials that we consider in this work, and that they provide us with bounds on the moments that uniform in $\eps$.}

Throughout this section we consider fluctuations which have period $L=2\pi$. 
In all cases, we will consider the Curie-Weiss interaction potential $F(x) = \frac{x^2}{2}$ and throughout Sections~\ref{sec:homog_first} and~\ref{sec:mean_field_first} we will fix the interaction strength to be $\theta = 5$. We choose this value because larger values of $\theta$ allow for bifurcations to occur at higher temperatures, i.e. lower $\beta$, which is easier to handle numerically. In fact, the relevant bifurcation parameter for our problem is given by the combination $\beta\theta$, see Eqn.~\eqref{e:critical_temp}. Fixing $\theta$ allows us to construct the bifurcation diagram by varying only the temperature.  It is also clear from Eqn.~\eqref{e:critical_temp} that this equation has no solutions for negative values of $\theta$, i.e. that no (pitchfork) bifurcations can occur for $\theta < 0$.

Using Eqn.~\eqref{eq:Diffusion_coeff}, we note that the diffusion coefficient for separable fluctuations in the potential is independent of $x$ and is given by
\begin{equation}
\label{eq:Diffusion_coeff_additive}
\M^{+}(x) = \frac{1}{\left(\frac{1}{ 2 \pi}\int_0^{2\pi} e^{-\beta V_1^{+}(x,y)} \ dy\right)\left(\frac{1}{ 2 \pi}\int_0^{2\pi} e^{\beta V_1^{+}(x,z)} \ dz\right)} = \frac{1}{I_0(\beta)I_0(-\beta)},
\end{equation}
where $I_0(\cdot)$ is the modified Bessel function of the first kind~\cite{Pavliotis_al2016}. On the other hand, for nonseparable fluctuations (cases $2$ and $4$ in Table~\ref{tab:pot}) we obtain
\begin{equation}
\label{eq:Diffusion_coeff_multiplicative}
\M^{\times}(x) = \frac{1}{\left(\frac{1}{ 2 \pi}\int_0^{2\pi} e^{-\beta V_1^{\times}(x,y)} \ dy\right)\left(\frac{1}{ 2 \pi}\int_0^{2\pi} e^{\beta V_1^{\times}(x,z)} \ dz\right)} = \frac{1}{I_0\left(\beta\frac{x^2}{2}\right)I_0\left(-\beta\frac{x^2}{2}\right)}.
\end{equation}
Furthermore, we obtain the following formulas for the partition functions
\begin{equation}
\label{eq:1d-Z}
\Z^{+}(x) = e^{-\beta\left(V_0(x)+\theta\left(\frac{x^2}{2}-m x\right)\right)}I_0(\beta), \quad \Z^{\times}(x) = e^{-\beta\left(V_0(x)+\theta\left(\frac{x^2}{2}-m x\right)\right)}I_0\left(\beta\frac{x^2}{2}\right).
\end{equation}
We can now solve the selfconsistency equation~\eqref{e:self-consist} and the equation for the critical temperature~\eqref{e:critical_temp} for the various potentials given in Table~\ref{tab:pot}. We will track each branch of the bifurcation diagram using arclength continuation, which will enable us to plot the first moment $m$ as a function of the inverse temperature $\beta$ for a fixed value of the interaction strength $\theta$. We do this using the Moore-Penrose quasi arclength continuation algorithm.\footnote{Rigorous mathematical construction of the arclength continuation methodology can be found, e.g., in~\cite{Krauskopf} and~\cite{ContAllgower}. Some useful practical aspects of implementing arclength continuation are also given in~\cite{matcont}. 
We use \textsc{Matlab}'s toolboxes to compute the integrals in~\eqref{e:selfconsist} and~\eqref{eq:variance_with_Rm} and thus need to solve
\[
\mathbf{F}([p,m]) = \left[
\begin{array}{c}
p - p_\infty(x;m,\beta,\theta)\\
m-R(m;\theta,\beta) \end{array}
\right]
= 0, \quad \textrm{ and } \quad
\mathbf{G}(\beta) = \beta - \frac{1}{\theta\int x^2p_\infty(x;0,\beta,\theta) \ dx} = 0,
\]
where $p_\infty(x;m,\beta,\theta)$ is a stationary solution of the McKean-Vlasov equation. We start the algorithm at a sufficiently large $\beta_0$, i.e. at a sufficiently low temperature for which we have a good initial guess for the value of the order parameter.} The stability of each branch was determined in two different ways: first, we checked whether it corresponded to a local minimum or maximum of the confining potential. 
Second, we solved the time-dependent McKean-Vlasov equation -- see details in Section~\ref{sec:evolution}--using a perturbation of the steady state belonging to each branch (for a particular value of $\beta$ and $\theta$) as initial condition.
Finally, we have confirmed the stability of each branch by computing the free energy~\eqref{e:free-energy} of a steady state from that branch at a particular value of $\beta$, chosen so that all the branches plotted were present.
Stable branches, plotted in blue in all the figures presented in this section, correspond to local minimizers of the free energy functional; unstable branches, plotted in red, correspond to local maxima of the free energy.

%
%
\subsection{Mean field limit of the homogenized system of SDEs - The $\epsilon\rightarrow 0, \, N\rightarrow\infty$ limit}\label{sec:homog_first}
As discussed before (see discussion of Corollary~\ref{prop:additive_fluctuations}), when the fluctuations are separable the partial free energy $\psi(x)$ defined in Eqn.~\eqref{eq:small_psi} drops out from the homogenized stationary Fokker-Planck equation. This implies, in particular, that the invariant measure(s) of the homogenized dynamics is(are) independent of the fluctuating part of the potential. In particular, there are still no phase transitions when the large-scale part of the potential is convex and still only one pitchfork bifurcation for the bistable potential case -- see Figure~\ref{fig:homog_bistable} -- where two new, stable, branches emerge from the zero mean solution. We note that in this case the homogenized confining potential in the homogenized equation depends on the inverse temperature $\beta$; see the inside panels in Figure~\ref{fig:homog_bistable}. In particular, the values of the 
local minima of the effective potential are  shifted, although their location remains the same, and there are no changes in the topology of the bifurcation diagrams.

For nonseparable fluctuations, the mean field and homogenization limits do not commute (see Prop.~\ref{prop:commutative}). In fact, the homogenization procedure can convexify the effective potential, and we still see no bifurcations when the large-scale part of the potential is convex, while for the bistable potential there is still only one phase transition. The effect of fluctuations on the bifurcation diagram is visible by a shift of the critical temperature at which the phase transition occurs.

Since there are no phase transitions for the convex potential (cases $1$ and $2$ in Table~\ref{tab:pot}) we do not present numerical results for this case. We present in Figure~\ref{fig:homog_bistable} the plots of $R(m;\theta,\beta)$ and the bifurcation diagrams for the bistable potential with separable and nonseparable fluctuations (cases $3$ and $4$, respectively). We observe that, for nonseparable fluctuations, the function $R(m;\theta,\beta)$ is flat around $m=0$;  see Figure~\ref{fig:bistable_mult_homog}. As we have already mentioned, the topology of the bifurcation diagram does  not change, in comparison to that of the bistable potential $V_0^b(x)$ with no fluctuations (see Figure~\ref{fig:bif_bistable_intro} 
for this case); thus, the effect of fluctuations is only observed by a shift in the critical temperature.

\begin{figure}[h!]
\begin{subfigure}{0.32\textwidth}
\includegraphics[width=\textwidth]{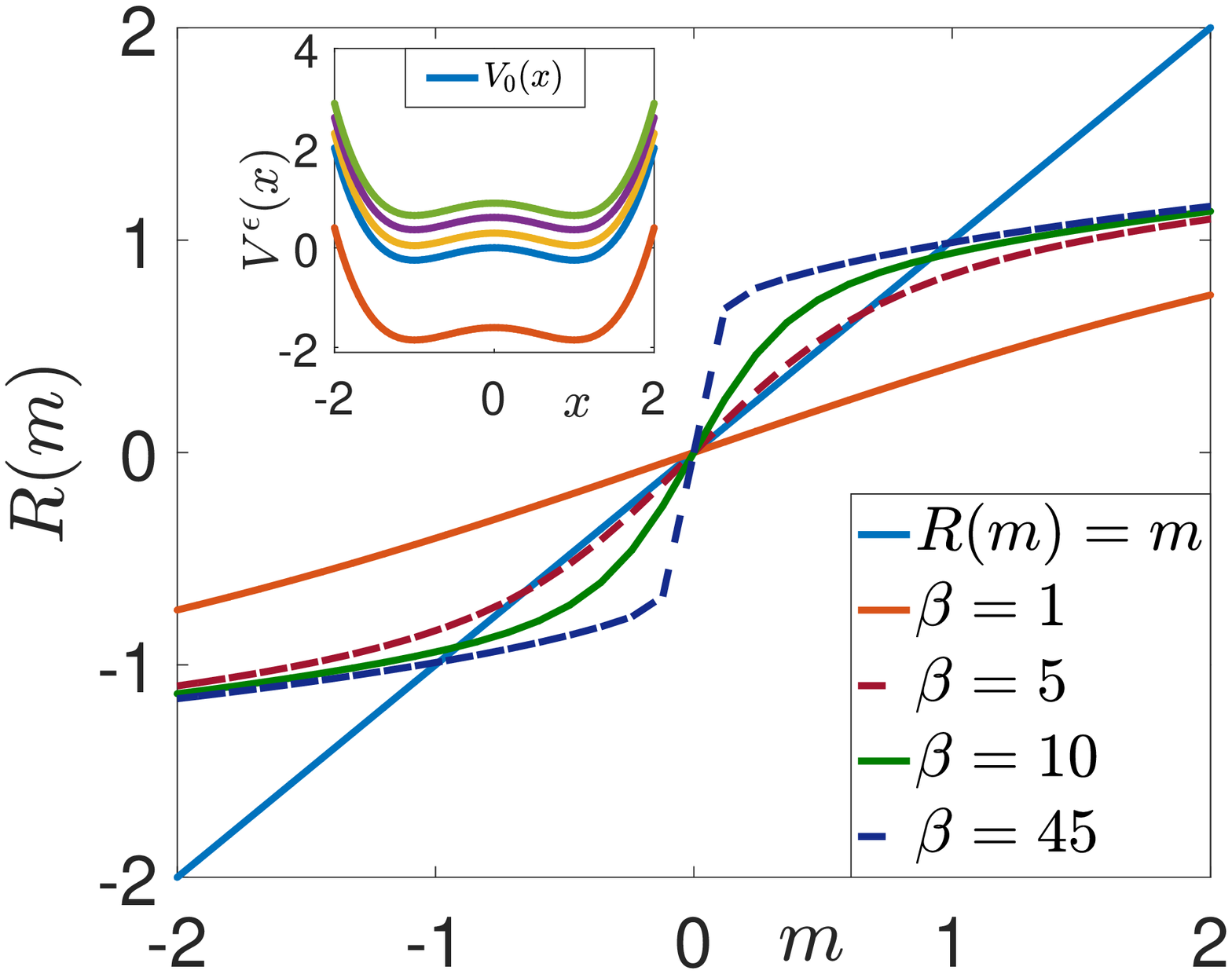}
\caption{Separable }\label{fig:bistable_add_homog}
\end{subfigure}
\begin{subfigure}{0.32\textwidth}
\includegraphics[width=\textwidth]{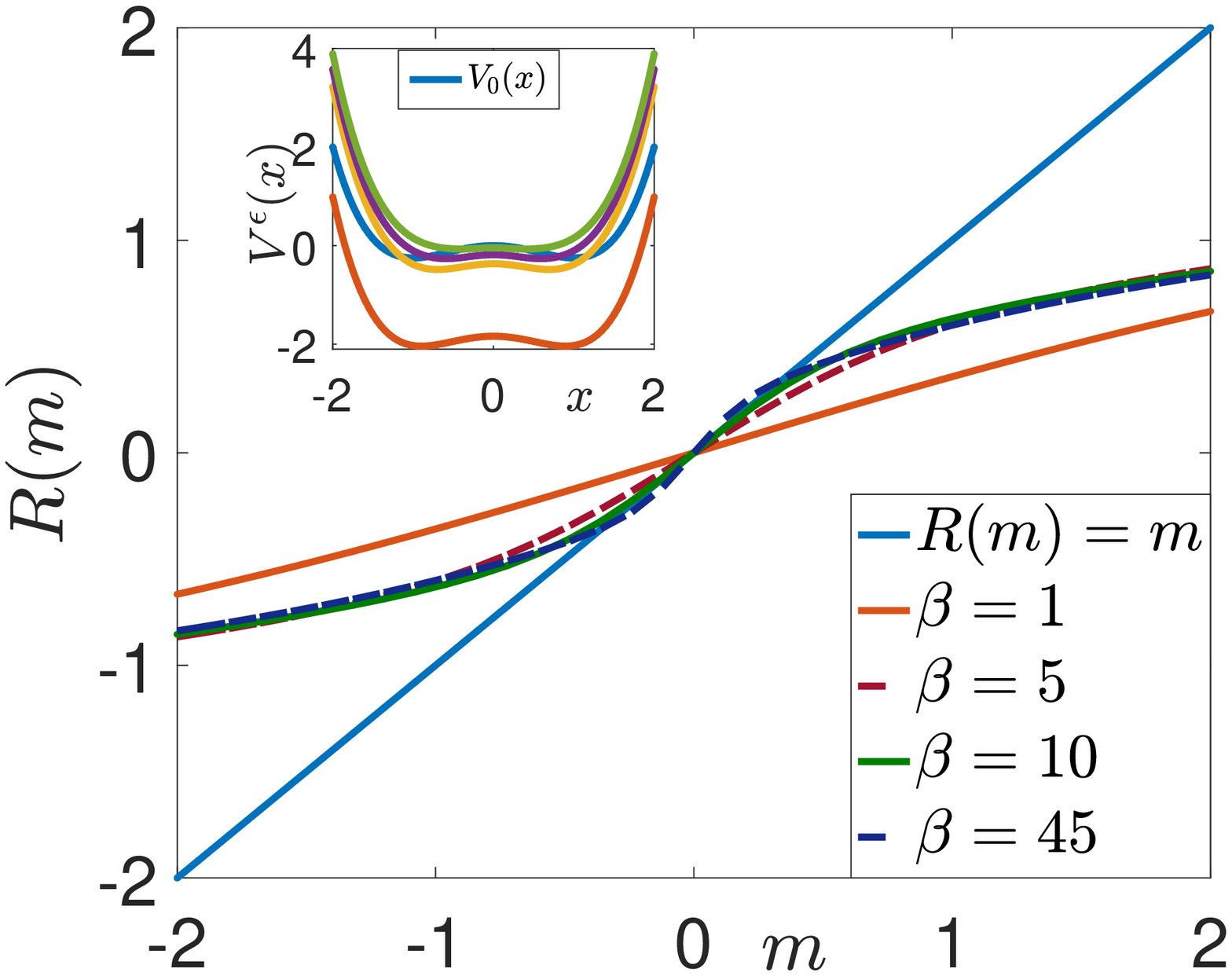}
\caption{Nonseparable}\label{fig:bistable_mult_homog}
\end{subfigure}
\begin{subfigure}{0.32\textwidth}
\includegraphics[width=\textwidth]{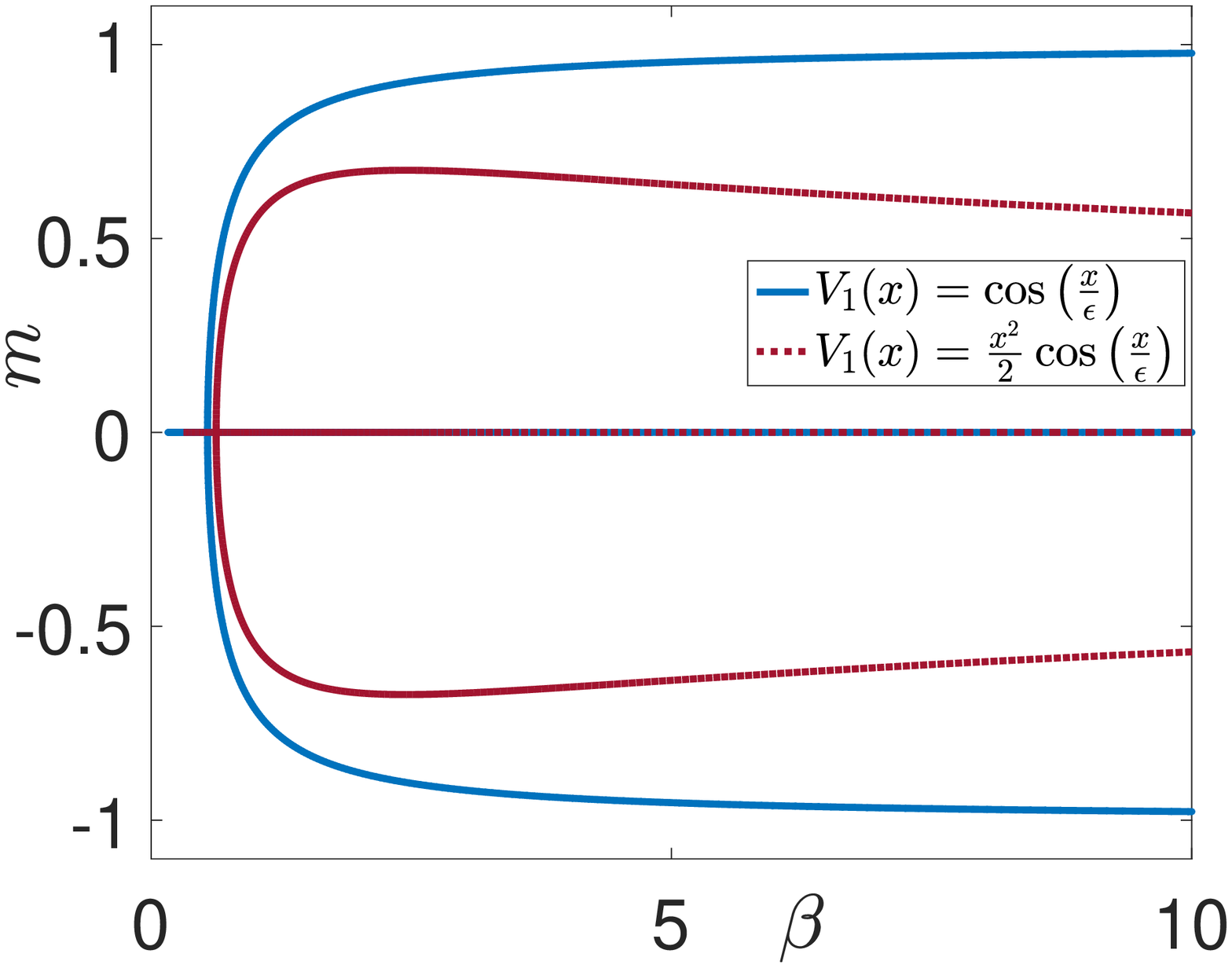}
\caption{Bifurcation diagram}\label{fig:bistable_bif_homog}
\end{subfigure}	
	\caption{Plot of $R(m;\theta,\beta)$ compared to the diagonal $y=x$ ( $R(m;\theta,\beta)=m$) for $\theta = 5,   \delta = 1, \, a = 5$ and various values of $\beta$ for the homogenized  bistable potentials with~\eqref{fig:bistable_add_homog} separable fluctuations (potentials for various values of $\beta$ shown on the inside panel), and  \eqref{fig:bistable_mult_homog}  nonseparable fluctuations (potentials for various values of $\beta$ shown on the inside panel).  \eqref{fig:bistable_bif_homog}  Bifurcation diagram of $m$ as a function of $\beta$ for the potentials in \eqref{fig:bistable_add_homog} (full line) and \eqref{fig:bistable_mult_homog} (dashed line).}
	 \label{fig:homog_bistable}
 \end{figure}

%
%
\subsection{Mean field limit of the multiscale system of SDEs: effects of finite $\epsilon$}\label{sec:mean_field_first}
In this section we present numerical results on the bifurcation diagram when we first pass to the mean field limit, while keeping $\eps$ small but finite. We are particularly interested in the finite $\eps$ effects on the bifurcation diagrams for the two-scale potentials presented in Table~\ref{tab:pot}.

\subsubsection{Convex confining potential with separable and nonseparable fluctuations}
We first consider Case $1$ in Table~\ref{tab:pot}: a convex large-scale potential with separable fluctuations. We present in Figure~\ref{fig:convex_potential_additive} the solution to the selfconsistency equation $R(m;\theta,\beta)=m$, the two-scale potential, and the bifurcation diagram for this case.
\begin{figure}[h!]
\begin{subfigure}{0.5\textwidth}
\includegraphics[width=\linewidth]{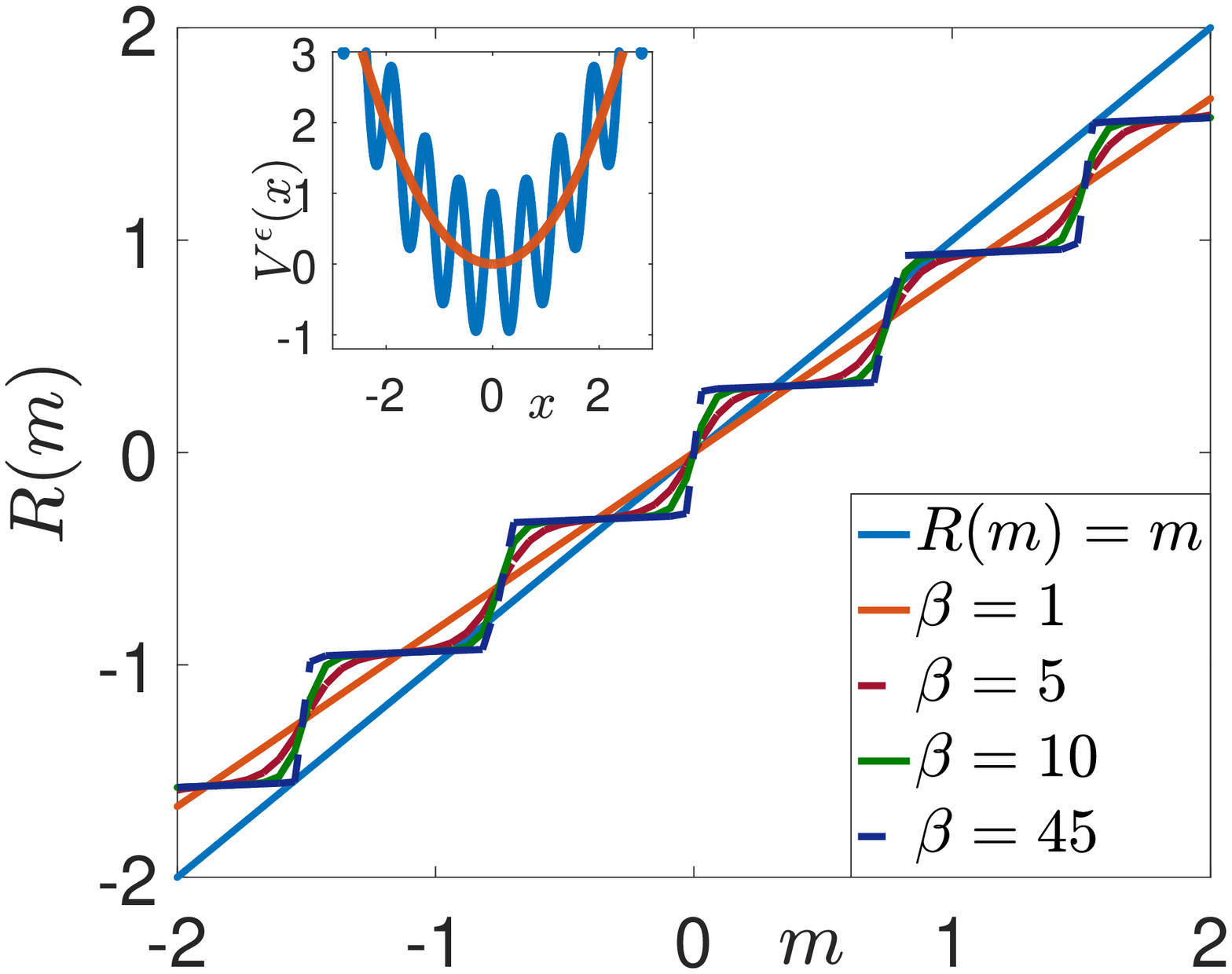}
\caption{$R(m;\theta,\beta)$}\label{fig:conv_add_mean}
\end{subfigure}
\begin{subfigure}{0.5\textwidth}
\includegraphics[width=\linewidth]{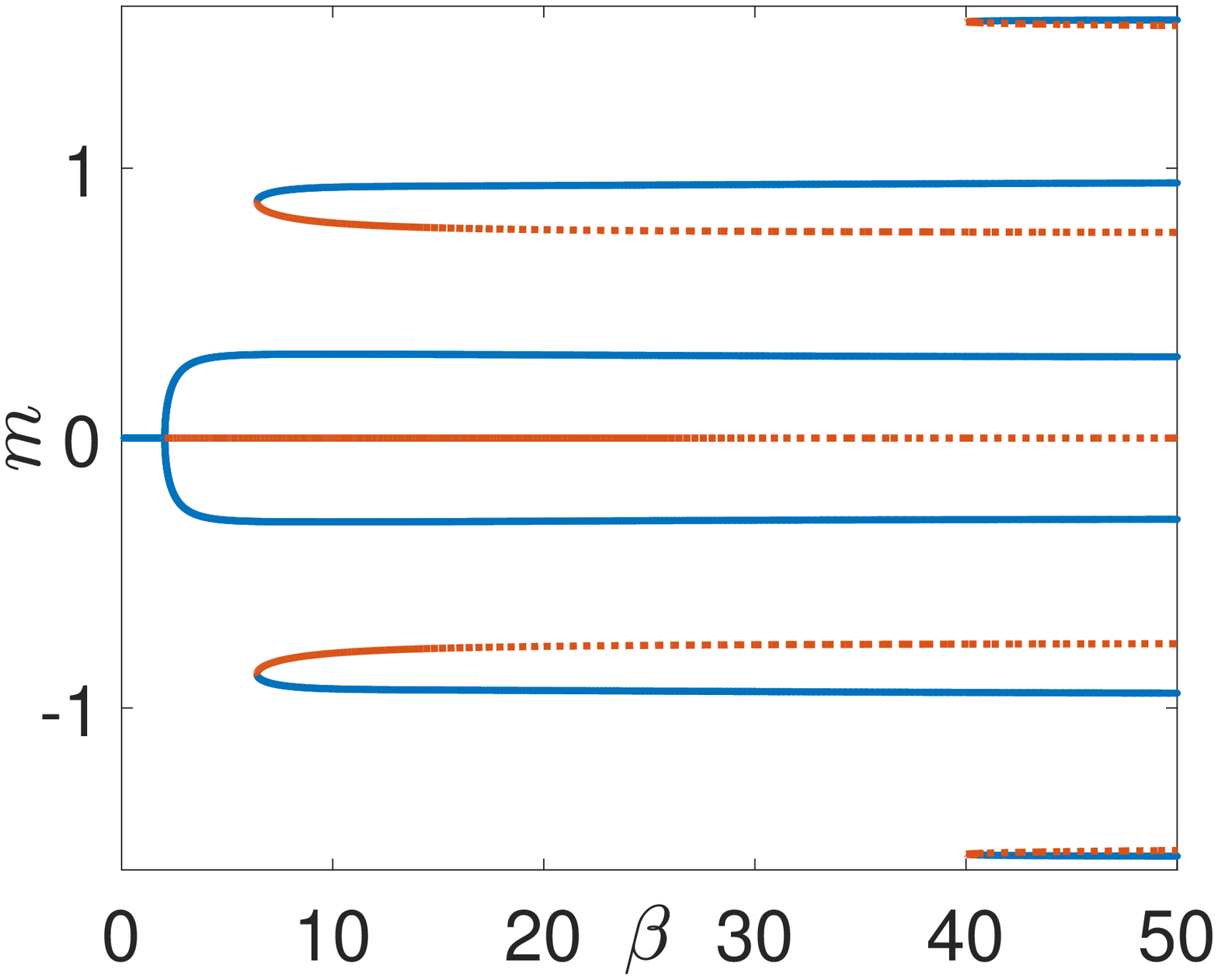}
\caption{Bifurcation Diagram}\label{fig:conv_add_bif}
\end{subfigure}
	 \caption{Results for case $1$: convex $V_0^c$ with separable fluctuations, for $\theta = 5, \, \delta = 1, \, \epsilon = 0.1$. \eqref{fig:conv_add_mean} $R(m^\epsilon;\theta,\beta)$ for various values of $\beta$, with the potential $V^\epsilon(x)$ (full line) compared with $V_0^c(x) $ (dashed line) in the inside panel. \eqref{fig:conv_add_bif} Bifurcation diagram of $m$ as a function of $\beta$. Full lines correspond to stable solutions, while dashed lines represent unstable ones.}
	 \label{fig:convex_potential_additive}
 \end{figure}
For all finite values of $\epsilon$, the resulting potential is nonconvex. This results in the selfconsistency equation having multiple solutions (in fact, as $\epsilon\rightarrow 0$, there are infinitely many solutions). In addition to the emerging pitchfork bifurcation (second order, or continuous, phase transition), we observe the emergence of discontinuous branches that correspond to metastable states, since they are not (global) minimizers of the free energy; see the results presented in Table~\ref{tab:FreeEnergy}. 

Next, we consider the second case in Table~\ref{tab:pot}: a convex large scale potential $V_0^c(x)$ with nonseparable fluctuations. Similarly, we present in Figure~\ref{fig:convex_potential_multiplicative} the solution to the selfconsistency equation $R(m^\epsilon;\theta,\beta)=m^\epsilon$, the two-scale potential, and the bifurcation diagram.
We note that, as we mentioned before, we restrict the nonsparable fluctuations to a finite interval. In our computations we use $a=5$, in the characteristic function in Table~\ref{tab:pot}.

\begin{figure}[h!]
\hskip-0.75cm
\begin{subfigure}{0.5\textwidth}
\includegraphics[width=0.95\linewidth]{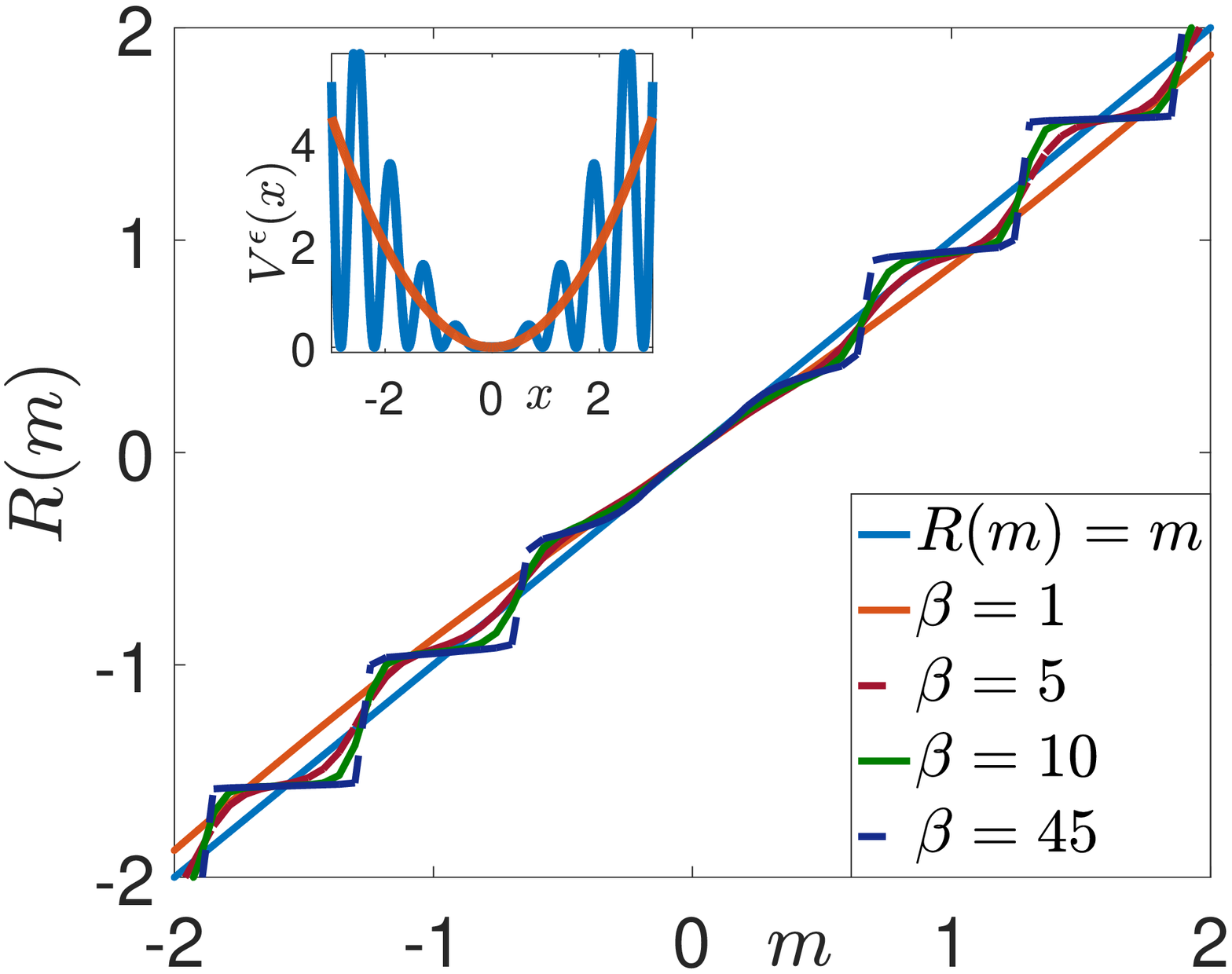}
\caption{$R(m;\theta,\beta)$}\label{fig:conv_mult_mean}
\end{subfigure}
\begin{subfigure}{0.5\textwidth}
\includegraphics[width=0.95\linewidth]{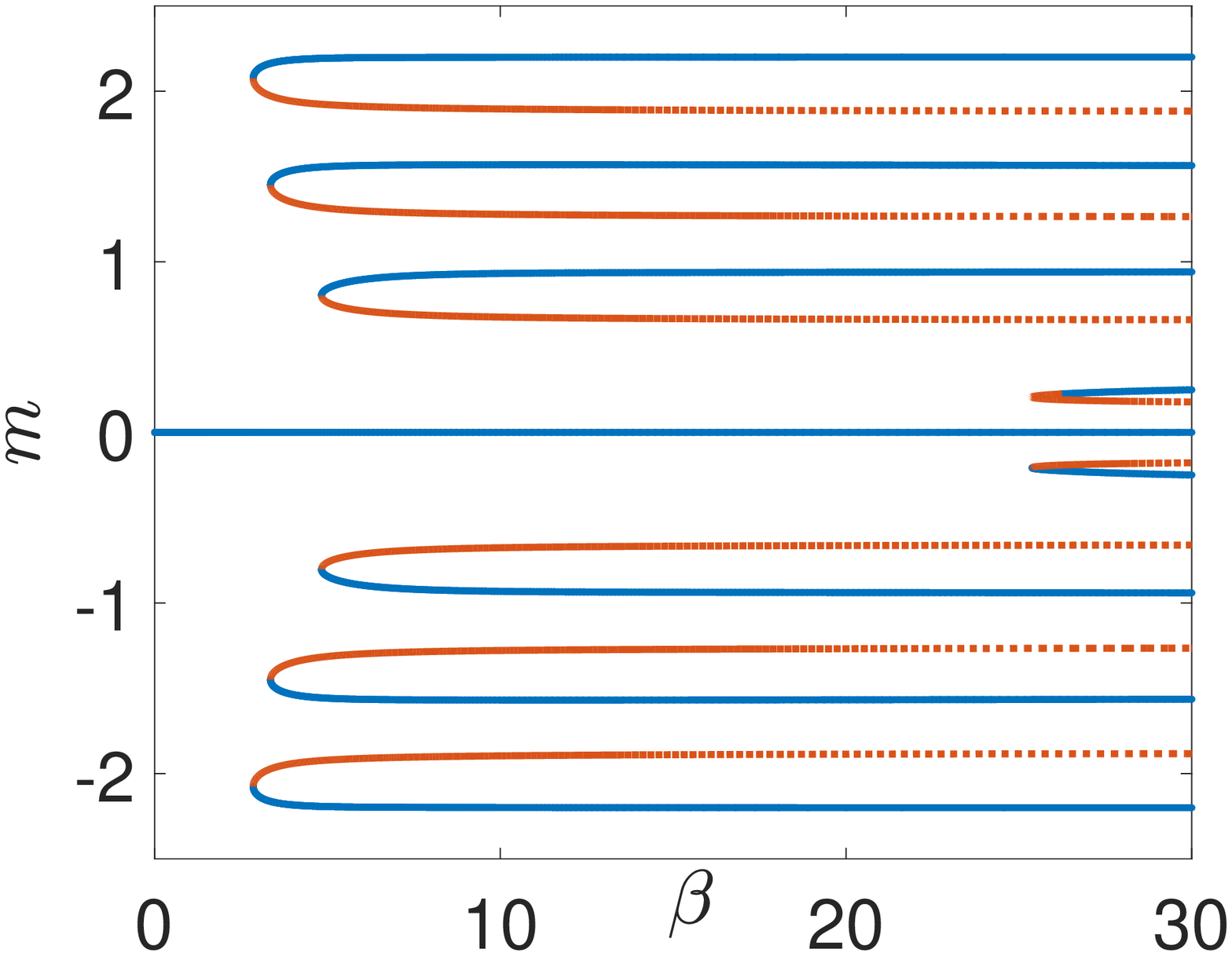}
\caption{Bifurcation diagram}\label{fig:conv_mult_bif}
\end{subfigure}

\begin{subfigure}{0.5\textwidth}
\centering
\begin{tabular}{|c|c|}
\hline
$\beta$ & $45$ \\
\hline
\multirow{9}{*}{\rotatebox[origin=c]{90}{Free Energy}} & $0.1441$ \\
							     & $0.3684$ \\
							     & $0.1433$ \\
							     & $0.3184$ \\
							     & $0.0976$ \\
							     & $0.2425$ \\
							     & $0.0625$ \\
							     & $0.0630$ \\
							     & $0.0586$ \\
\hline 
\end{tabular}
\caption{Free energy values at $\beta = 45$}\label{fig:conv_mult_FE}
\end{subfigure}
\hskip-0.85cm
\begin{subfigure}{0.5\textwidth}
\includegraphics[width=5.0cm,height=5.0cm]{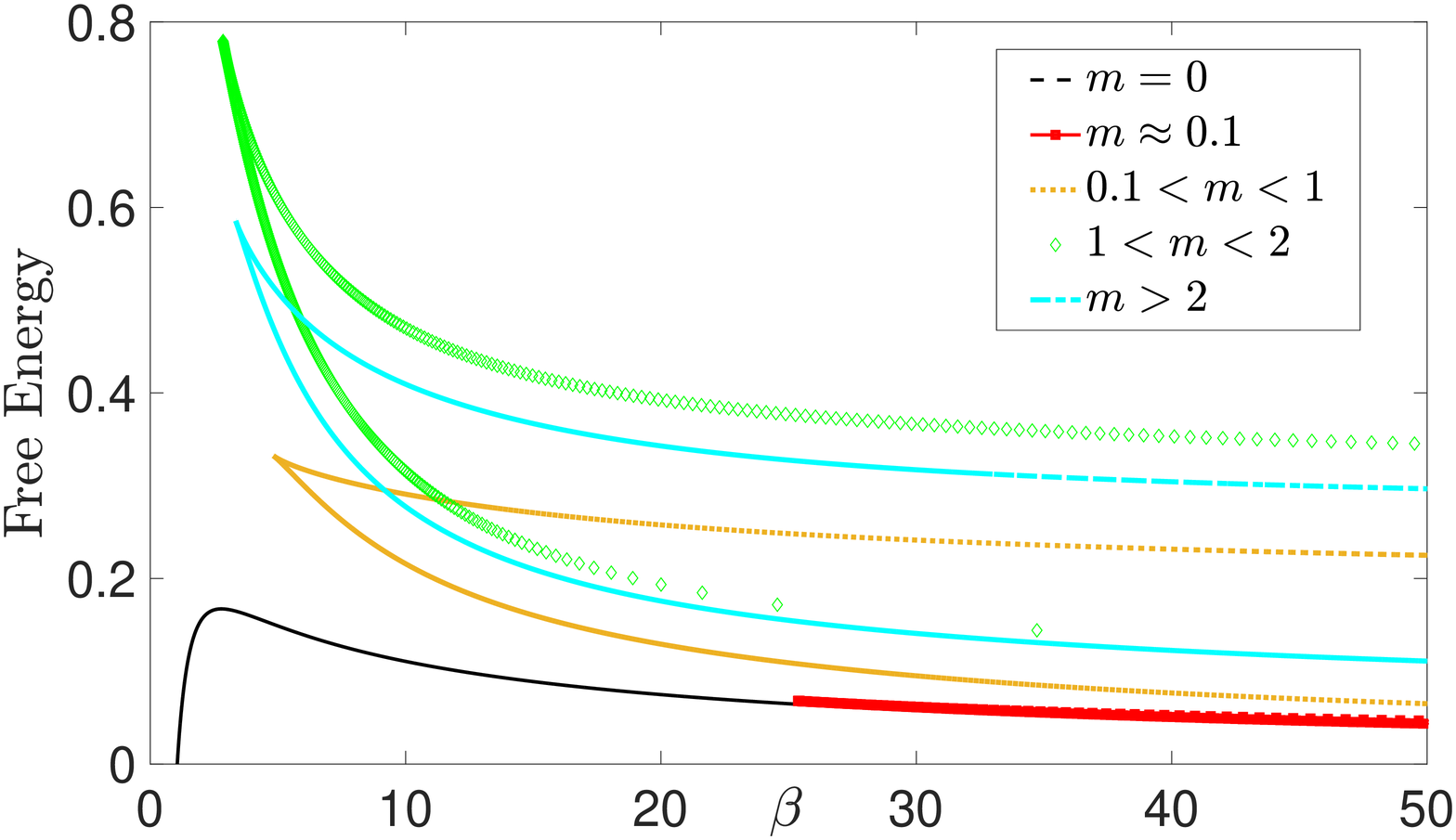}
\caption{Free energy of each branch}\label{fig:conv_mult_FEf}
\end{subfigure}
\caption{Results for case $2$: convex $V_0$ with nonseparable fluctuations, for $\theta = 5, \, \delta = 1, \, \epsilon = 0.1$. \eqref{fig:conv_mult_mean} $R(m;\theta,\beta)$ for various values of $\beta$, with the potential $V^\epsilon(x)$ (full line) compared with $V_0^c(x) $ (dashed line) in the inside panel. \eqref{fig:conv_mult_bif} Bifurcation diagram of $m$ as a function of $\beta$. Full lines correspond to stable solutions, while dashed lines represent unstable ones. \eqref{fig:conv_mult_FE} Values of the freeenergy of the steady state in each branch of~\eqref{fig:conv_mult_bif} for $\beta = 45$. \eqref{fig:conv_mult_FEf} Free energy of each branch of the bifurcation diagram.}
\label{fig:convex_potential_multiplicative}
\end{figure}

We observe in Figure~\ref{fig:conv_mult_bif} that no pitchfork bifurcations appear; all new branches that appear do not emerge continuously from the mean zero solution. This is due to the flatness observed in the potential around $m=0$ (see Figure~\ref{fig:conv_mult_mean}). Furthermore, the mean zero solution remains the global minimizer of the free energy for all values of $\beta$. This is tabulated in Table~\ref{fig:conv_mult_FE}. The free energies of the different branches are presented in Figure~\ref{fig:conv_mult_FEf}. These new branches correspond to metastable states.

We have checked the stability of each branch by computing the free energy~\eqref{e:free-energy} of a steady state from that branch at a particular value of $\beta$, chosen so that all the branches plotted were present.  We summarize the results in Table 2. Since we only consider symmetric potentials, it is sufficient to calculate the free energy for the branches with, say, nonnegative values of $m$. 
In each column of Table~\ref{tab:FreeEnergy}, the values of the free energy are presented from the branch with largest value of $m$ to the lowest; 
the last value presented in each column corresponds to the branch with $m=0$. We summarize the results in Table~\ref{tab:FreeEnergy}.

\begin{table}[h!]
\centering
\begin{tabular}{|c|c|c|c|c|}
\hline
Figure & $6$ & $7$ & $8$ & $9$ \\
\hline
$\beta$ & $45$ & $29$ & $20$ & $8$\\
\hline
\multirow{9}{*}{Free Energy} & $\ 0.3080$ & $0.1441$ & $-0.5827$ & $-1.7409$ \\
& $\ 0.3066$ & $0.3684$ & $-0.5674$ & $-0.9933$ \\
& $-0.4600$ & $0.1433$ & $-1.0918$ & $-0.8241$ \\
& $-0.3908$ & $0.3184$ & $-0.7727$ & $\ 0.0856$ \\
& $-0.8593$ & $0.0976$ & $-0.8868$ & \\
& $-0.6514$ & $0.2425$ & $-0.6903$ & \\
& & $0.0625$ & & \\
& & $0.0630$ & & \\
& & $0.0586$ & & \\
\hline 
\end{tabular}
\caption{Free energy of a steady state in each branch of Figures~\ref{fig:convex_potential_additive}-\ref{fig:potential_multiplicative} for fixed values of $\beta$.\label{tab:FreeEnergy}}
\end{table}

We observe that the branch corresponding to a pitchfork bifurcation (i.e., second order phase transition), when present, has the lowest value of the free energy, 
i.e., it is the globally stable one. 
Furthermore, when a pitchfork bifurcation does not occur--see Figure~\ref{fig:convex_potential_multiplicative}--the branch corresponding to $m=0$ is the one with the lowest value of the free energy.
Finally, we observe that the stability of the branches in Figure~\ref{fig:bistable_mult_bif} does not alternate in the same manner as in the previous figures.  This is due to the flatteness of the potential around $x=0$ for nonseparable oscillations.

The results on the stability of the different branches that are reported in this section are preliminary. A more thorough study of the local (linear) and global stability of the stationary states of the McKean-Vlasov dynamics in multiwell potentials will be presented elsewhere. We mention in passing the early rigorous work on the global stability of the steady states for the McKean-Vlasov equation in~\cite{Tamura1987} as well as the careful study of the connection between the loss of linear stability of the uniform state and phase transitions for the McKean-Vlasov equation on the torus (without a confining potential) and with finite range interactions  in~\cite{ChayesPanferov2010}.
%
%
\subsubsection{Bistable confining potential with separable and nonseparable fluctuations}
Here we consider Cases $3$ and $4$ in Table~\ref{tab:pot}, the bistable potential $V_0^b(x)$. In this case, the large-scale potential exhibits a second order phase 
transition even in the absence of small-scale fluctuations (see the pitchfork bifurcation in Figure~\ref{fig:bif_bistable_intro}) due to the existence of two local minima for $V_0^b(x)$. We are interested in analyzing the topological changes that rapid oscillations in the potential induce to the bifurcation diagram.

We start with separable potentials--see Figure~\ref{fig:potential_additive}. We observe that the selfconsistency equation $R(m^\epsilon;\theta,\beta)=m^\epsilon$ 
exhibits a larger number of solutions for finite $\epsilon$, which, as for the convex case, result in the emergence of metastable states that are not continuously connected with the mean zero Gibbs state.

\begin{figure}[ht]
\begin{subfigure}{0.5\textwidth}
\includegraphics[width=\linewidth]{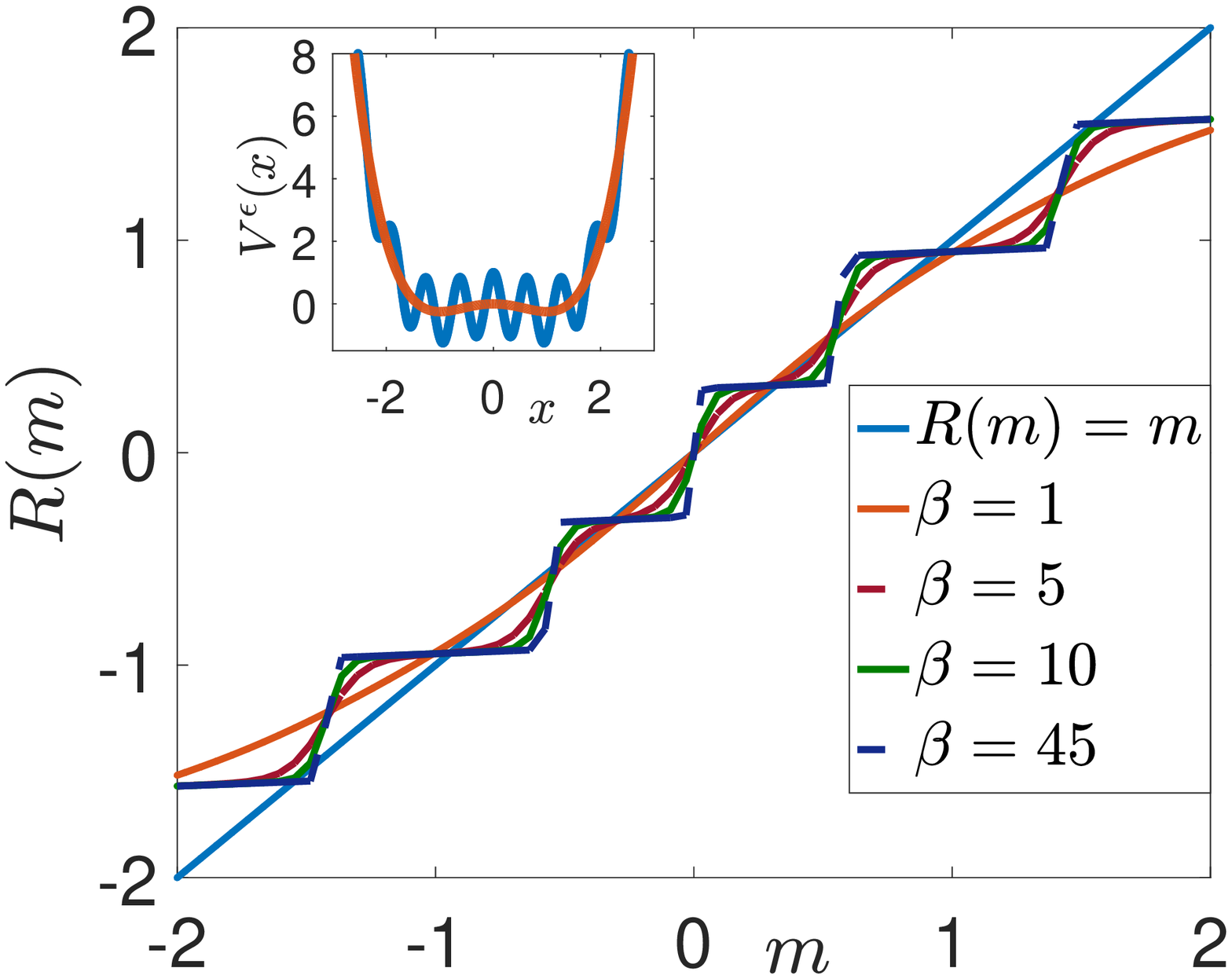}
\caption{$R(m;\theta,\beta)$}\label{fig:bistable_add_mean}
\end{subfigure}
\begin{subfigure}{0.5\textwidth}
\includegraphics[width=\linewidth]{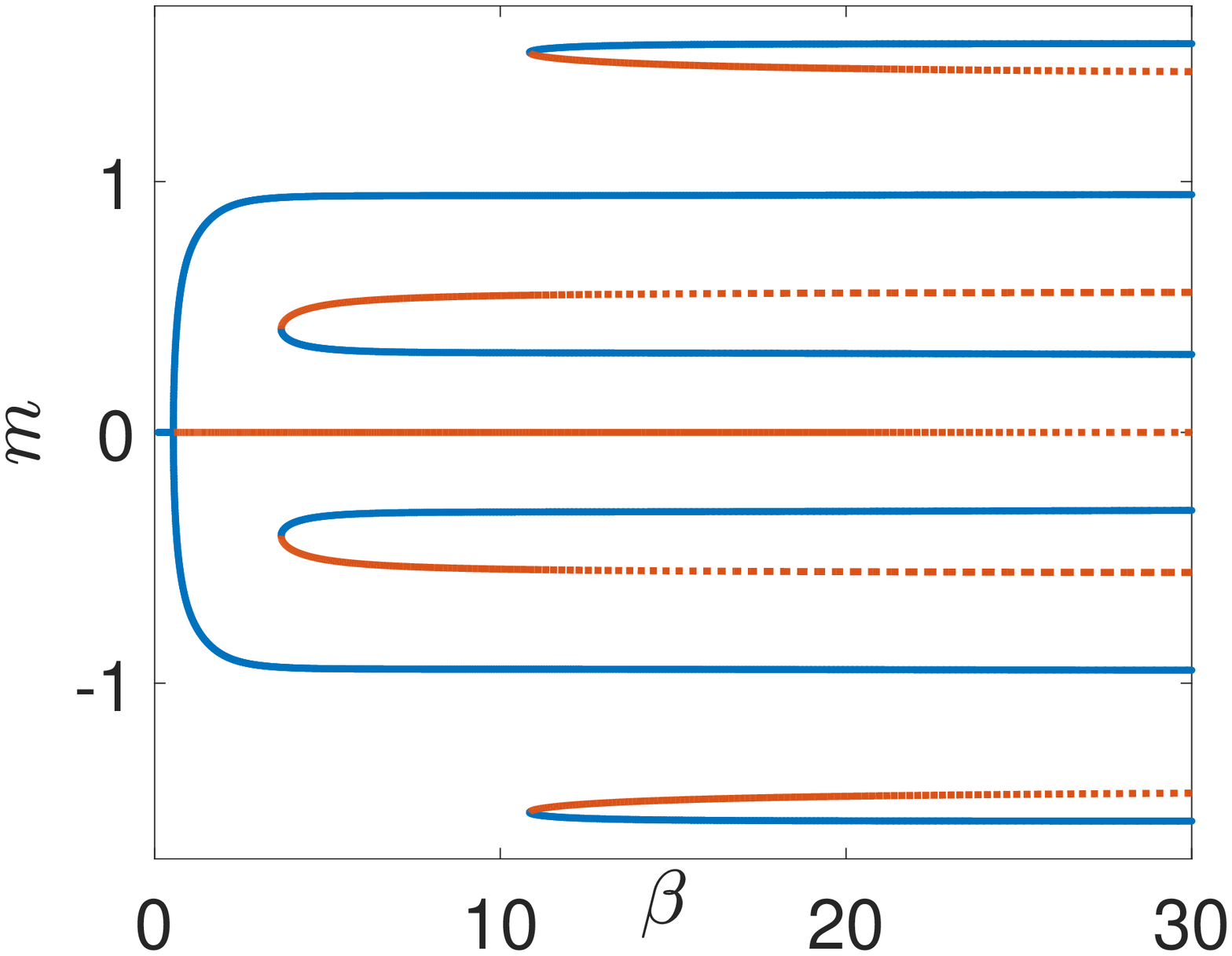}
\caption{Bifurcation diagram}\label{fig:bistable_add_bif}
\end{subfigure}
	\caption{Results for case $3$: bistable $V_0^b$ with separable fluctuations, for $\theta = 5, \, \delta = 1, \, \epsilon = 0.1$. \eqref{fig:bistable_add_mean} $R(m^\epsilon;\theta,\beta)$ for various values of $\beta$, with the potential $V^\epsilon(x)$ (full line) compared with $V_0^b(x) $ (dashed line) in the inside panel. \eqref{fig:bistable_add_bif} Bifurcation diagram of $m$ as a function of $\beta$. Full lines correspond to stable solutions, while dashed lines represent unstable ones.}
	 \label{fig:potential_additive}
 \end{figure}

Similarly, for the last case in Table~\ref{tab:pot}, case $4$ (bistable potential $V_0^b(x)$ and nonseparable fluctuations), there are more solutions to the selfconsistency equation. However, the flatness of the potential (and therefore of the curves $R(m;\theta,\beta)$ near $m=0$) reduces the number of additional branches. Moreover, the topological structure of the bifurcation diagram changes, and we now observe a nonparabolic curve for the main branch, which bifurcates from the mean zero solution via a pitchfork bifurcation. 

\begin{figure}[ht]
\begin{subfigure}{0.5\textwidth}
\includegraphics[width=\linewidth]{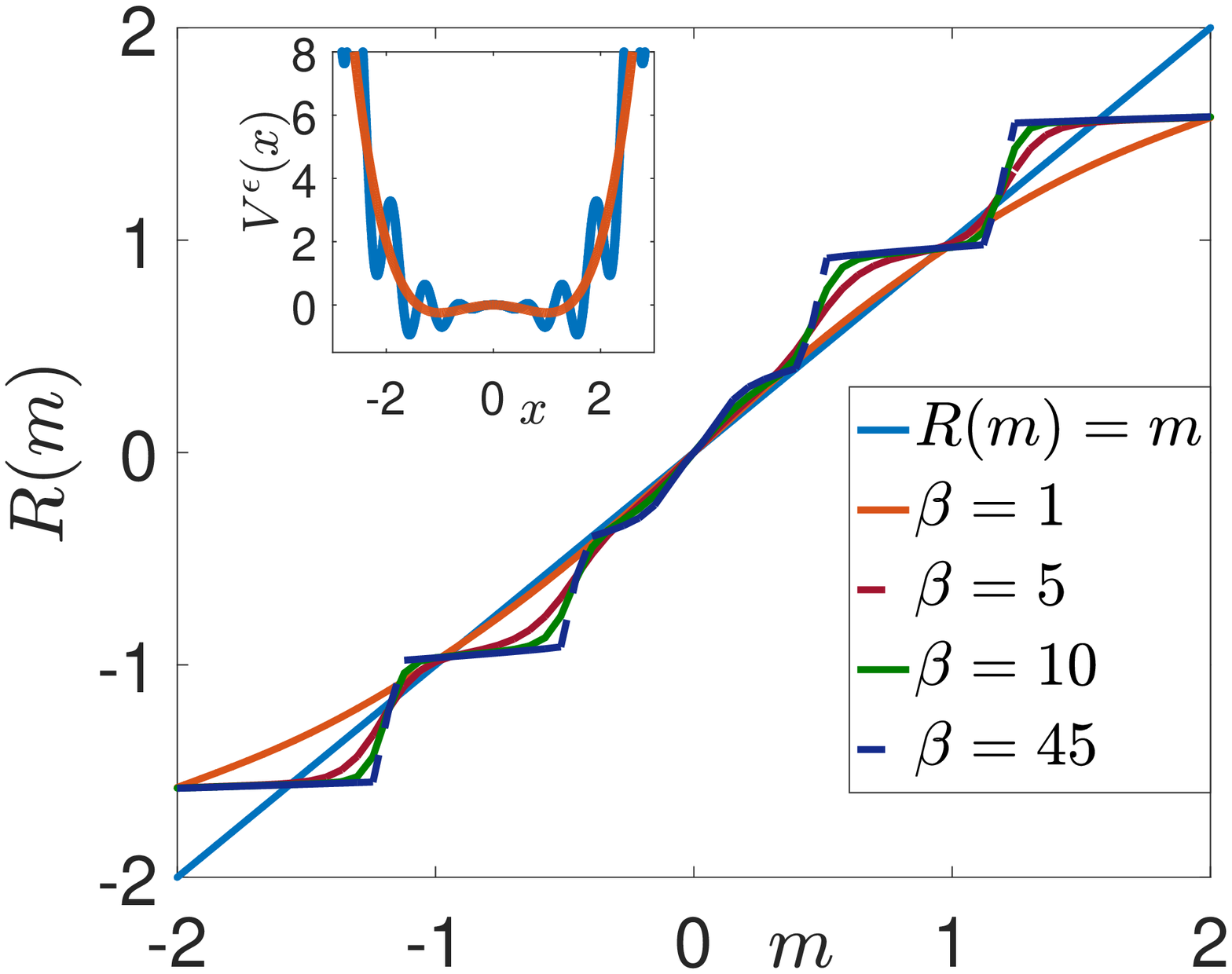}
\caption{$R(m;\theta,\beta)$}\label{fig:bistab_mult_mean}
\end{subfigure}
\begin{subfigure}{0.5\textwidth}
\includegraphics[width=\linewidth]{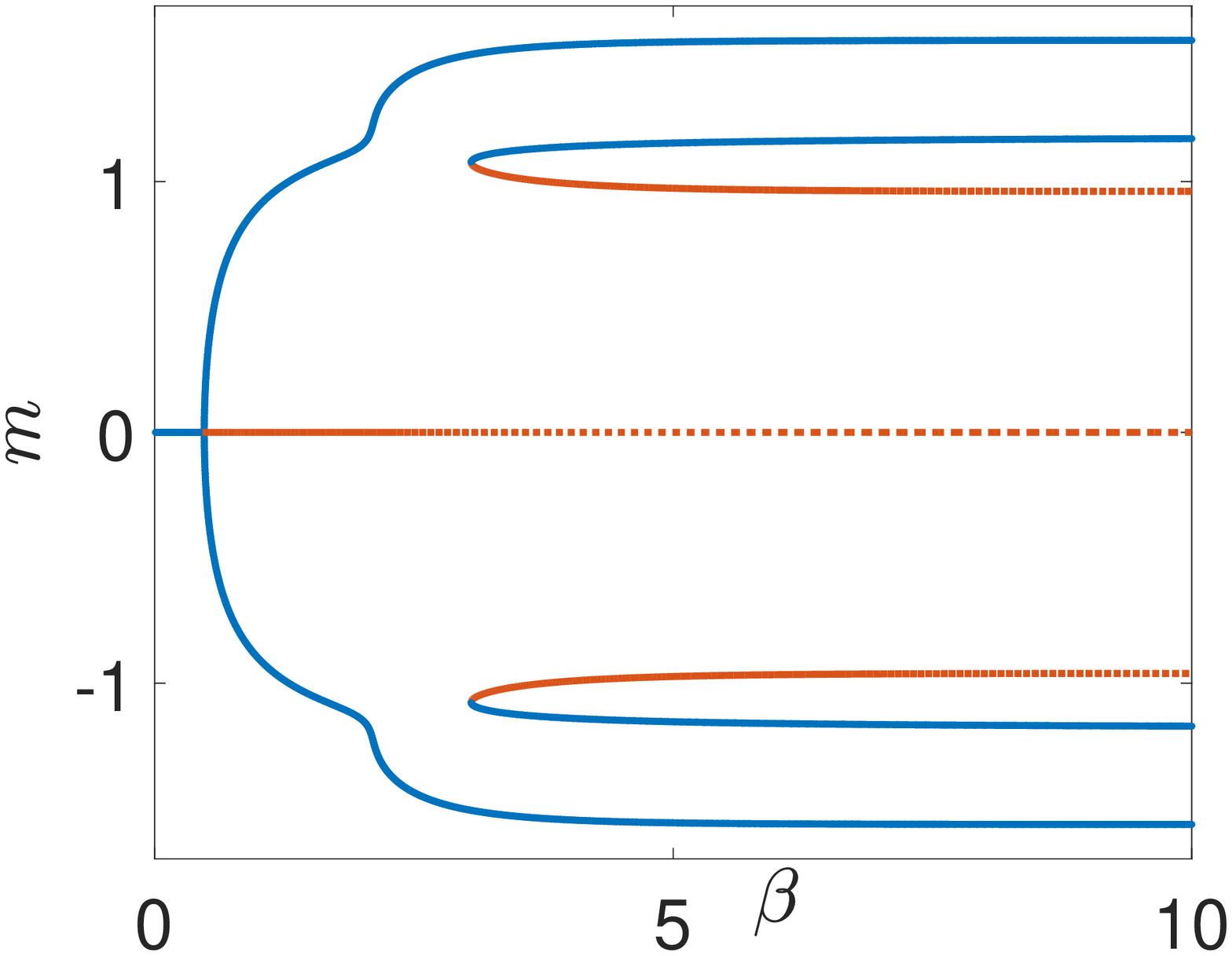}
\caption{Bifurcation diagram}\label{fig:bistable_mult_bif}
\end{subfigure}
	\caption{Results for case $4$: bistable $V_0^b$ with nonseparable fluctuations, for $\theta = 5, \, \delta = 1, \, \epsilon = 0.1$. \eqref{fig:bistab_mult_mean} $R(m^\epsilon;\theta,\beta)$ for various values of $\beta$, with the potential $V^\epsilon(x)$ (full line) compared with $V_0^b(x) $ (dashed line) in the inside panel. \eqref{fig:bistable_mult_bif} Bifurcation diagram of $m$ as a function of $\beta$. Full lines correspond to stable solutions, while dashed lines represent unstable ones.}
	 \label{fig:potential_multiplicative}
 \end{figure}

%
%
\subsection{Numerical study of the critical temperature as a function of $\epsilon$}
Here we study the influence of finite $\epsilon$ on the critical temperature $\beta_C$, the solution of~\eqref{eq:VarianceLimit} for two-scale potentials, after which continuous phase transitions (pitchfork bifurcations) occur. We do this by solving the equation (we only consider symmetric potentials)
\begin{equation}
\label{eq:critical_temp_eqn}
\theta^{-1}\beta^{-1} = \int_\R x^2 p_\infty(x;\theta,\beta,m^\epsilon=0) \ dx,
\end{equation}
for the various potentials in Table~\ref{tab:pot}. 

\begin{figure}[ht]
\begin{subfigure}{0.32\textwidth}
\includegraphics[width=\linewidth]{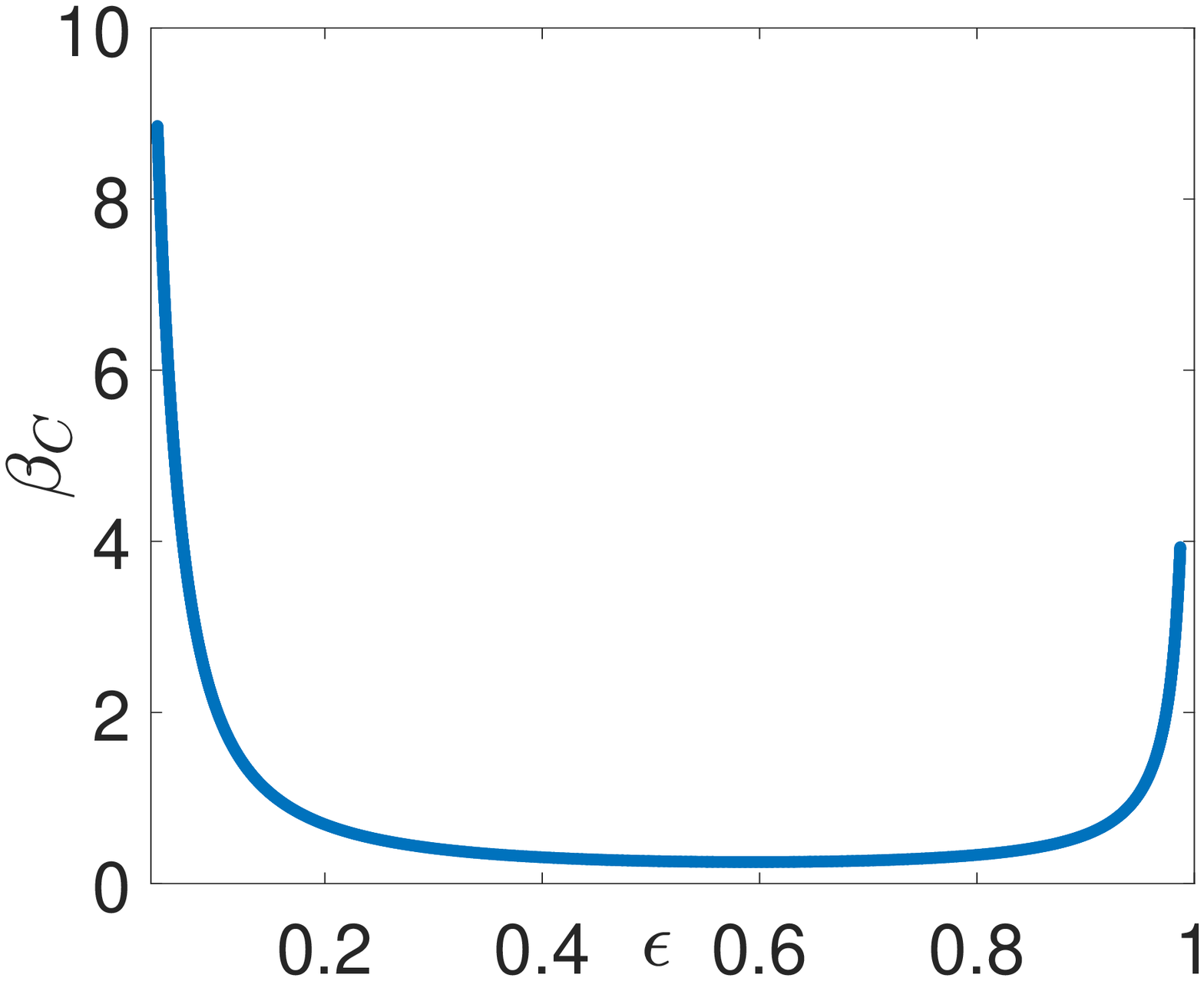}
\caption{Case 1}\label{fig:B_C_1}
\end{subfigure}
\begin{subfigure}{0.32\textwidth}
\includegraphics[width=\linewidth]{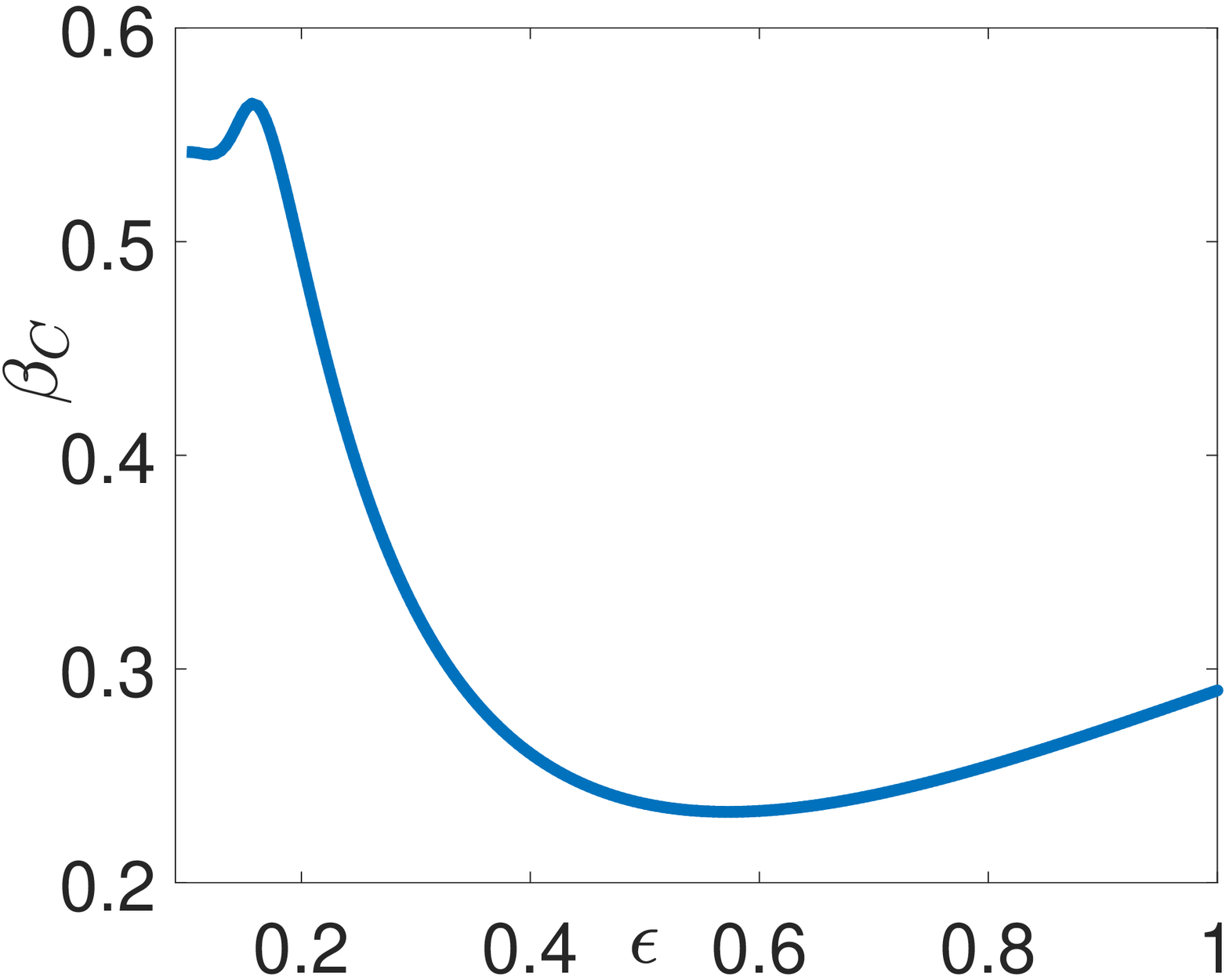}
\caption{Case 3}\label{fig:B_C_3}
\end{subfigure}
\begin{subfigure}{0.32\textwidth}
\includegraphics[width=\linewidth]{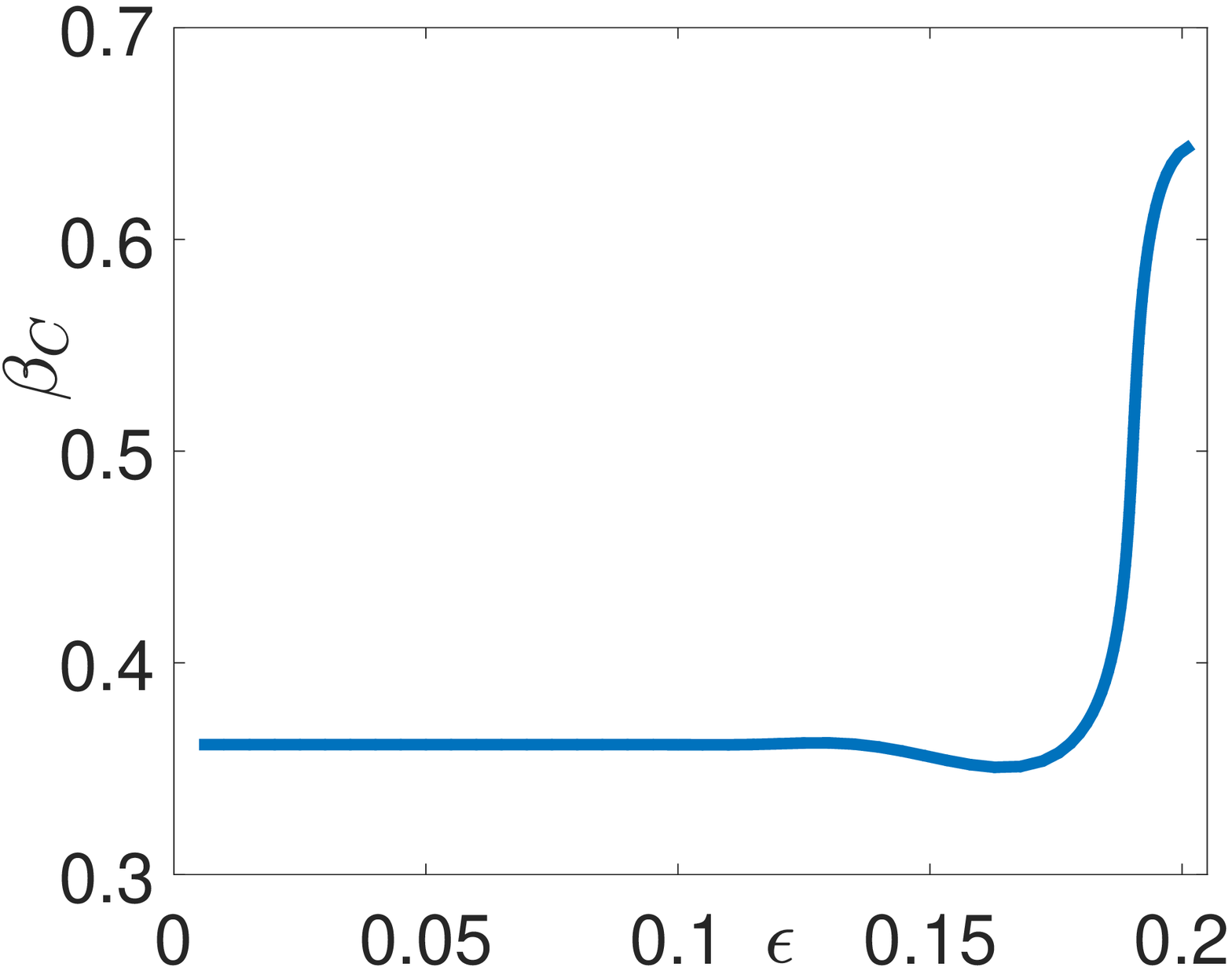}
\caption{Case 4}\label{fig:B_C_4}
\end{subfigure}
\caption{Critical temperature $\beta_C$ as a function of $\epsilon$ for the multiscale Fokker-Planck equation with $\theta = 5$ for cases \eqref{fig:B_C_1} $1$ - $V^\eps(x) = \frac{x^2}{2} + \delta\cos\left(\frac{x}{\eps}\right)$, \eqref{fig:B_C_3} $3$ - $V^\eps(x) = \frac{x^4}{4} - \frac{x^2}{2} + \delta\cos\left(\frac{x}{\eps}\right)$, and \eqref{fig:B_C_4} $4$ - $V^\eps(x) = \frac{x^4}{4} -\frac{x^2}{2}\left(1 - \delta\cos\left(\frac{x}{\eps}\right)\right)$  in Table~\ref{tab:pot}.}
	 \label{fig:critical_temperature}
 \end{figure}

We present in Figure~\ref{fig:critical_temperature} plots of the critical temperature, $\beta_C$ as a function of $\epsilon$ for a fixed $\theta = 5$. The results are presented for cases $1$ (Figure~\ref{fig:B_C_1}), $3$ (Figure~\ref{fig:B_C_3}) and $4$ (Figure~\ref{fig:B_C_4}) from Table~\ref{tab:pot}. 
We do not present the remaining case because, as can be observed in Figure~\ref{fig:conv_mult_bif}, there is no pitchfork bifurcation from the zero mean solution for case 2. The dependence of the critical temperature on $\eps$ is different for separable and nonseparable potentials. It appears that the critical temperature can change considerable by varying $\eps$, which implies that a different number of branches might be present in the bifurcation diagram at a fixed temperature, for different values of $\eps$. This issue will be studied in detail in future work.


%
%
\subsection{Simulations of the interacting particles system}\label{sec:MC}
In this section we present the results of Monte Carlo (MC) simulations for the system of interacting diffusions, both for the full, i.e. $\eps-$dependent,~\eqref{eq:system_of_sdes_in_1D} and for the homogenized dynamics~\eqref{eq:system_of_homogenized_sdes}. Our focus is on the study of the convergence of the interacting particles system to their equilibrium state. It should be emphasized that no phase transitions occur for the finite dimensional particles system. However, the numerical simulation of the two interacting particles systems,~\eqref{eq:system_of_sdes_in_1D} and the homogenized particle system~\eqref{eq:system_of_homogenized_sdes} clearly exhibit the lack of commutativity between the mean field and homogenization limits.

For the full dynamics~\eqref{eq:system_of_sdes_in_1D}, we used $\delta = 1$ and $\epsilon = 0.1$. We solved the SDEs using the Euler-Maruyama scheme. For the homogenized dynamics~\eqref{eq:system_of_homogenized_sdes}, since the noise is multiplicative (for nonseparable potentials), we used the Milstein scheme. In both cases, the time step used 
was $dt = 0.01$, which is of $O(\epsilon^2)$. Finally, in both cases we initialized the $N$ particles as being normally distributed, with mean zero and variance $4$, which was large enough so that all the local minima were contained within two standard deviations of the Gaussian distribution.

In Figures~\ref{fig:convex_add_position}--\ref{fig:convex_add_mean} we present the results of our simulations for Case 1 in Table~\ref{tab:pot}, the convex potential with separable fluctuations $V^{\eps}(x) = \frac{x^2}{2} + \delta \cos \left(\frac{x}{\eps} \right)$. In Figure~\ref{fig:convex_add_position} we present snapshots of the position of each of the $N=1,000$ particles for $t=0$ (top panels), $t=100$ (middle panels) and $t=5,000$ (bottom panels). The left panels show the results for $\epsilon=0.1$, while the right panels show the results for the homogenized system.
\begin{figure}[h!]
\hskip-0.75cm
	\includegraphics[width=0.550\linewidth]{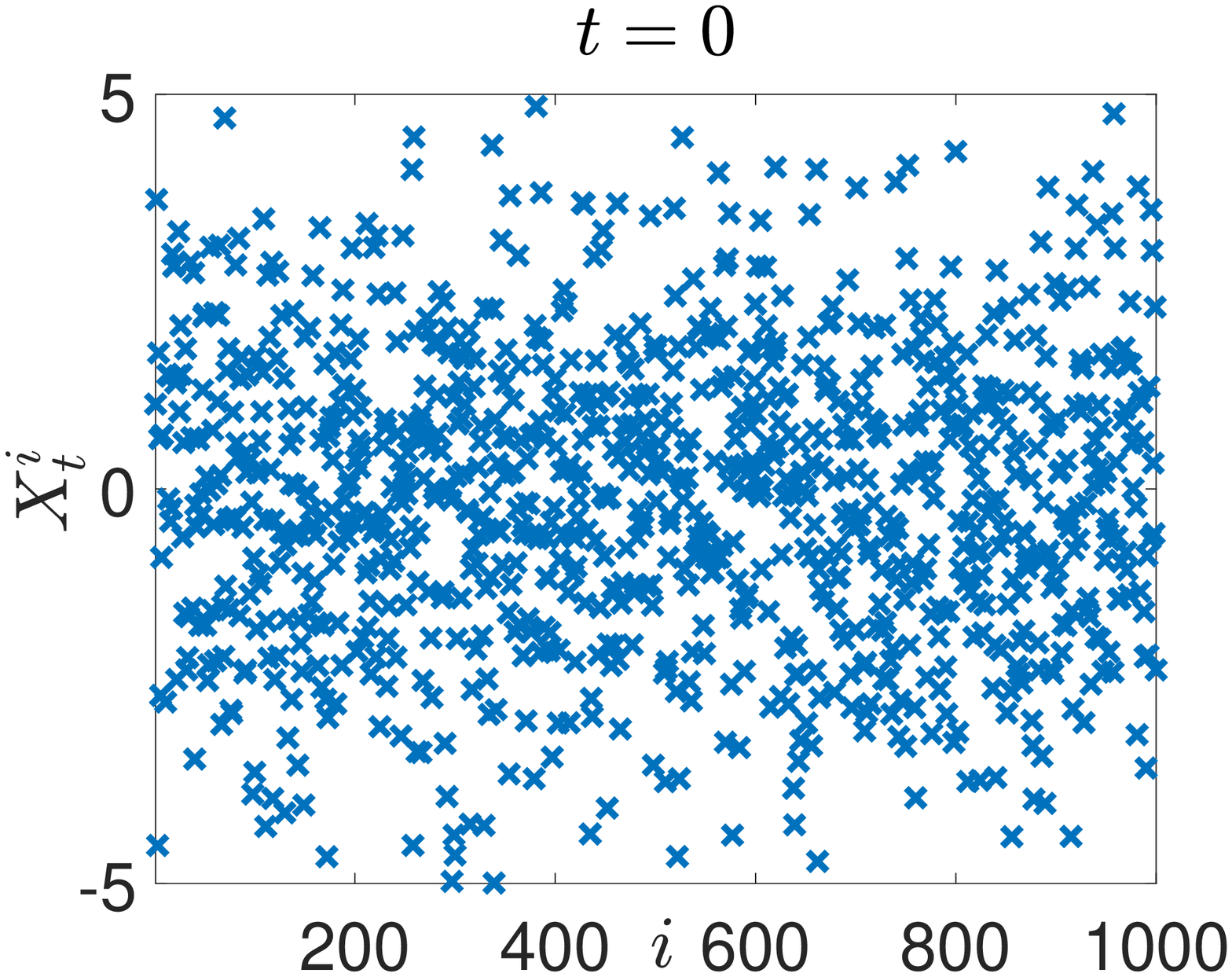}
	\includegraphics[width=0.550\linewidth]{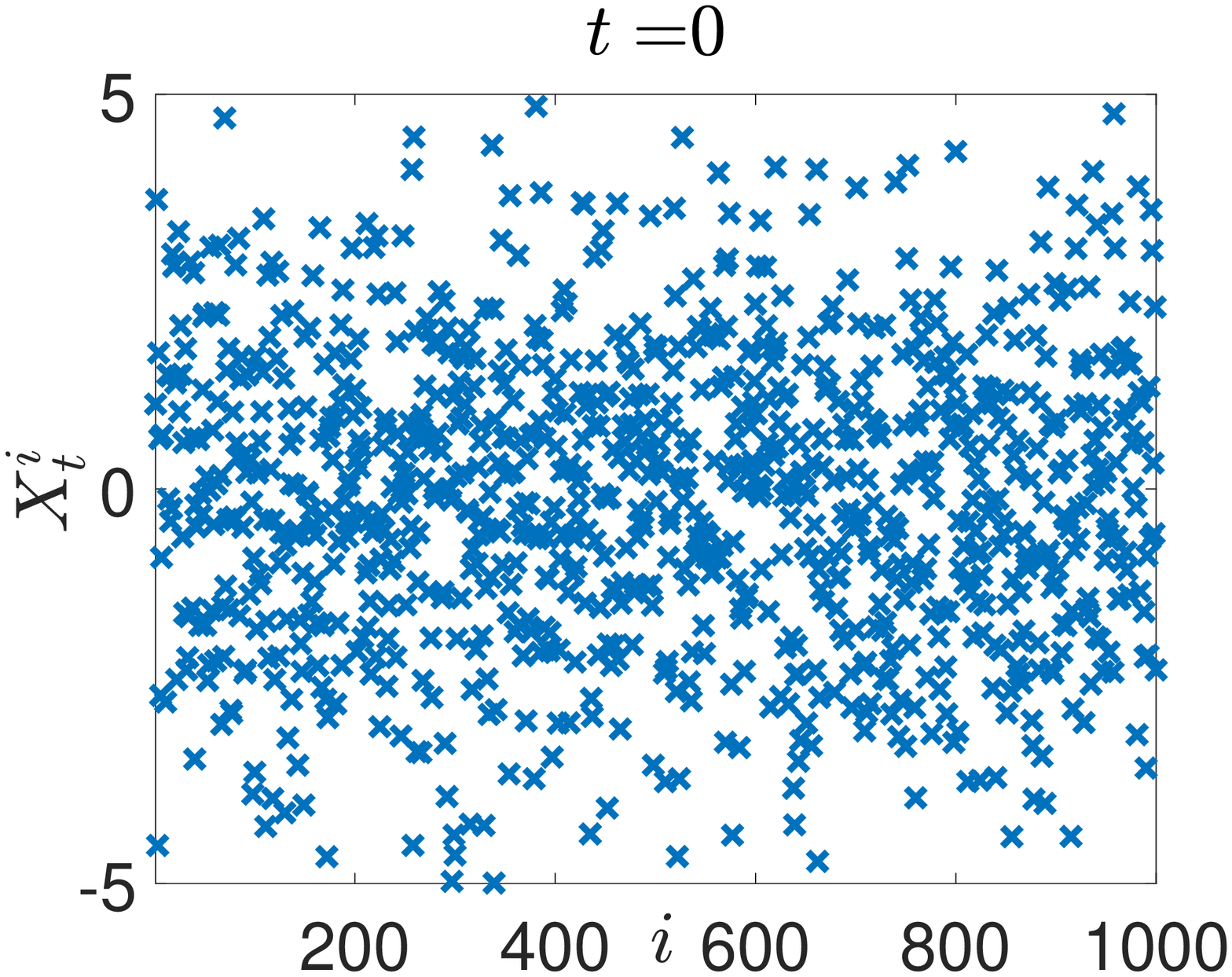}

\hskip-0.75cm
	\includegraphics[width=0.550\linewidth]{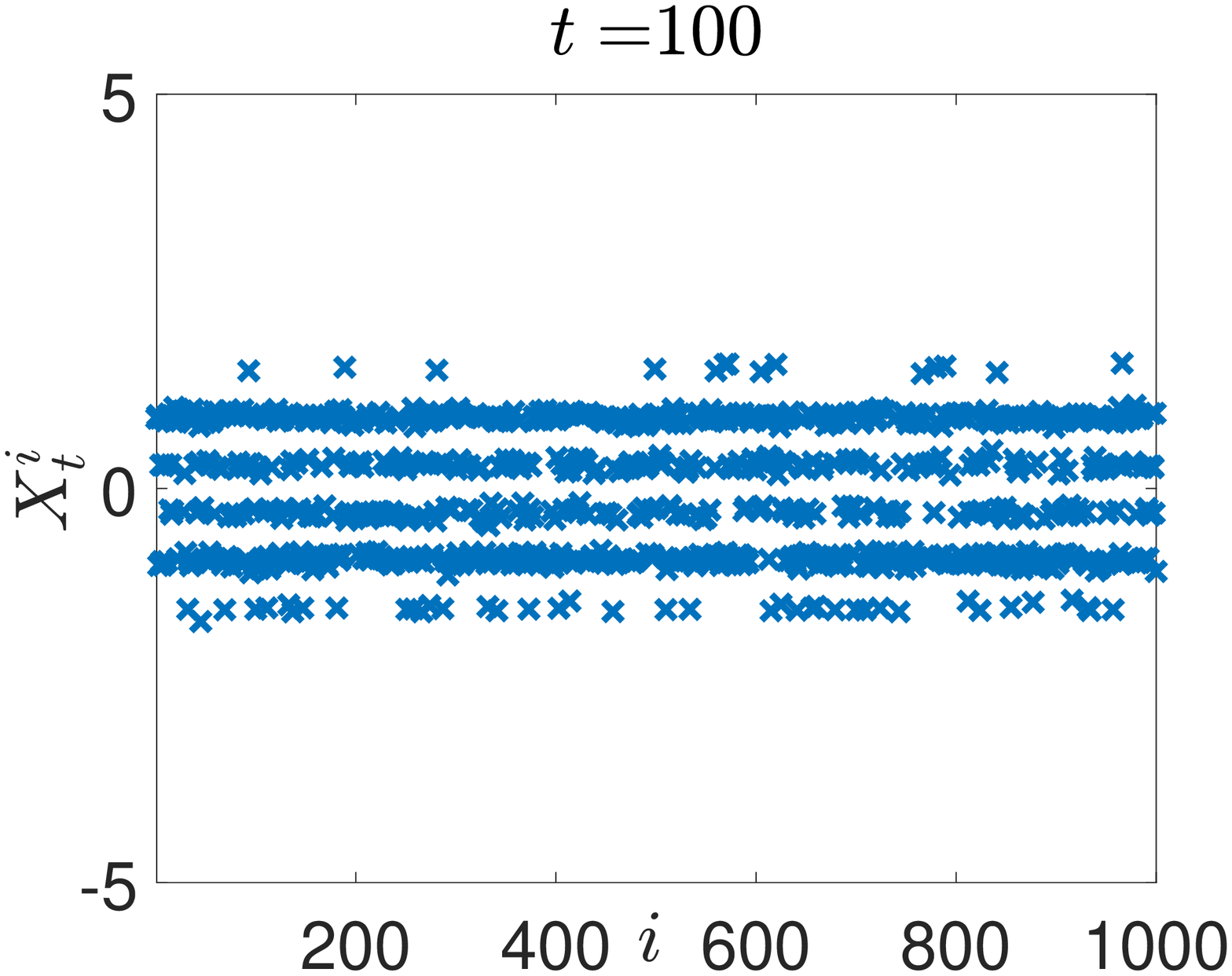}
	\includegraphics[width=0.550\linewidth]{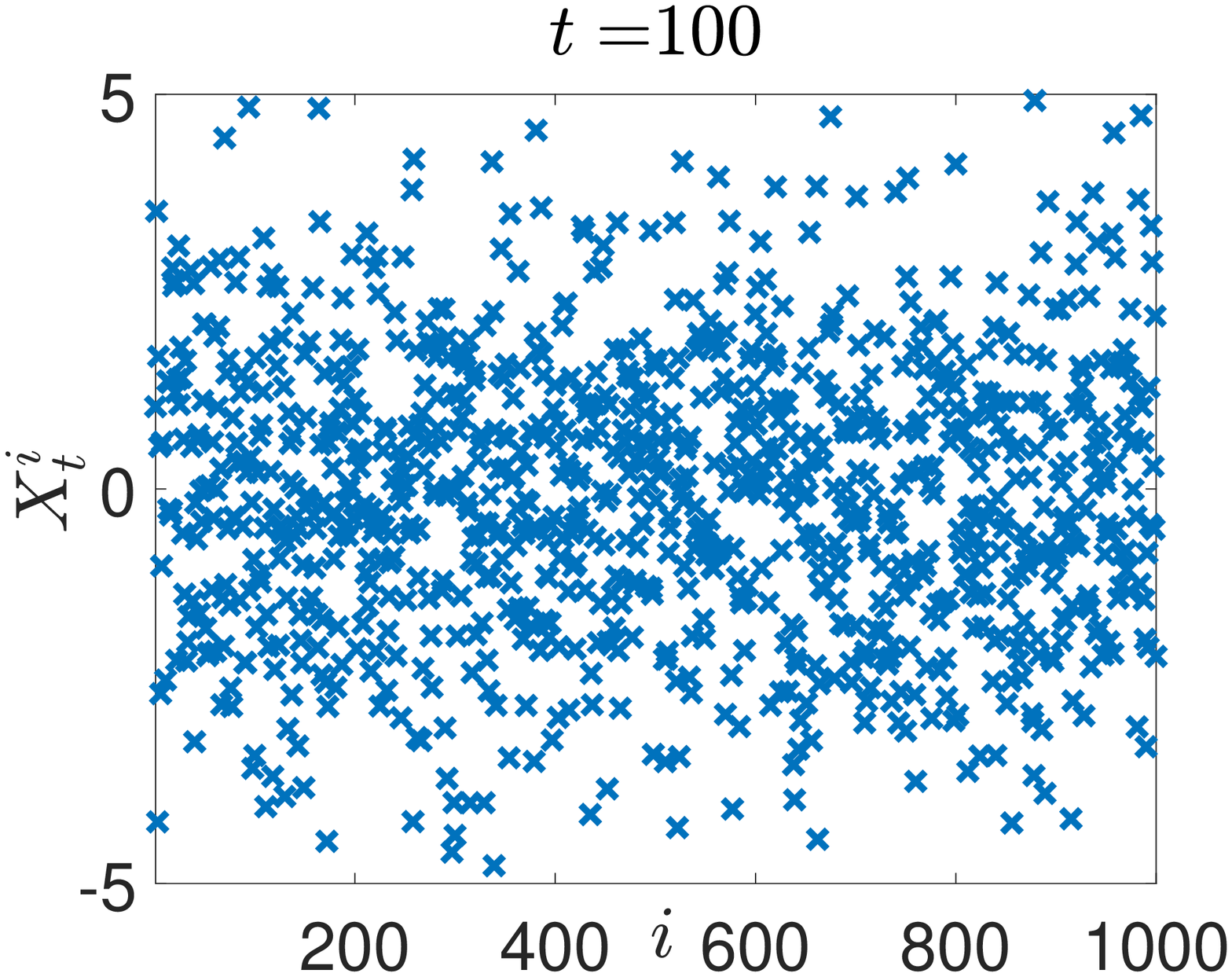}

\hskip-0.75cm
	\includegraphics[width=0.550\linewidth]{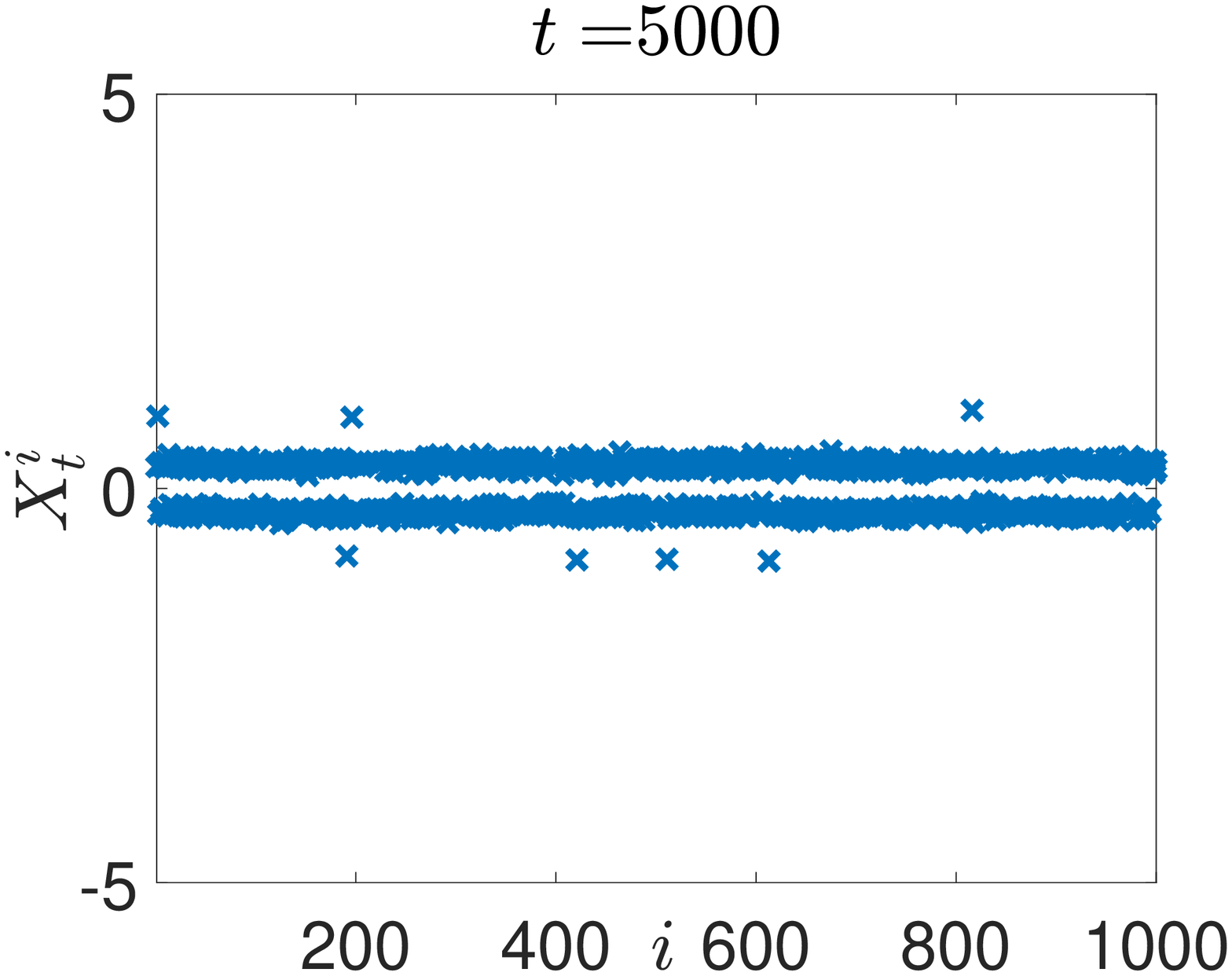}
	\includegraphics[width=0.550\linewidth]{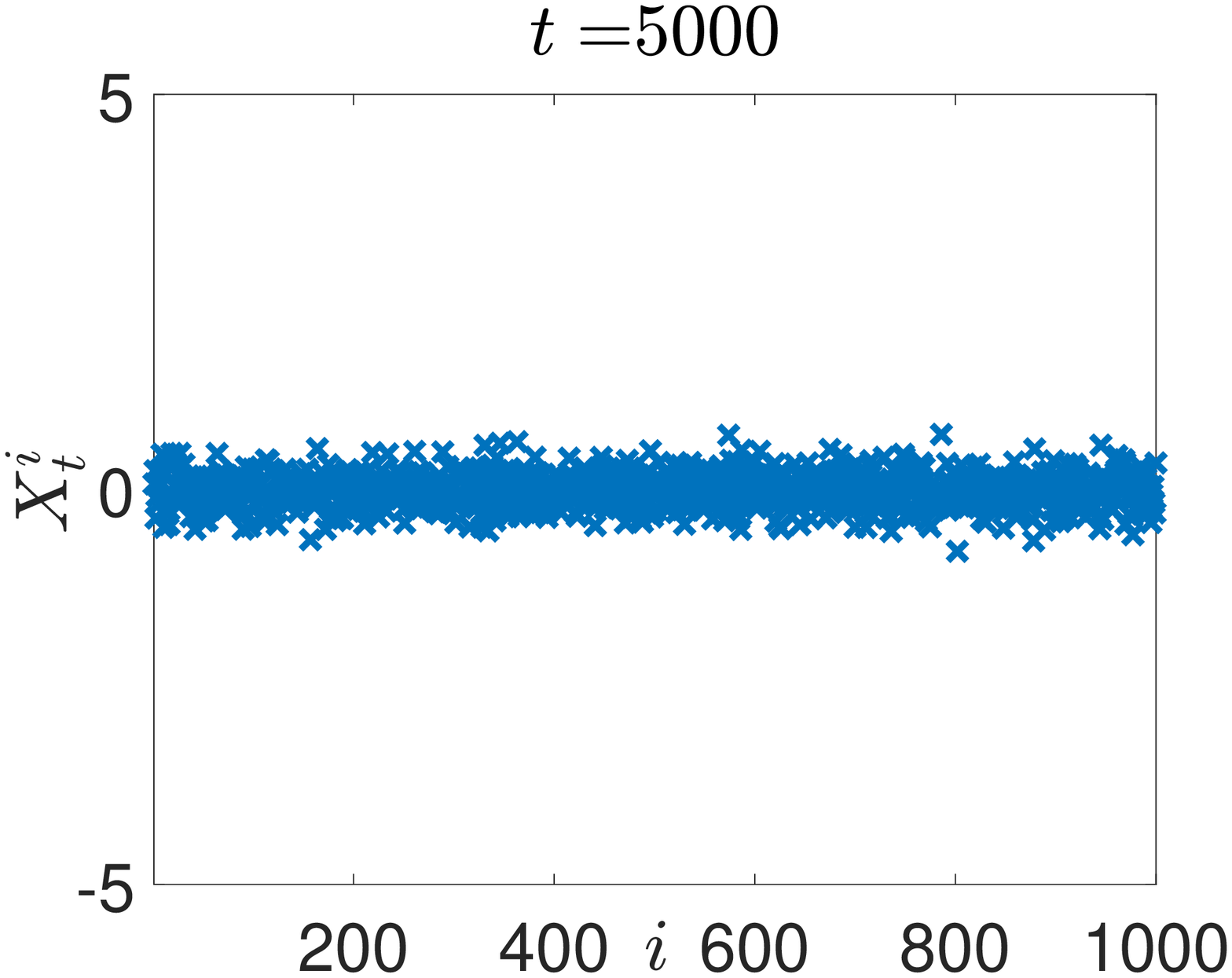}

\caption{Position of $N=1,000$ particles for $V^{\eps}(x) = \frac{x^2}{2} + \delta \cos \left(\frac{x}{\eps} \right)$, with $\theta = 2$, $\beta = 8$, $\delta = 1$. Left: Eqn.~\eqref{eq:system_of_sdes_in_1D} with $\epsilon = 0.1$. Right: homogenized SDEs~\eqref{eq:system_of_homogenized_sdes}. }
	 \label{fig:convex_add_position}
 \end{figure}
In Figure~\ref{fig:convex_add_hist}, we present snapshots of the histogram for the $N=1,000$ particles for the same time and parameter values, which are $\delta = 1, \, \epsilon = 0.1, \, \theta = 2$ and $\beta = 8$. On the $t=5,000$ snapshot, we superpose the corresponding invariant measure, rescaled for comparison, and we observe that the empirical density of the system of interacting diffusions converges to the steady state solution computed by solving the stationary McKean-Vlasov equation. It is clear from the histograms at $t=5,000$ that, even though the invariant distribution of the full dynamics converges weakly to the invariant distribution of the homogenized SDEs, see Corollary~\ref{prop:additive_fluctuations}, the two distributions are quite different for finite values of $\eps$.
\begin{figure}[h!]
\hskip-0.75cm
	\includegraphics[width=0.550\linewidth]{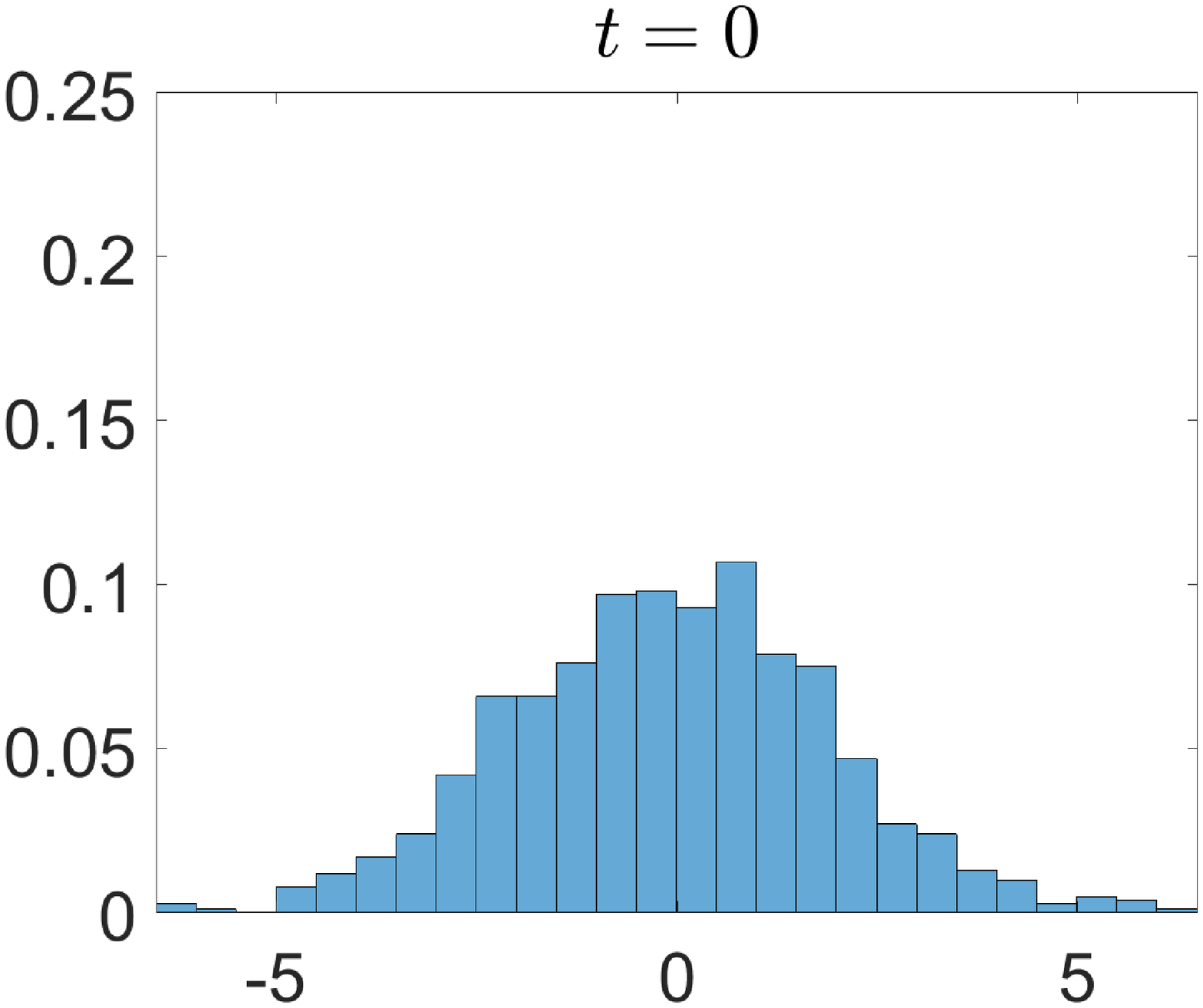}
	\includegraphics[width=0.550\linewidth]{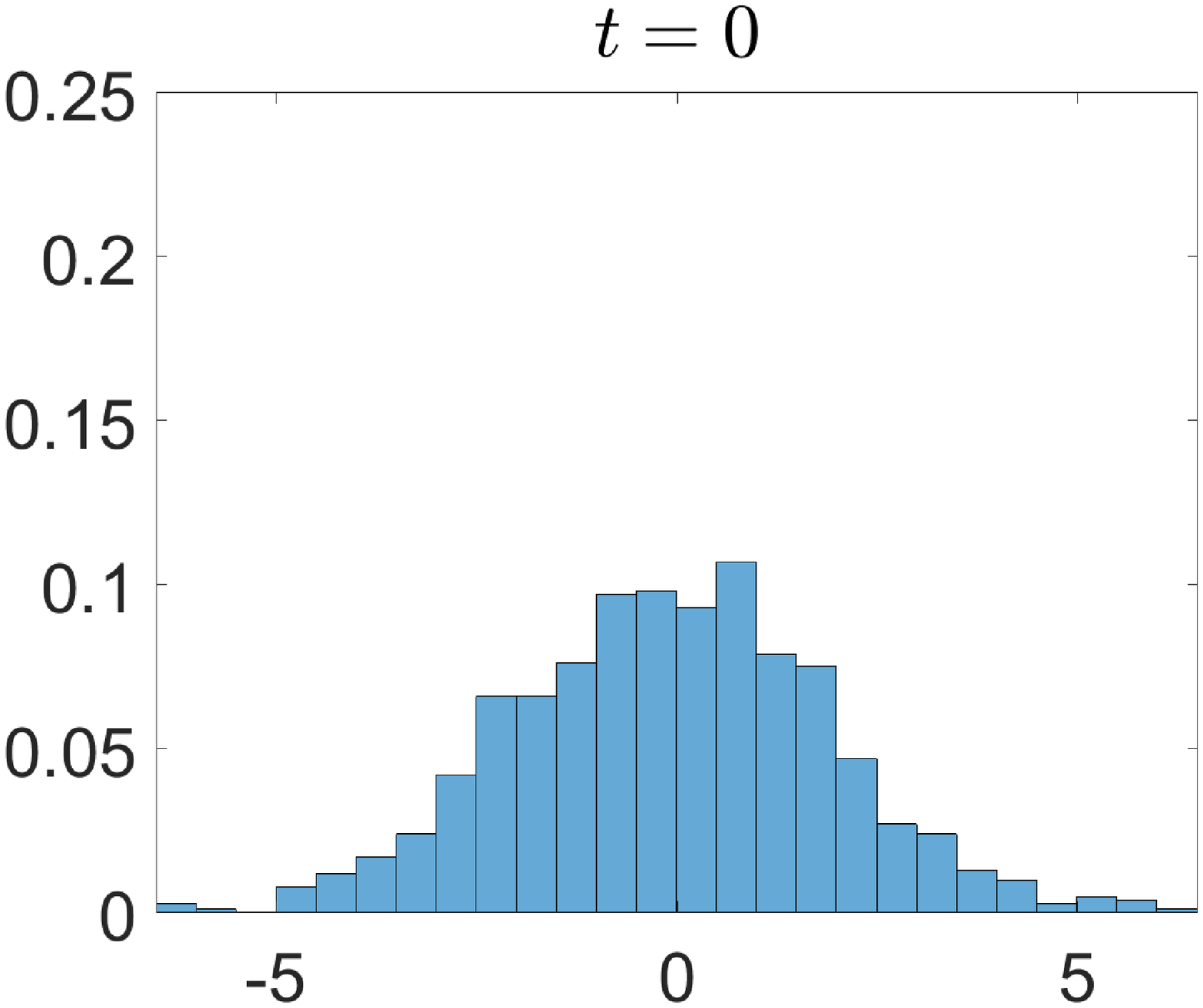}

\hskip-0.75cm
	\includegraphics[width=0.550\linewidth]{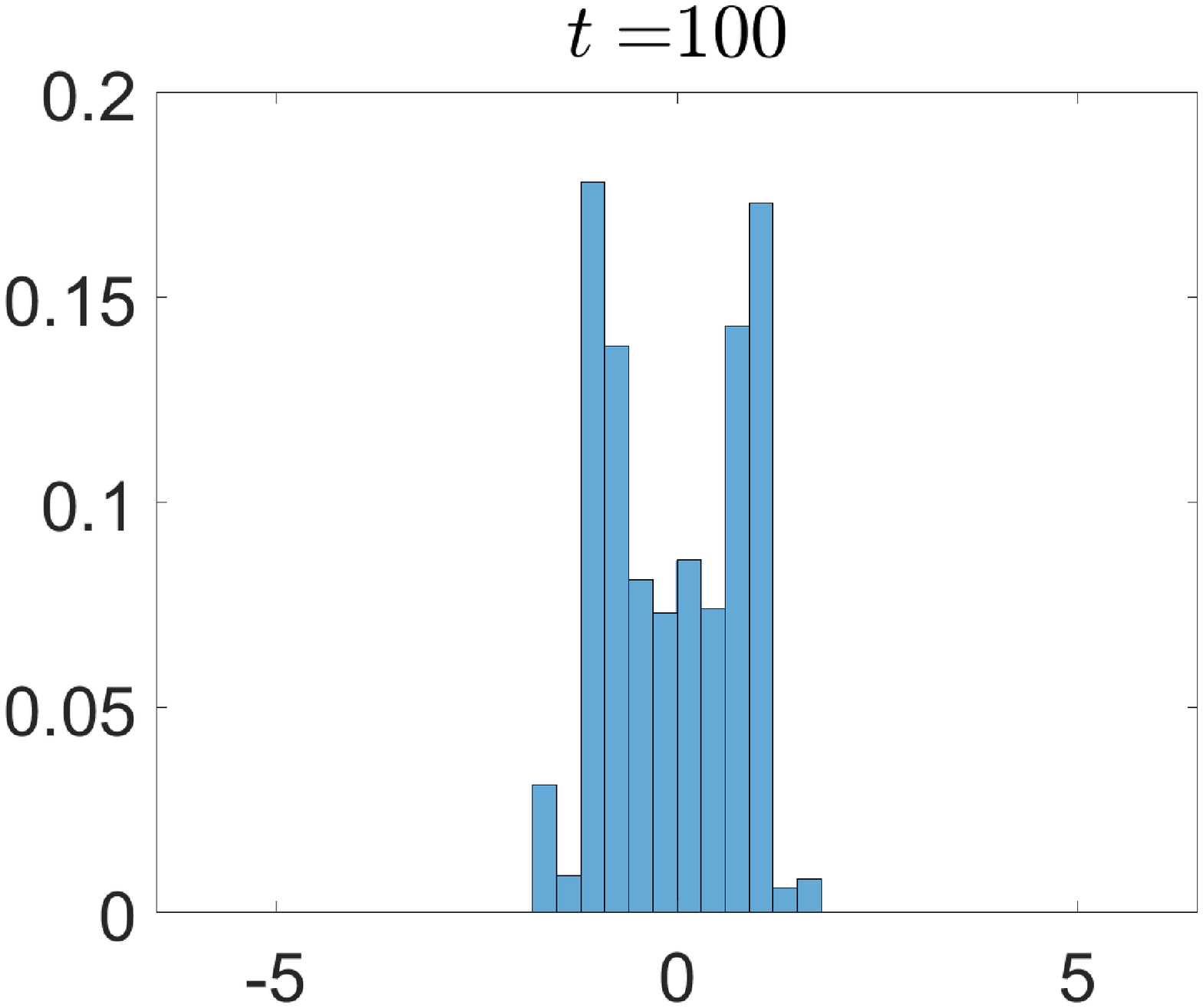}
	\includegraphics[width=0.550\linewidth]{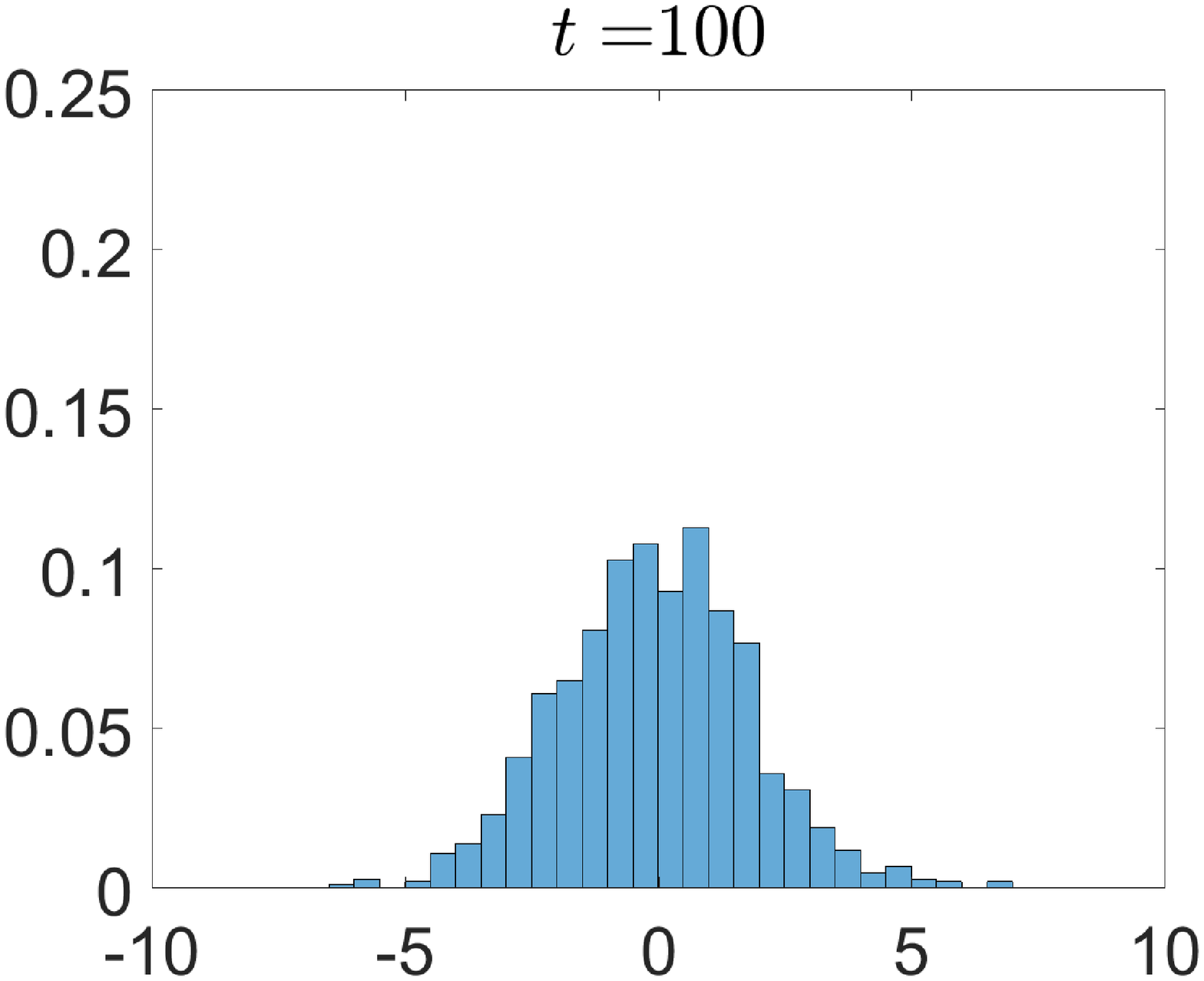}

\hskip-0.75cm
	\includegraphics[width=0.550\linewidth]{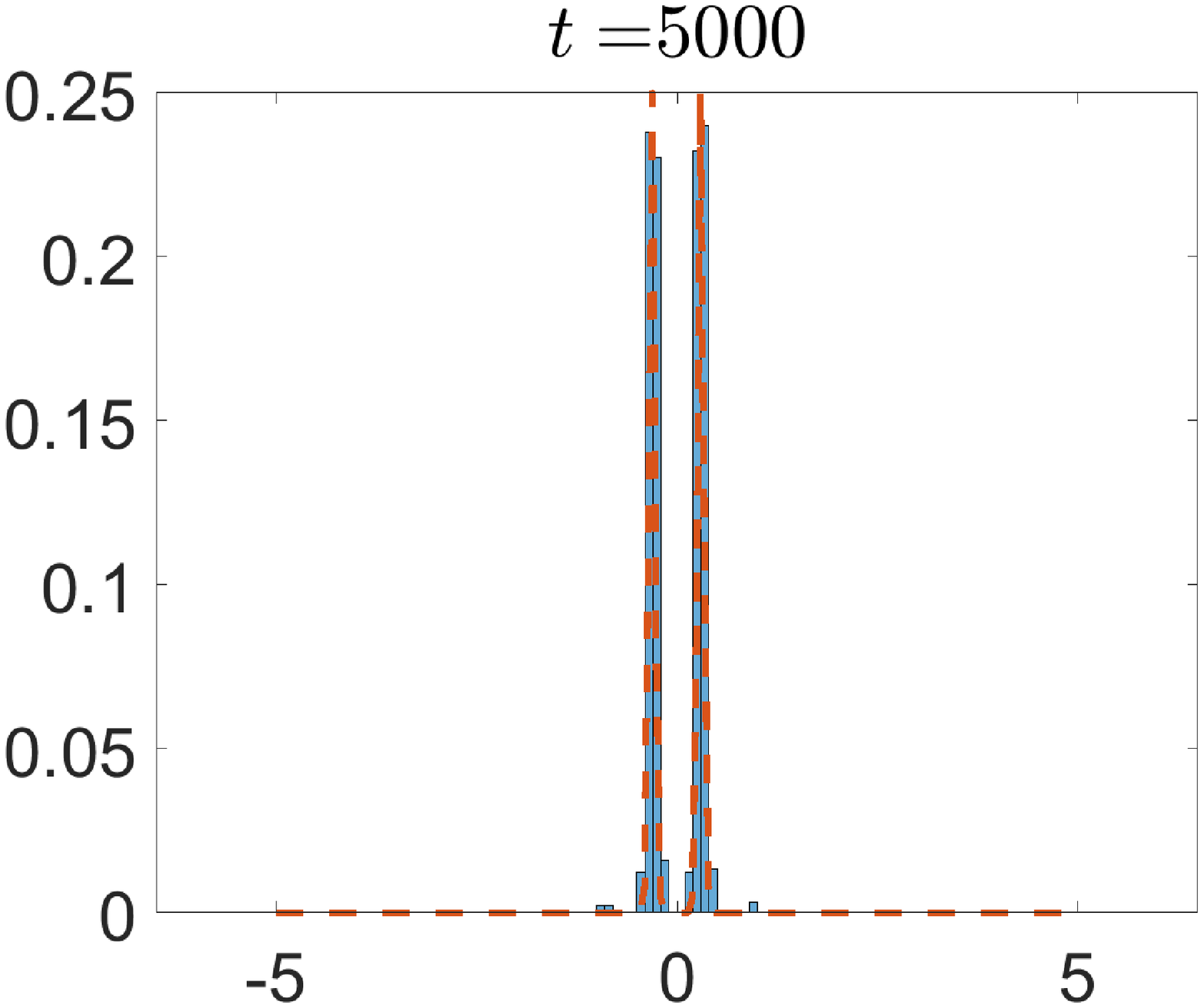}
	\includegraphics[width=0.550\linewidth]{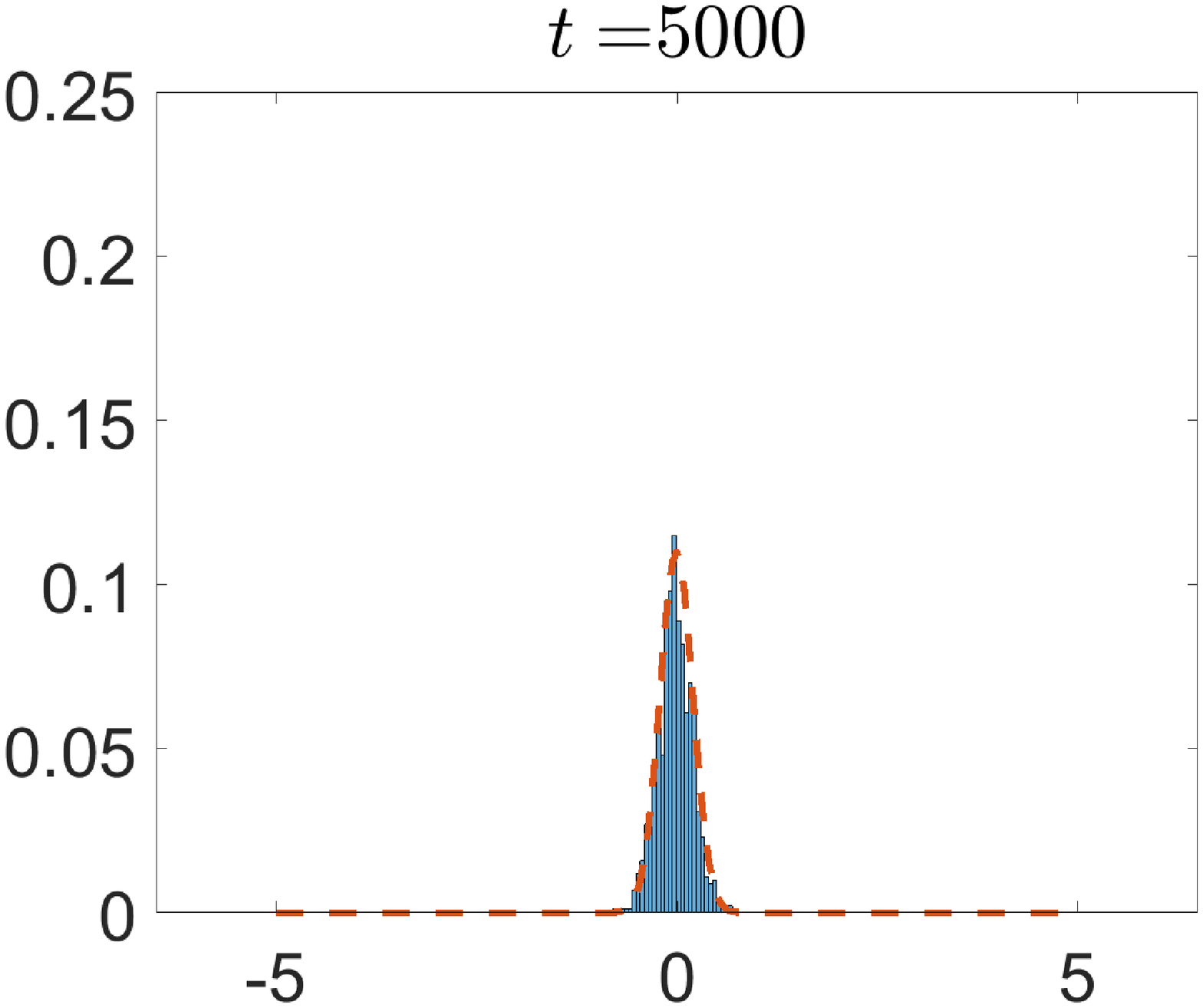}

\caption{Histogram of $N=1,000$ particles for $V^{\eps}(x) = \frac{x^2}{2} + \delta \cos \left(\frac{x}{\eps} \right)$, with $\theta = 2$, $\beta = 8$, $\delta = 1$. Left: Eqn.~\eqref{eq:system_of_sdes_in_1D} with $\epsilon = 0.1$. Right: homogenized SDEs~\eqref{eq:system_of_homogenized_sdes}. }
	 \label{fig:convex_add_hist}
 \end{figure}
We also calculate the mean of the interacting particle system 
\begin{equation}\label{e:mean_N}
\bar{X}_t^N := \frac{1}{N}\sum_{i=1}^N X_t^i,
\end{equation} 
as a function of time $t$. We observe that in both cases, the average converges to $0$ as expected, but that the convergence for the homogenized SDE~\eqref{eq:system_of_homogenized_sdes} is slower. The position of the $N$ particles follows approximately the same qualitative behavior (with the particles clustering close to 
$0$), but as we can see from the corresponding histogram there exist additional wells (nonconvexity) for the finite $\epsilon$ case.
\begin{figure}[h!]
\hskip-0.75cm
	\includegraphics[width=0.550\linewidth]{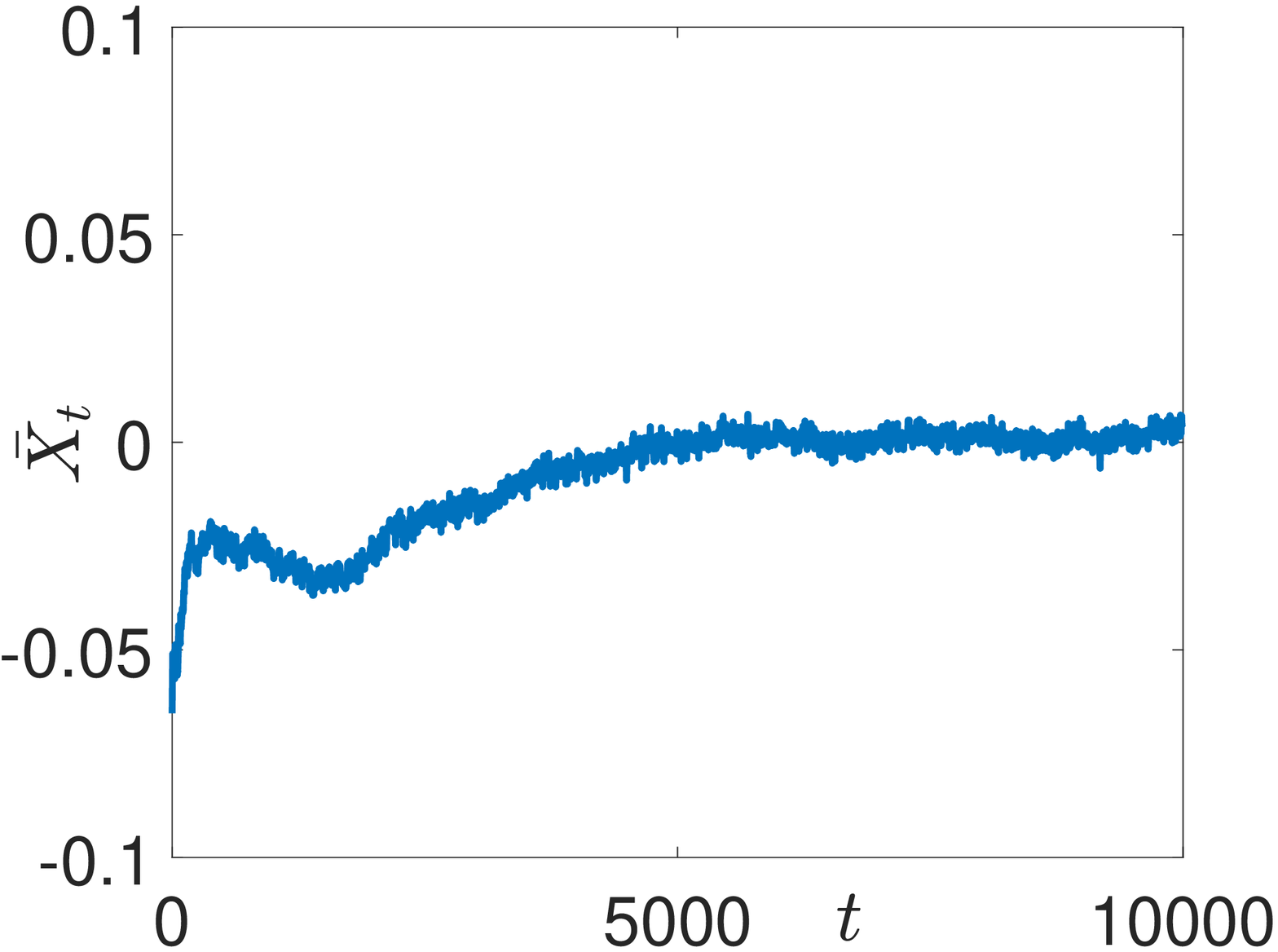}
	\includegraphics[width=0.550\linewidth]{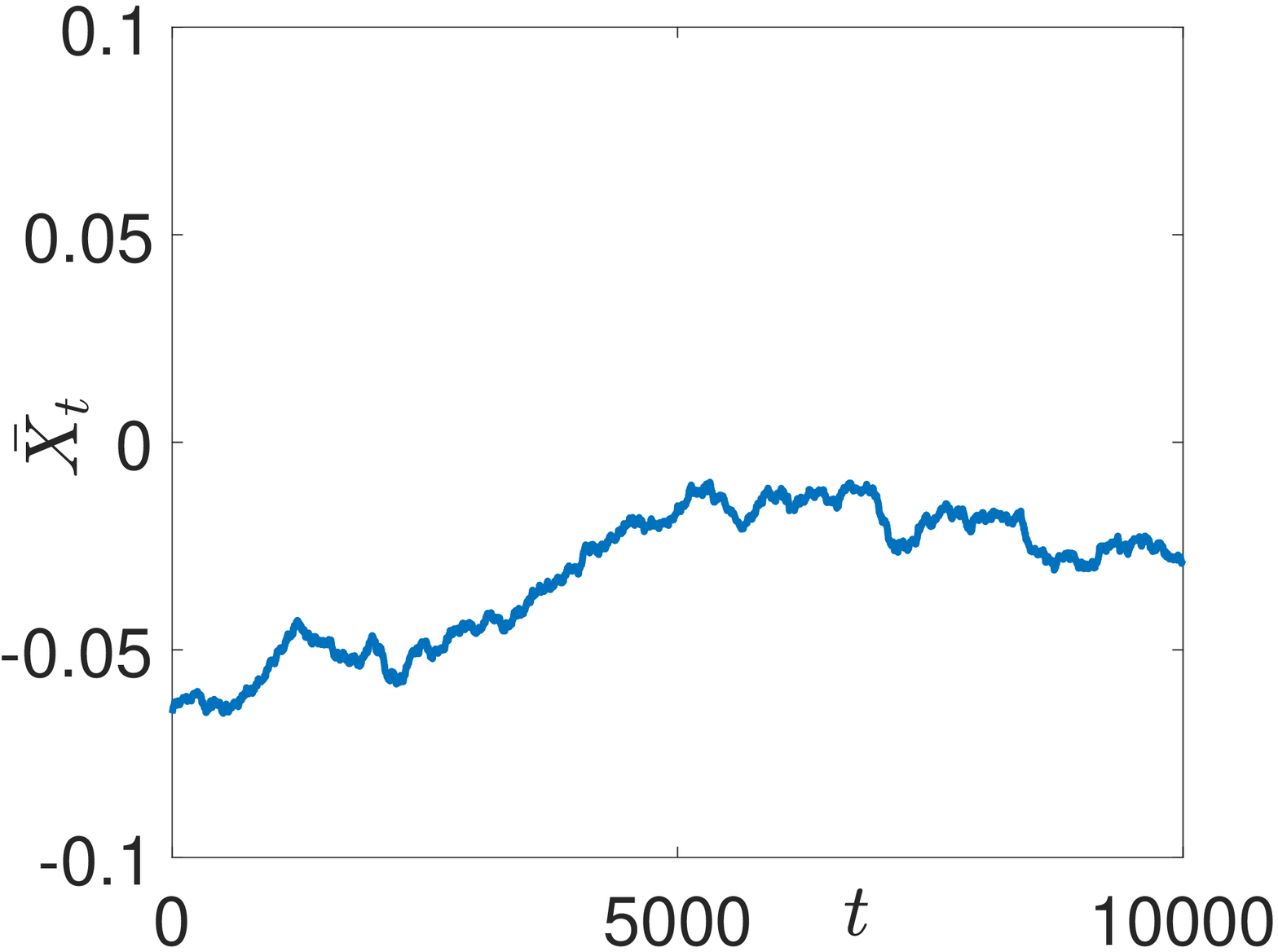}
\caption{Time evolution of the mean $\bar{X}_t^N = \frac{1}{N}\sum_{i=1}^N X_t^i$ of $N=1,000$ particles for $V^{\eps}(x) = \frac{x^2}{2} + \delta \cos \left(\frac{x}{\eps} \right)$, with $\theta = 2$, $\beta = 8$, $\delta = 1$. Left: Eqn.~\eqref{eq:system_of_sdes_in_1D} with $\epsilon = 0.1$. Right: homogenized SDEs~\eqref{eq:system_of_homogenized_sdes}.  }

	 \label{fig:convex_add_mean}
 \end{figure}

We performed similar experiments for Case $4$ in Table~\ref{tab:pot} (i.e., $V^\eps(x) = \frac{x^4}{4}-\frac{x^2}{2}\left(1-\delta\chi_{[-a,a]}(x)\cos\left(\frac{x}{\eps}\right)\right)$). Here we used $N=500$ particles, and smaller values of $\theta$ and $\beta$. The parameters used were $\theta = 0.5$, $\beta \approx 5.6$, $\delta = 1$ and $\epsilon = 0.1$ and the results are plotted in Figures~\ref{fig:bistable_mult_position}-\ref{fig:bistable_mult_mean}.

In Figure~\ref{fig:bistable_mult_position} we present snapshots of the position of each of the $N=500$ particles for $t=0$ (top panels), $t=100$ (middle panels) 
and $t=5,000$ (bottom panels). The left panels show the results for $\epsilon=0.1$, while the right panels show the results for the homogenized SDE~\eqref{eq:system_of_homogenized_sdes}.
\begin{figure}[h!]
\hskip-0.75cm
	\includegraphics[width=0.550\linewidth]{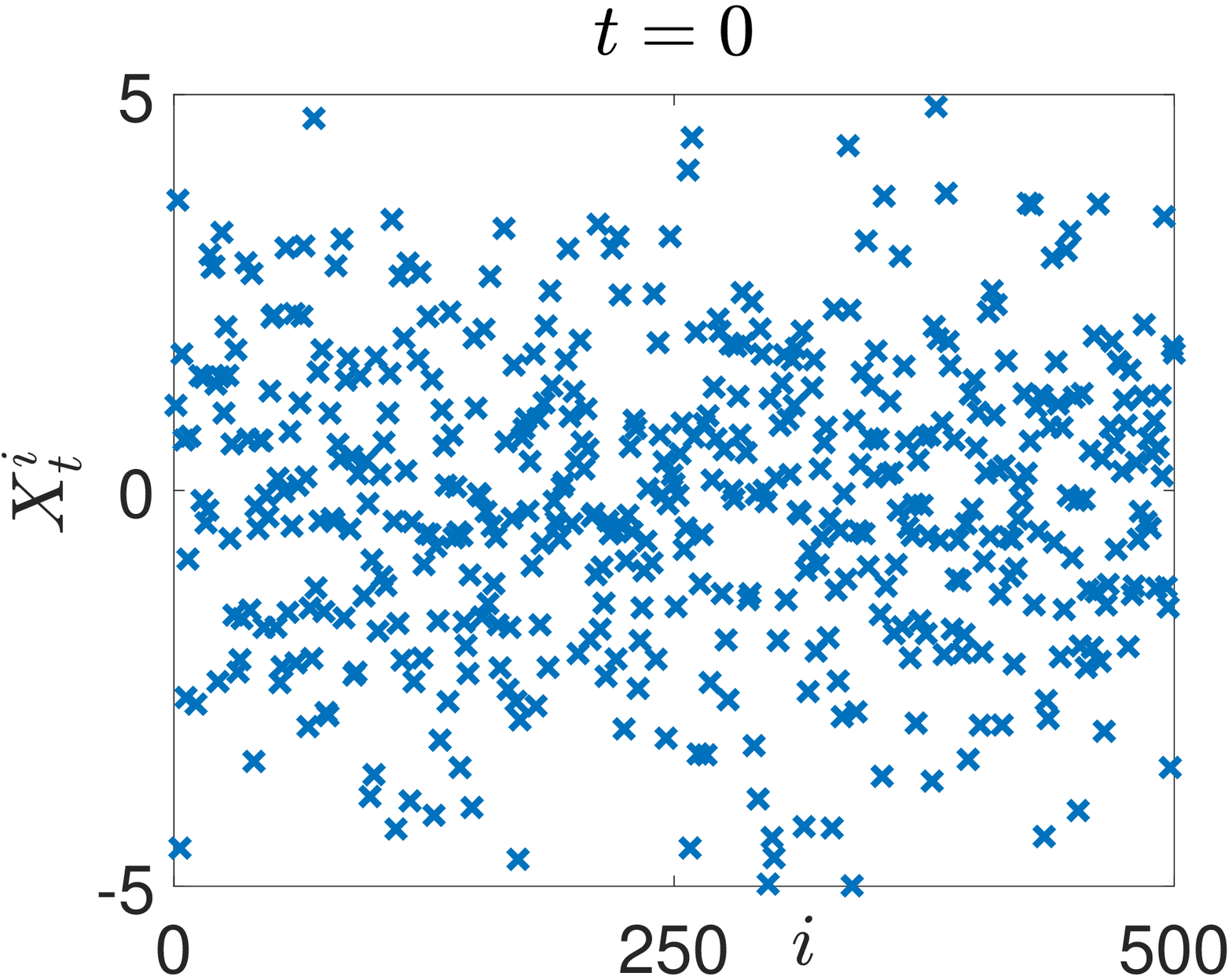}
	\includegraphics[width=0.550\linewidth]{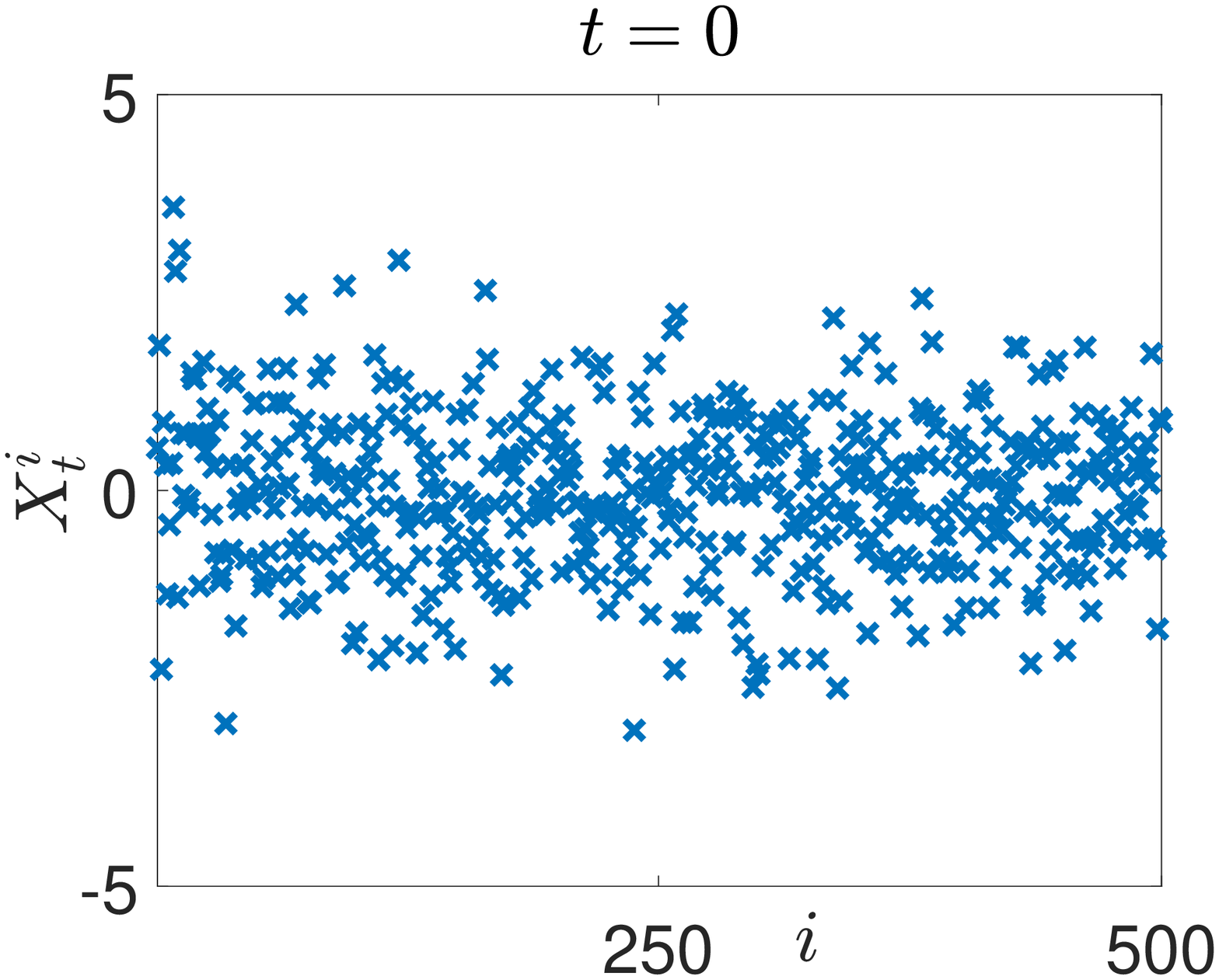}

\hskip-0.75cm
	\includegraphics[width=0.550\linewidth]{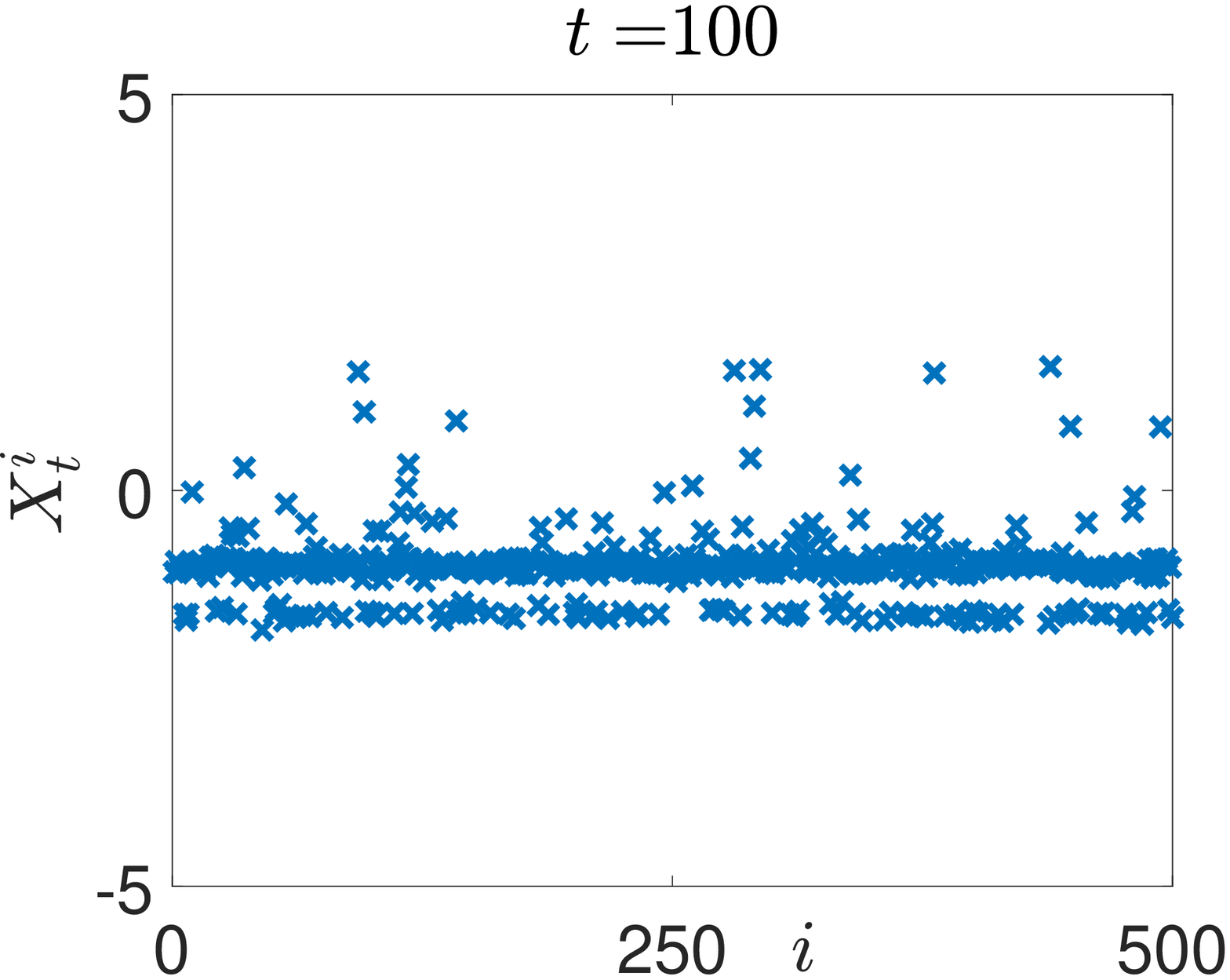}
	\includegraphics[width=0.550\linewidth]{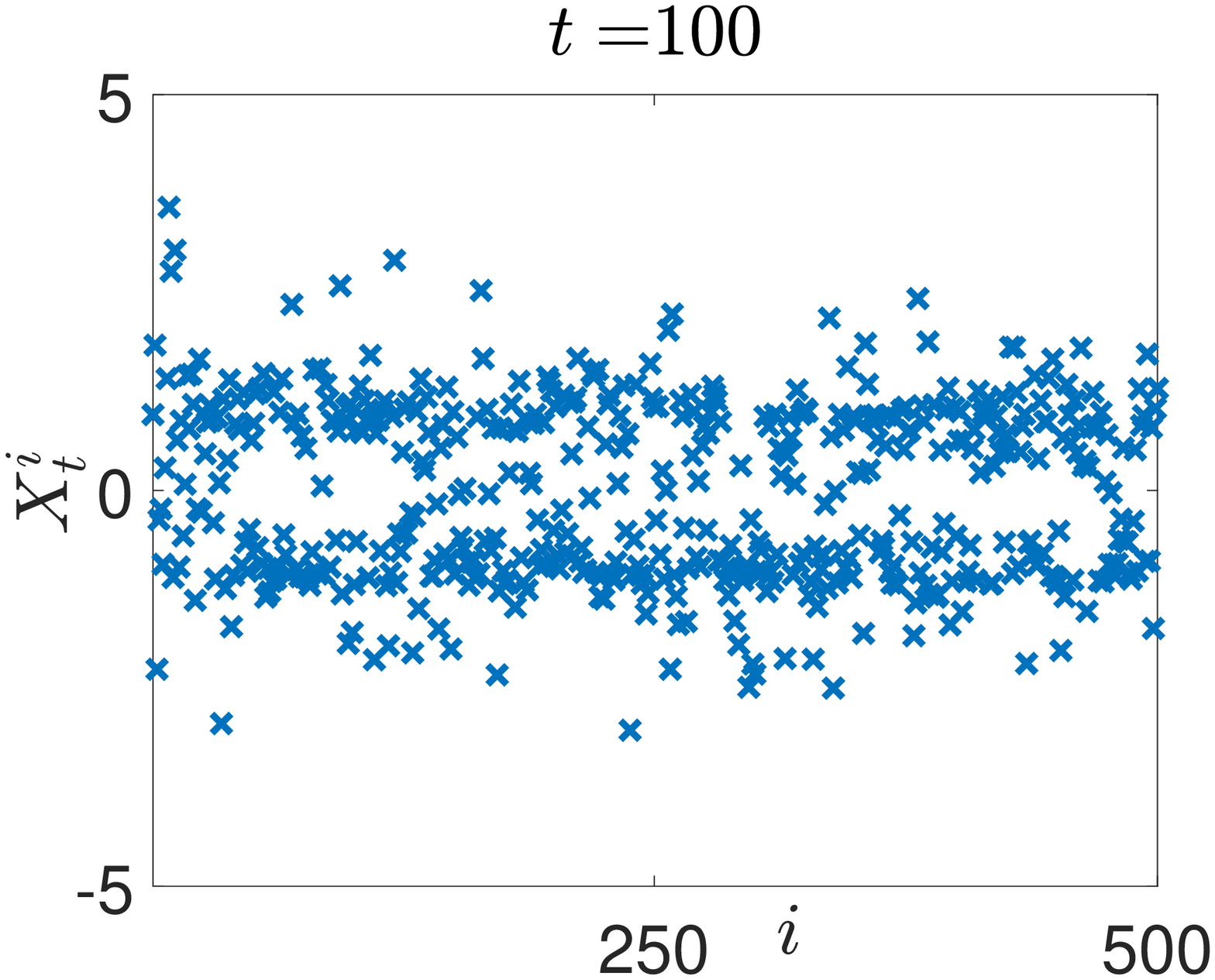}

\hskip-0.75cm
	\includegraphics[width=0.550\linewidth]{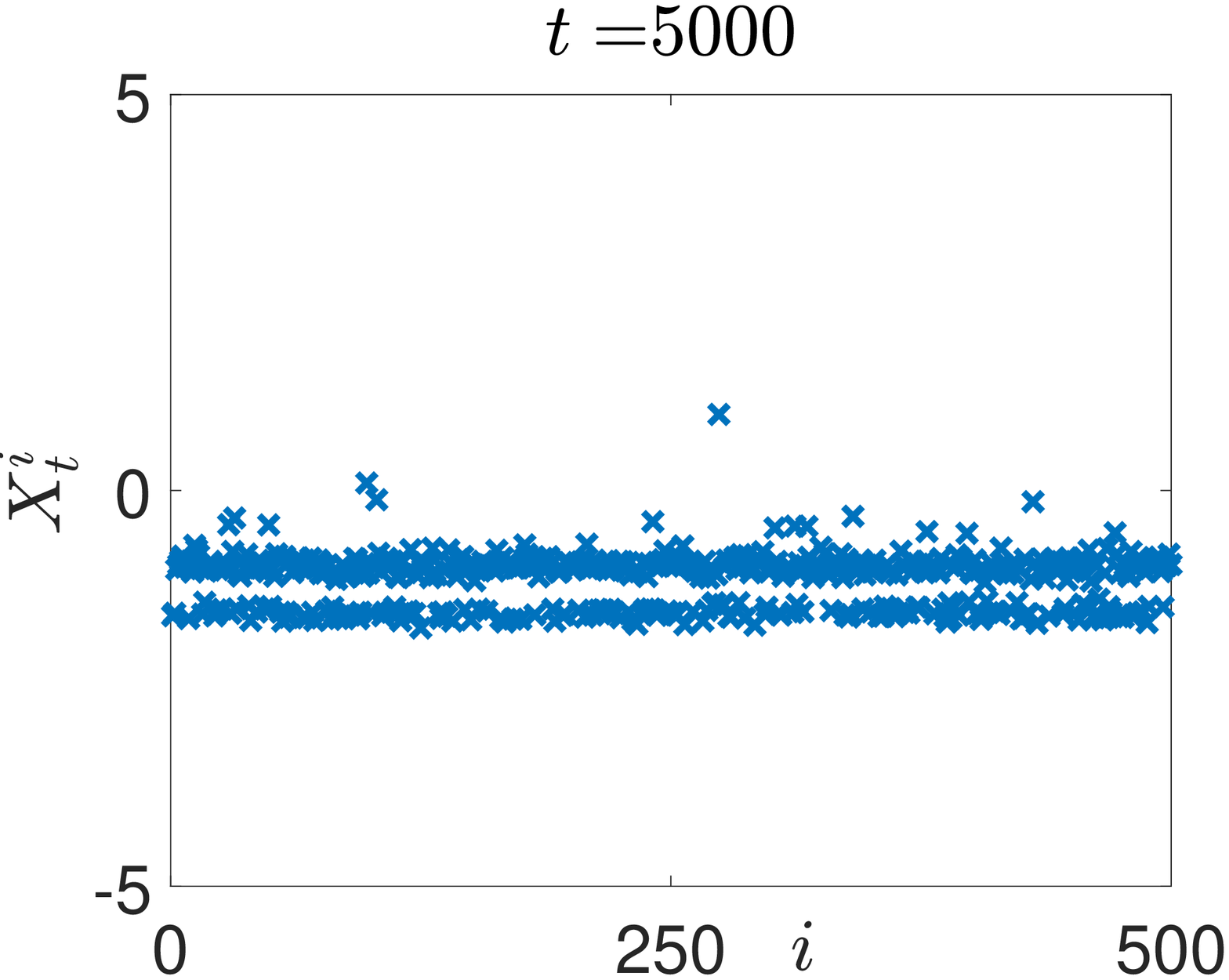}
	\includegraphics[width=0.550\linewidth]{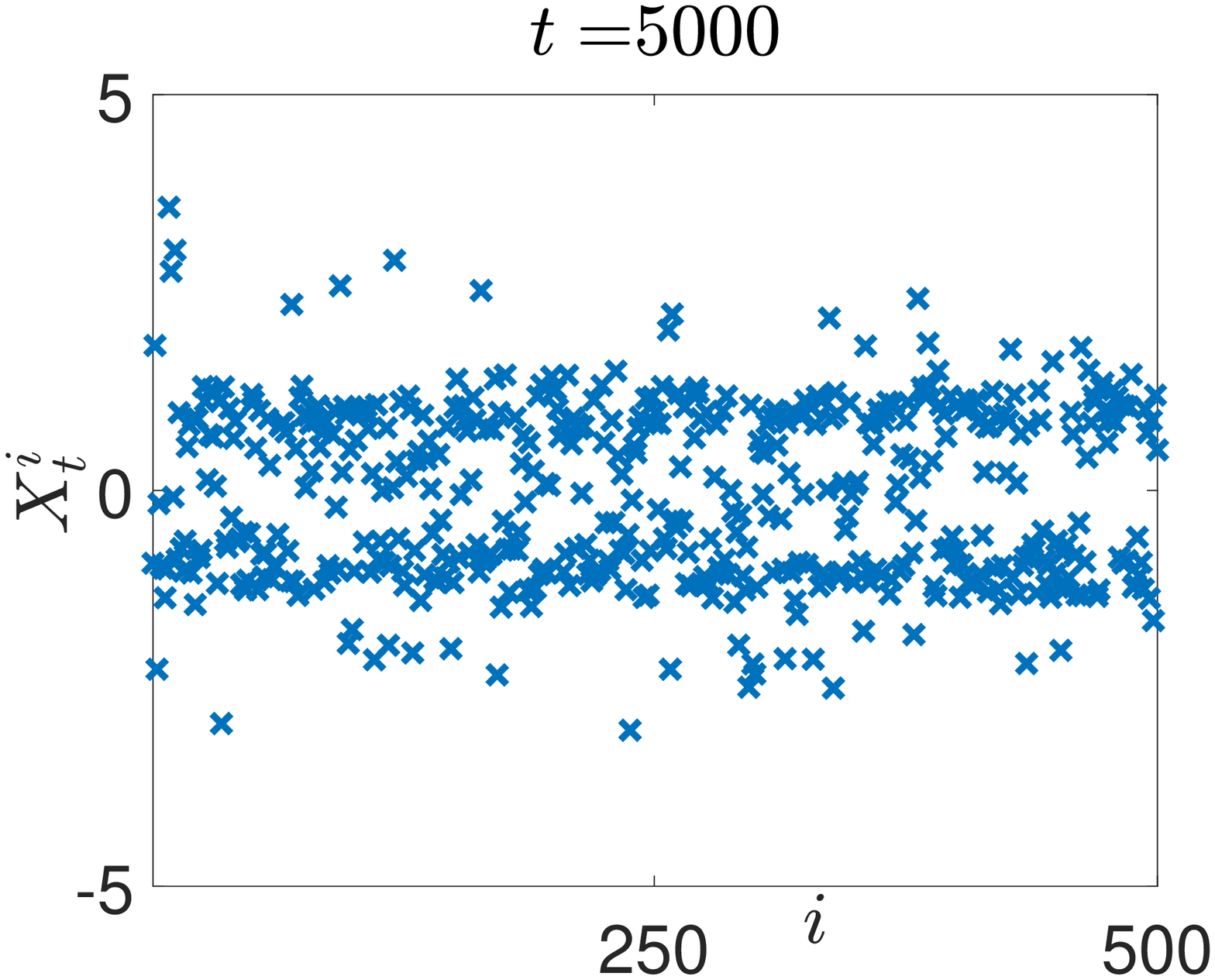}
\caption{Position of $N=500$ particles  for $V^{\eps}(x) = \frac{x^4}{4} - \frac{x^2}{2}\left(1- \delta \cos \left(\frac{x}{\eps} \right)\right)$, with $\theta = 0.5$, $\beta \approx 5.6$, $\delta = 1$. Left: Eqn.~\eqref{eq:system_of_sdes_in_1D} with $\epsilon = 0.1$. Right: homogenized SDEs~\eqref{eq:system_of_homogenized_sdes}. }
	 \label{fig:bistable_mult_position}
 \end{figure}
Here we can observe the noncommutativity of the limits: the particles evolve towards different steady states, which shows the effect of the fluctuations on the critical temperature $\beta_C$ at which phase transitions occur. 
This will be confirmed below when we present the mean value of the solution. 

We present in Figure~\ref{fig:bistable_mult_histogram} snapshots of the histogram of the $N=500$ particles at $t=0, \, t=100$ and $t=5,000$. 
Again, we observe that the particles converge to different equilibria, the homogenized system converging to a mean zero distribution with peaks at $1$ and $-1$, while for positive values of the parameter $\epsilon$ the system converges to a distribution with $\bar{X}_t=-1$. 
Similarly to the previous case, we superpose the corresponding invariant measure, rescaled for comparison, for this parameter regime on the $t=5,000$  snapshot, 
and again we observe that the empirical density of the system of interacting diffusions converges to the steady state solution computed by solving the stationary McKean-Vlasov equation, which is also obtained by time-evolution of the Fokker-Planck equation (see Figures~\ref{fig:convex_add_FP} and~\ref{fig:bistable_mult_FP} in Section~\ref{sec:evolution}).
\begin{figure}[h!]
\hskip-0.75cm
	\includegraphics[width=0.550\linewidth]{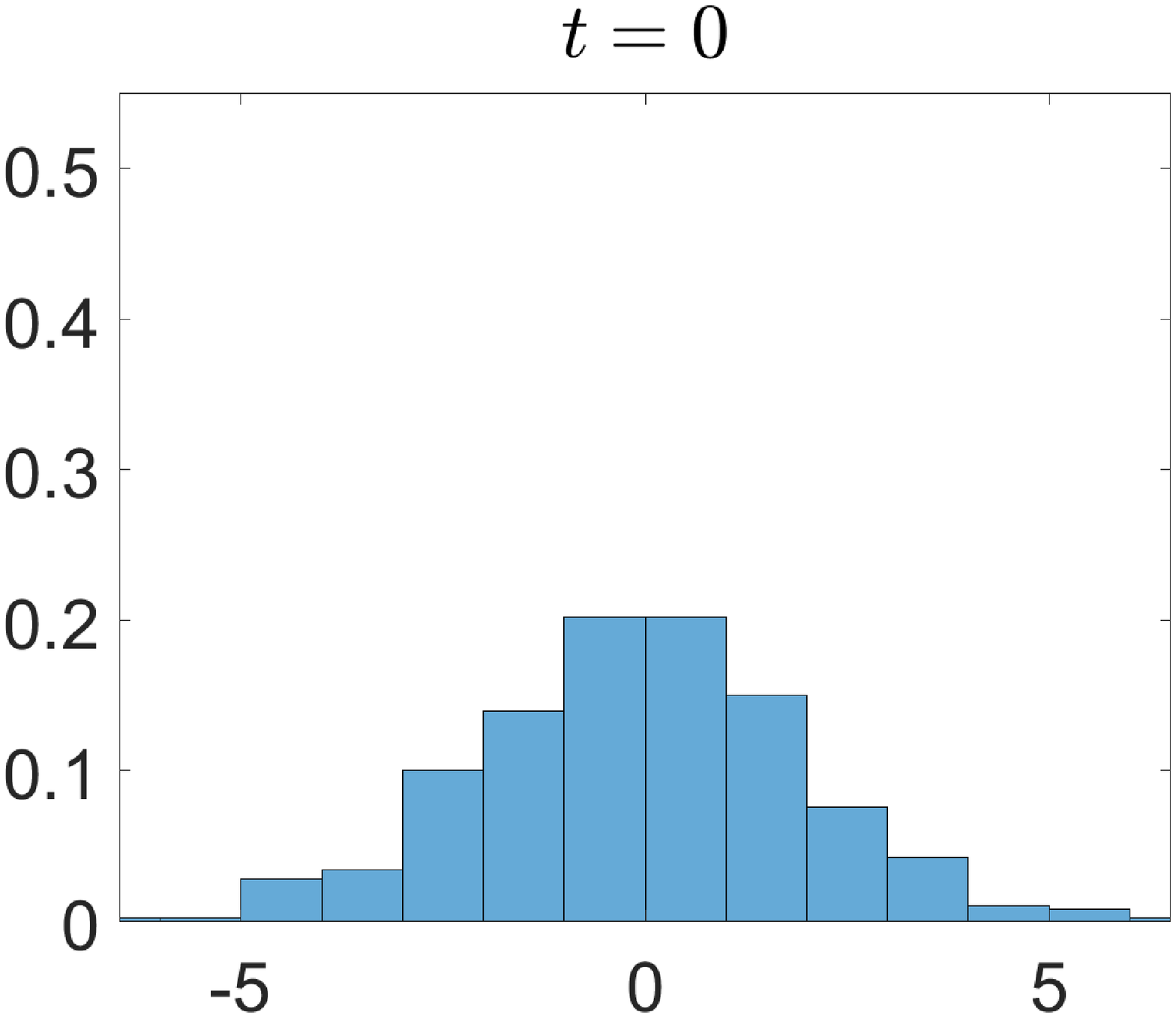}
	\includegraphics[width=0.550\linewidth]{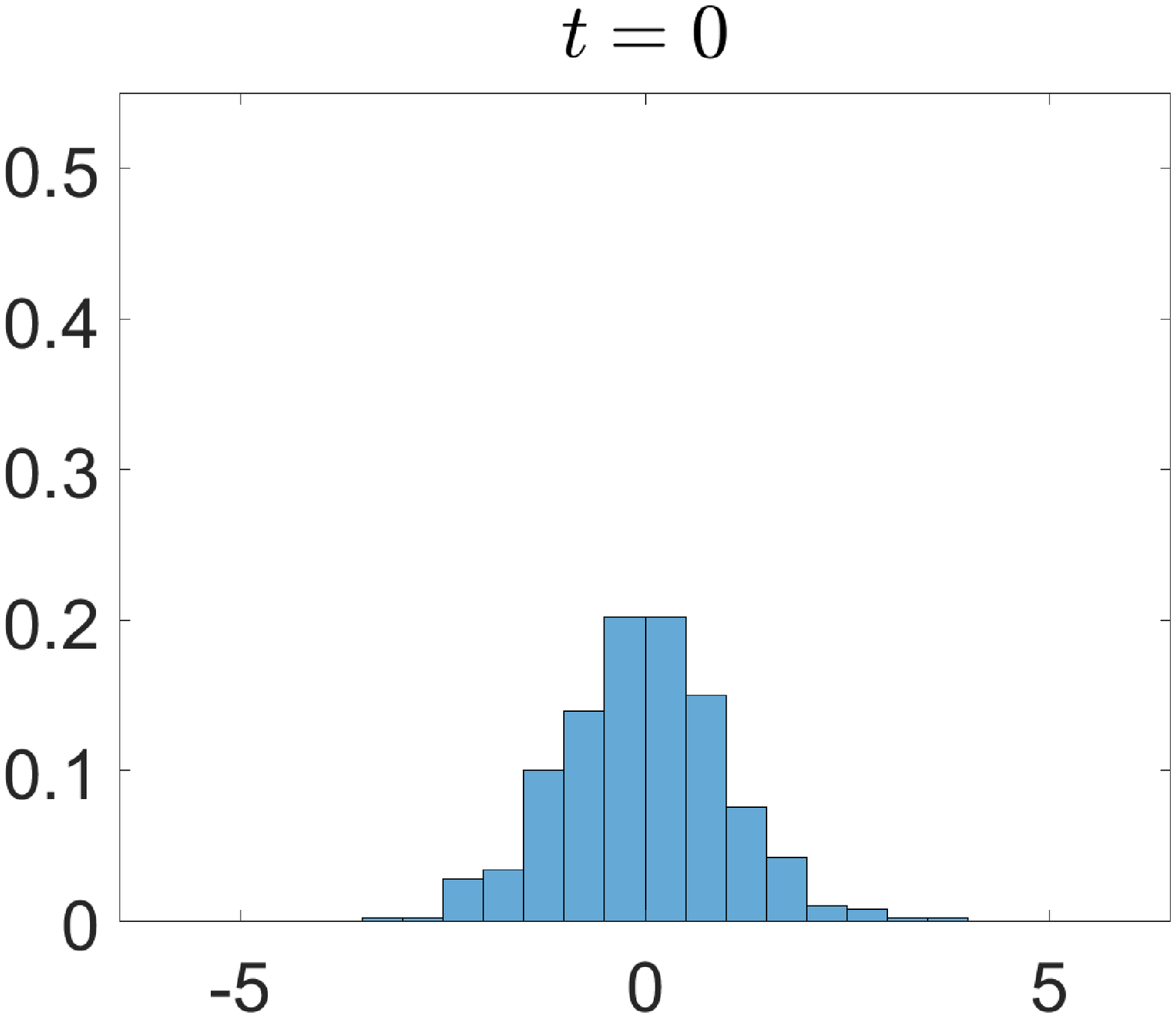}

\hskip-0.75cm
	\includegraphics[width=0.550\linewidth]{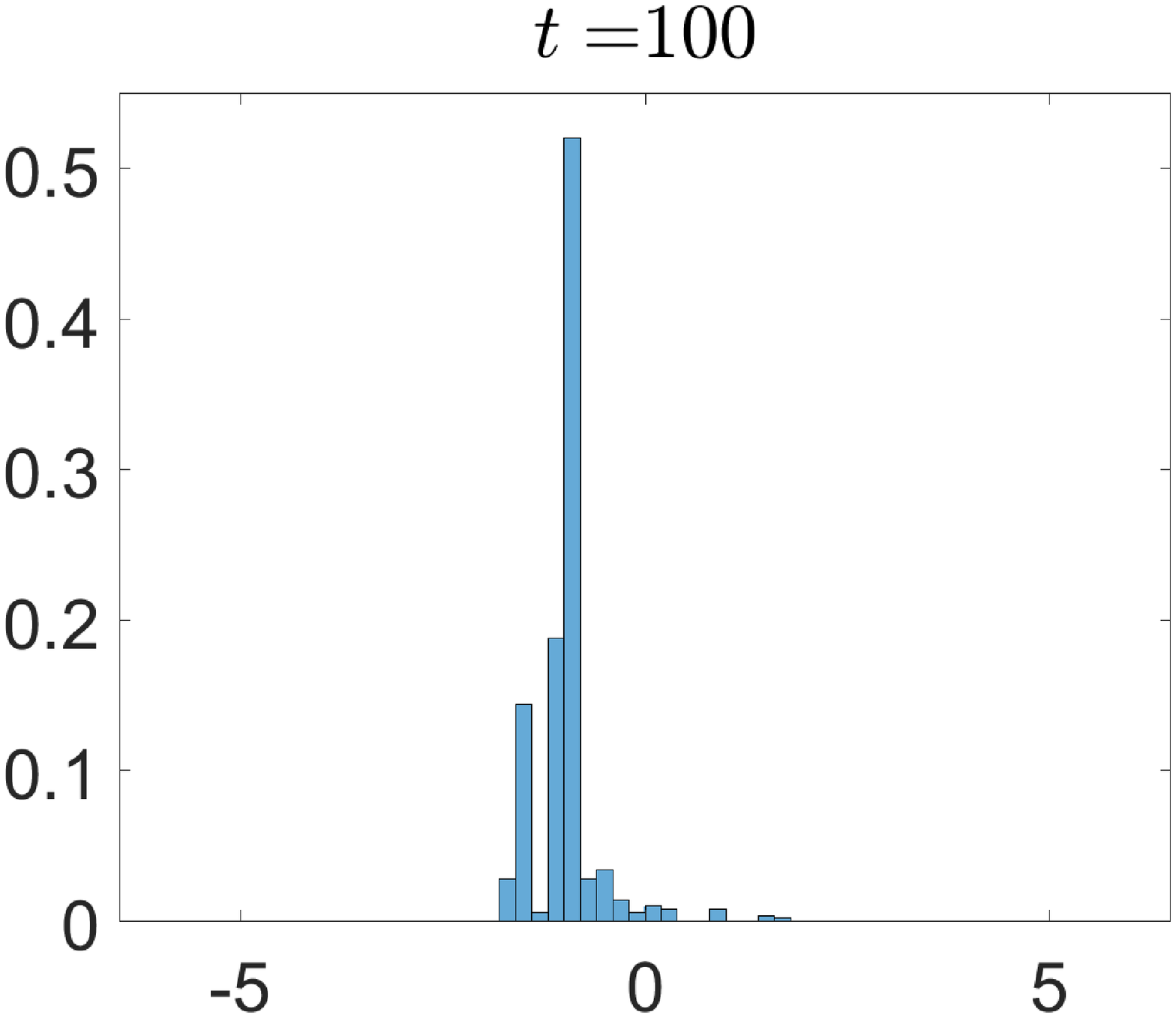}
	\includegraphics[width=0.550\linewidth]{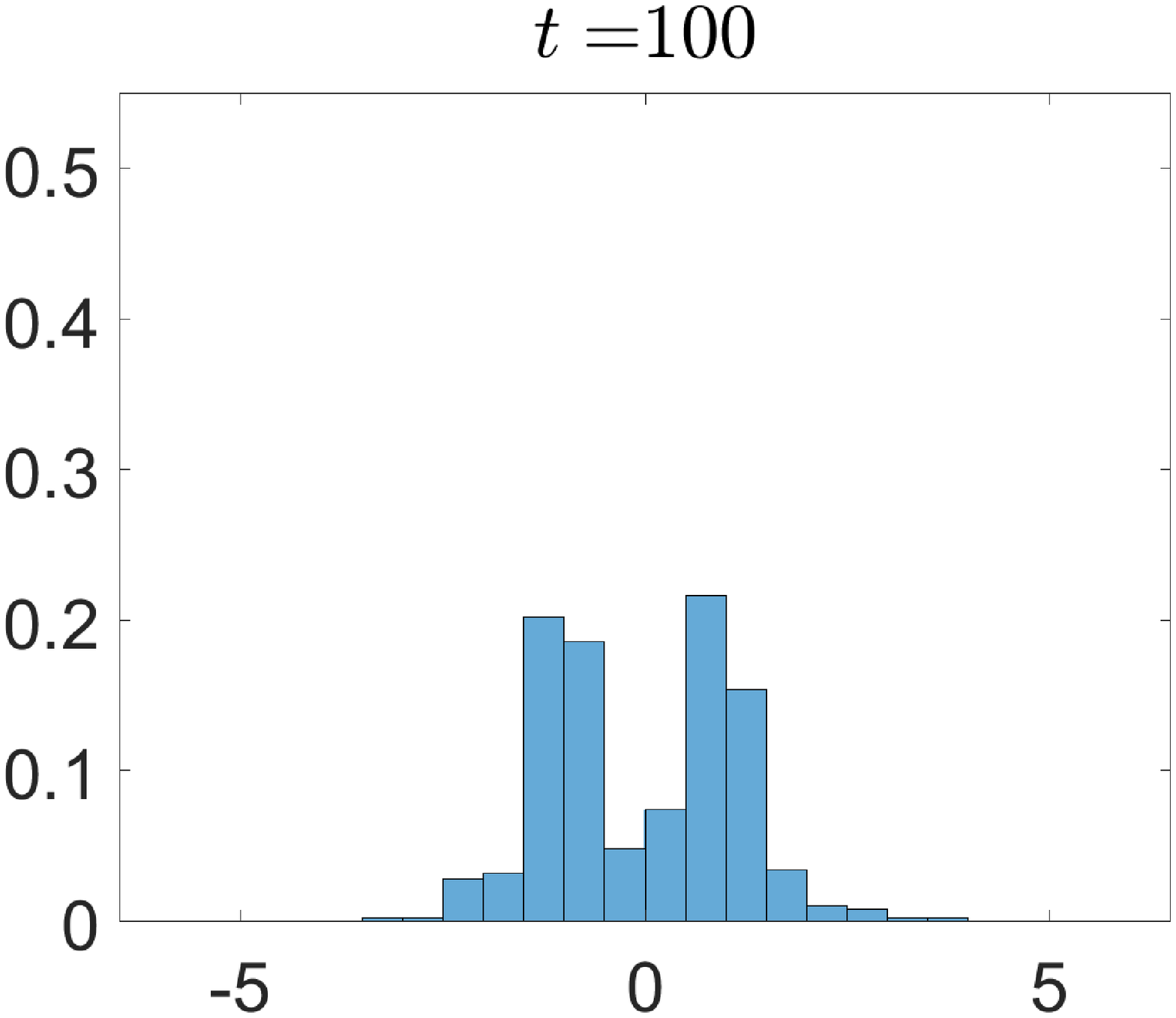}

\hskip-0.75cm
	\includegraphics[width=0.550\linewidth]{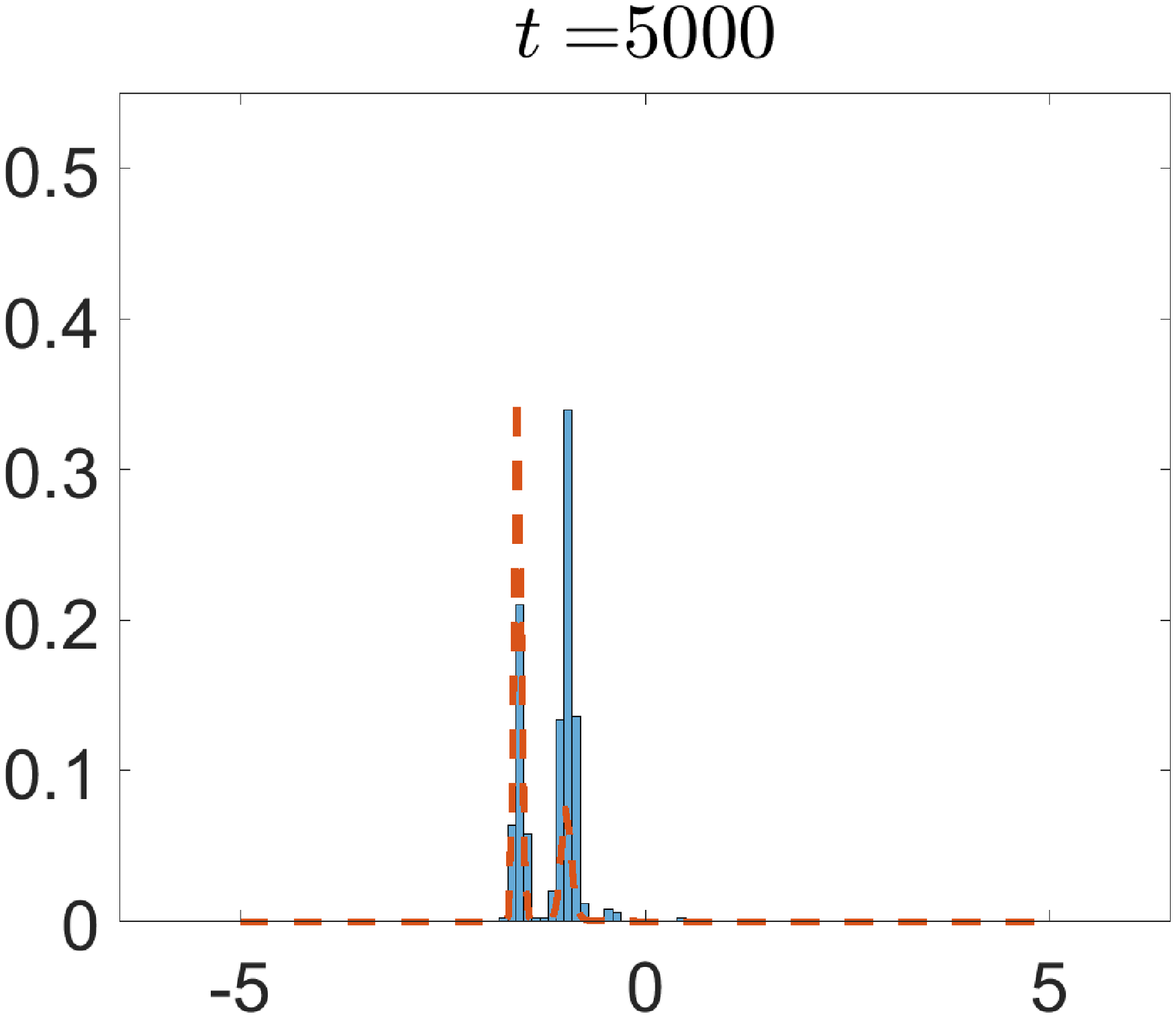}
	\includegraphics[width=0.550\linewidth]{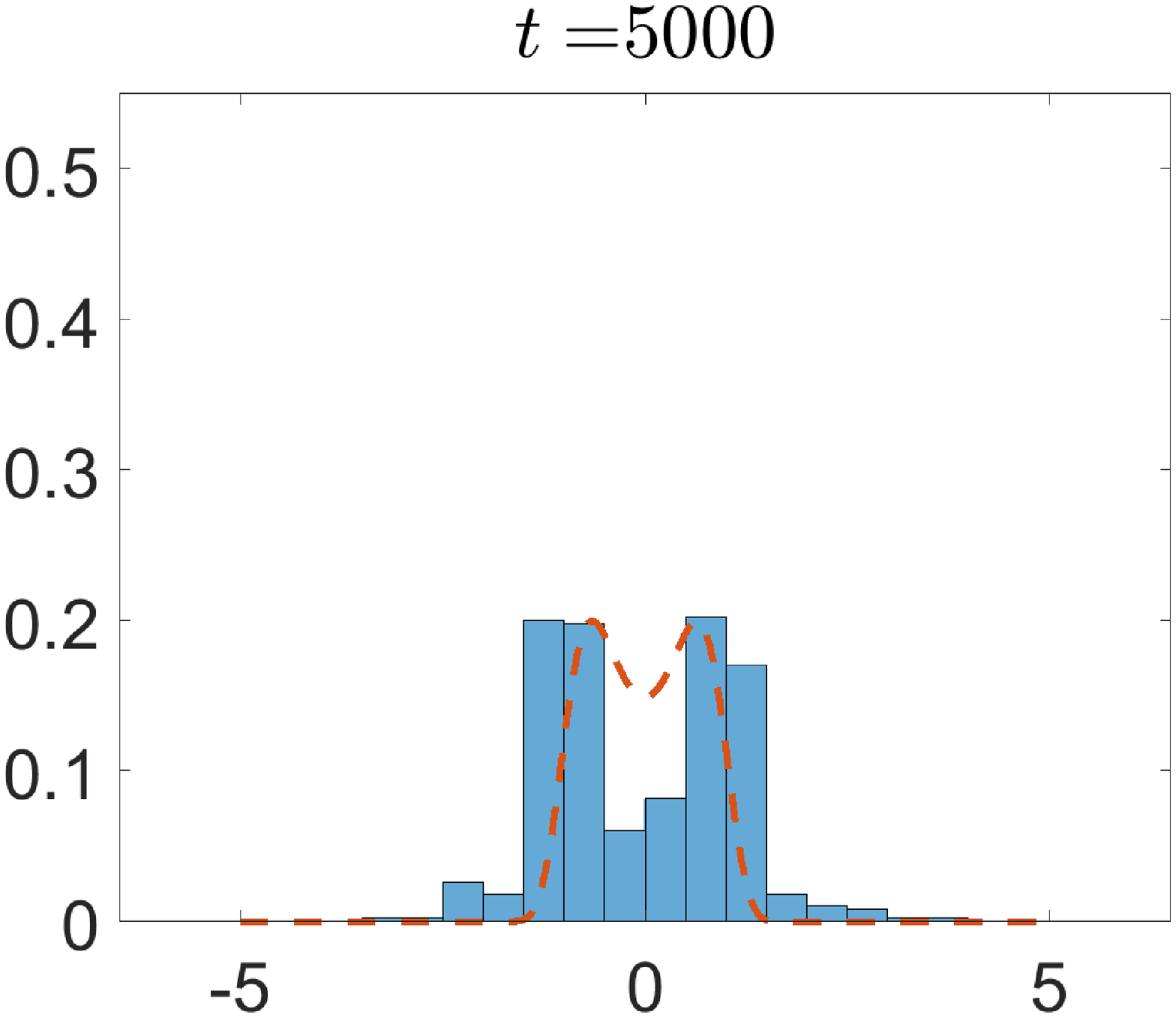}
	\caption{Histogram of $N=500$ particles  for $V^{\eps}(x) = \frac{x^4}{4} - \frac{x^2}{2}\left(1- \delta \cos \left(\frac{x}{\eps} \right)\right)$, with $\theta = 0.5$, $\beta \approx 5.6$, $\delta = 1$. Left: Eqn.~\eqref{eq:system_of_sdes_in_1D} with $\epsilon = 0.1$. Right: homogenized SDEs~\eqref{eq:system_of_homogenized_sdes}. }
	 \label{fig:bistable_mult_histogram}
 \end{figure}

Finally, we plot in Figure~\ref{fig:bistable_mult_mean} the average $\bar{X}_t^N$ of the $N=500$ particles for the case of a bistable large-scale potential with nonseparable fluctuations. We observe here that the critical temperature for the homogenized dynamics is different than that for the full dynamics. In particular,  the phase transition occurs for $\beta\approx 10.4 > 5.6$ for the homogenized problem, while for finite values of $\epsilon$ there already exist several branches at this value of $\beta$.

\begin{figure}[h!]
\hskip-0.75cm
	\includegraphics[width=0.550\linewidth]{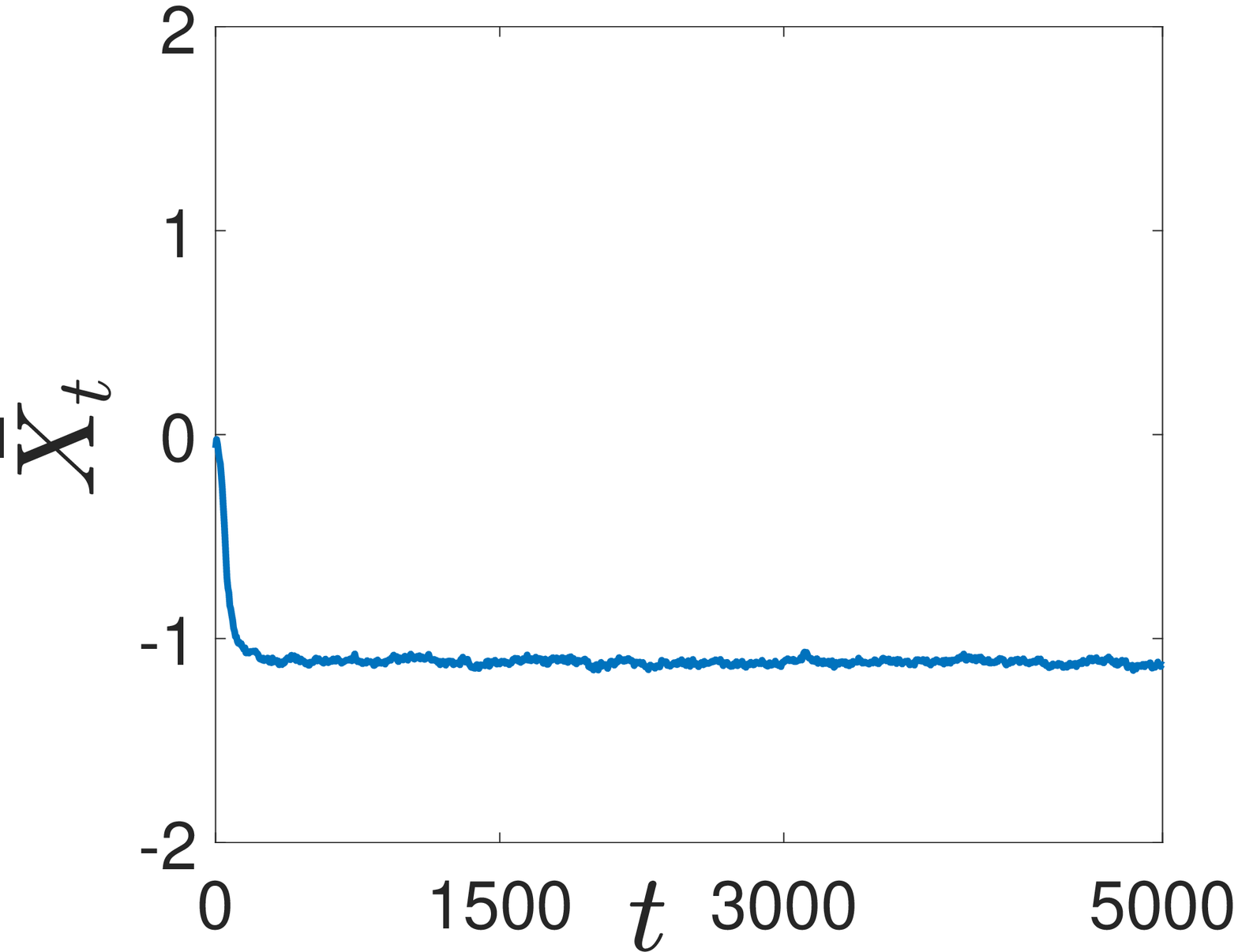}
	\includegraphics[width=0.550\linewidth]{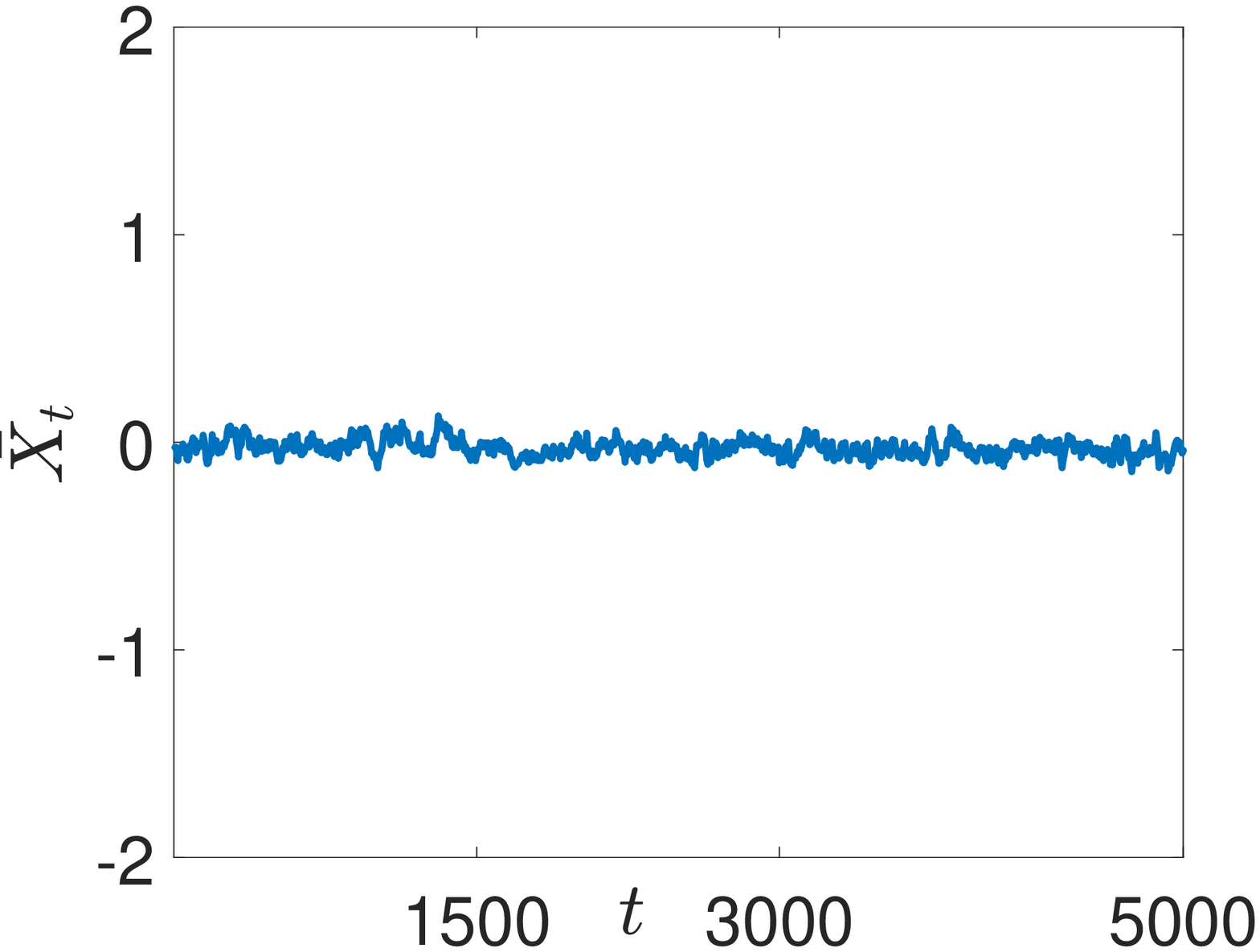}
	\caption{Time evolution of the average $\bar{X}_t^N = \frac{1}{N}\sum_{i=1}^N X_t^i$ of $N=500$ particles  for $V^{\eps}(x) = \frac{x^4}{4} - \frac{x^2}{2}\left(1- \delta \cos \left(\frac{x}{\eps} \right)\right)$, with $\theta = 0.5$, $\beta \approx 5.6$, $\delta = 1$. Left: Eqn.~\eqref{eq:system_of_sdes_in_1D} with $\epsilon = 0.1$. Right: homogenized SDEs~\eqref{eq:system_of_homogenized_sdes}.  }
	 \label{fig:bistable_mult_mean}
 \end{figure}

\subsection{Time dependent McKean-Vlasov evolution}
\label{sec:evolution}
We performed time dependent simulations of the evolution of the nonlinear McKean-Vlasov equation both for the full and for the homogenized dynamics. We present below the results corresponding to the cases presented for the Monte-Carlo simulations.

We recall that, for the case when we take $N\rightarrow\infty$ first while keeping $\epsilon >0$ fixed, the McKean-Vlasov-Fokker-Planck equation that we need to solve is
\begin{equation}
\label{eq:FP_for_multiscale_s4}
\frac{\partial p}{\partial t} = \frac{\partial}{\partial x}\left(\beta^{-1}\frac{\partial p}{\partial x} + \partial_x V^\epsilon(x) p + \theta \left(x - \int x p(x,t) \, dx \right) p\right),
\end{equation}
whereas for the case when we first homogenize the dynamics and then pass to the mean field limit the McKean-Vlasov equation becomes
\begin{equation}
\label{eq:FP_for_homogenized_s4}
\frac{\partial p}{\partial t} =\frac{\partial}{\partial x}\left[ \beta^{-1}\frac{\partial\left(\M(x) p\right)}{\partial x} + \M(x)\left( V'_0(x) +\psi'(x) + \theta \left(x - \int x p(x,t) \, dx \right) \right)p  + \beta^{-1}\frac{\partial \M(x)}{\partial x}p\right],
\end{equation}
with $\psi(x)$ and $\M(x)$ given by~\eqref{eq:small_psi} and~\eqref{eq:Diffusion_coeff}, respectively.

To solve the the McKean-Vlasov evolution PDE we used the positivity preserving, entropy decreasing finite volume scheme from~\cite{Carrillo2015}. We point out that this scheme solves the equations using no-flux boundary conditions. We use these boundary conditions and a sufficiently large domain. We used the same initial conditions for the time-dependent Fokker-Planck simulations as the ones used for the Monte-Carlo simulations, i.e., the initial condition was the PDF for a normal distribution with mean zero and variance 4. However, for the bistable large-scale potential with nonseparable fluctuations in the finite but positive $\epsilon$ case -- see left panel on Figure~\ref{fig:bistable_mult_FP} --  we needed to use a different initial condition: here we used a normal distribution with mean $-0.1$ and variance $4$. This is likely because the value of $\beta$ we chose here was close to the bifurcation point and the mean-zero solution was still being picked up on the time evolution.

 We present below the results for the case of a convex large-scale potential $V_0^c$ with separable fluctuations -- the same case presented in Figures~\ref{fig:convex_add_position}-~\ref{fig:convex_add_mean}. The parameters used were $\theta = 2$, $\beta = 8$, $\delta = 1$, $\epsilon = 0.1$. 
\begin{figure}[h!]
\hskip-0.75cm
	\includegraphics[width=0.550\linewidth]{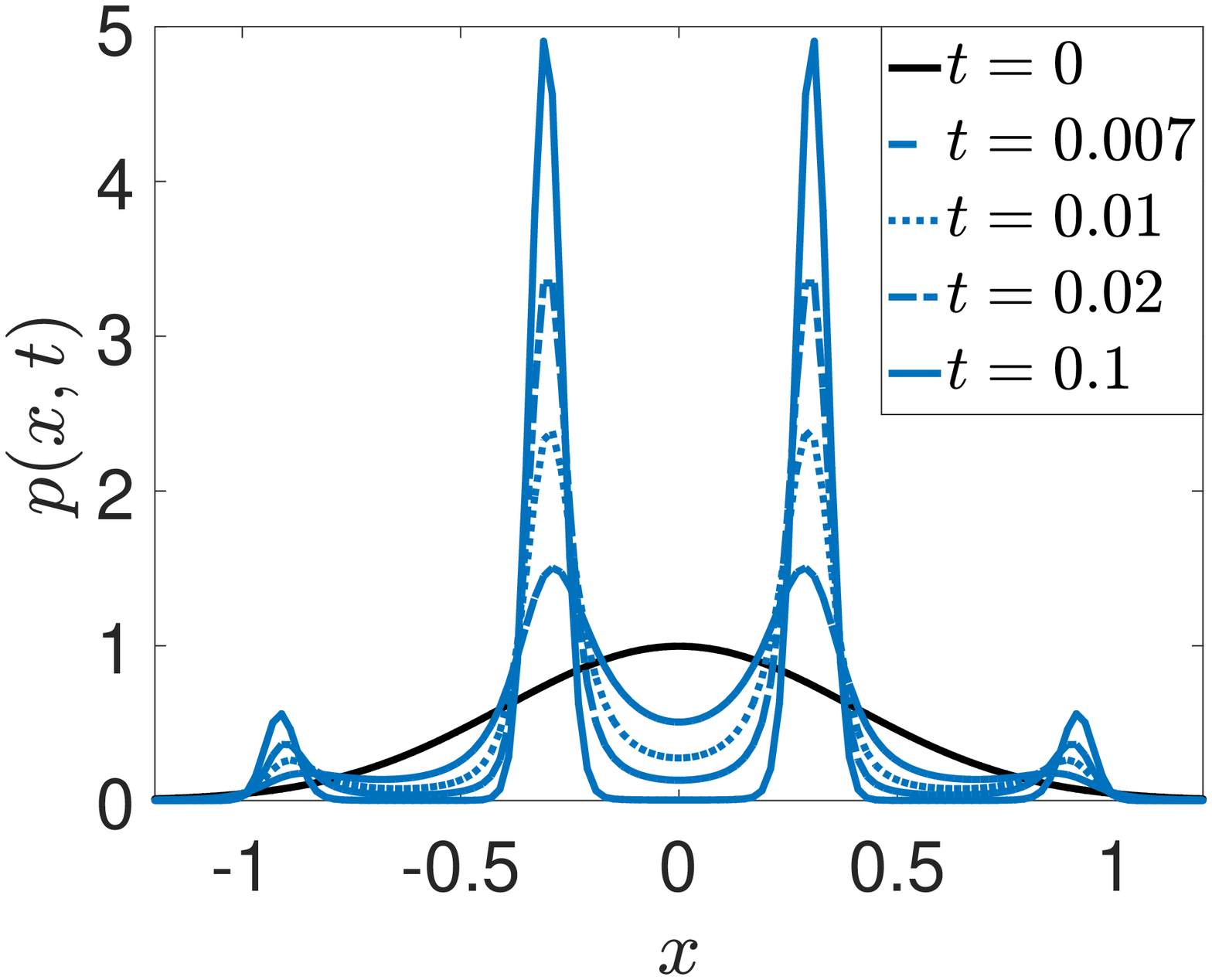}
	\includegraphics[width=0.550\linewidth]{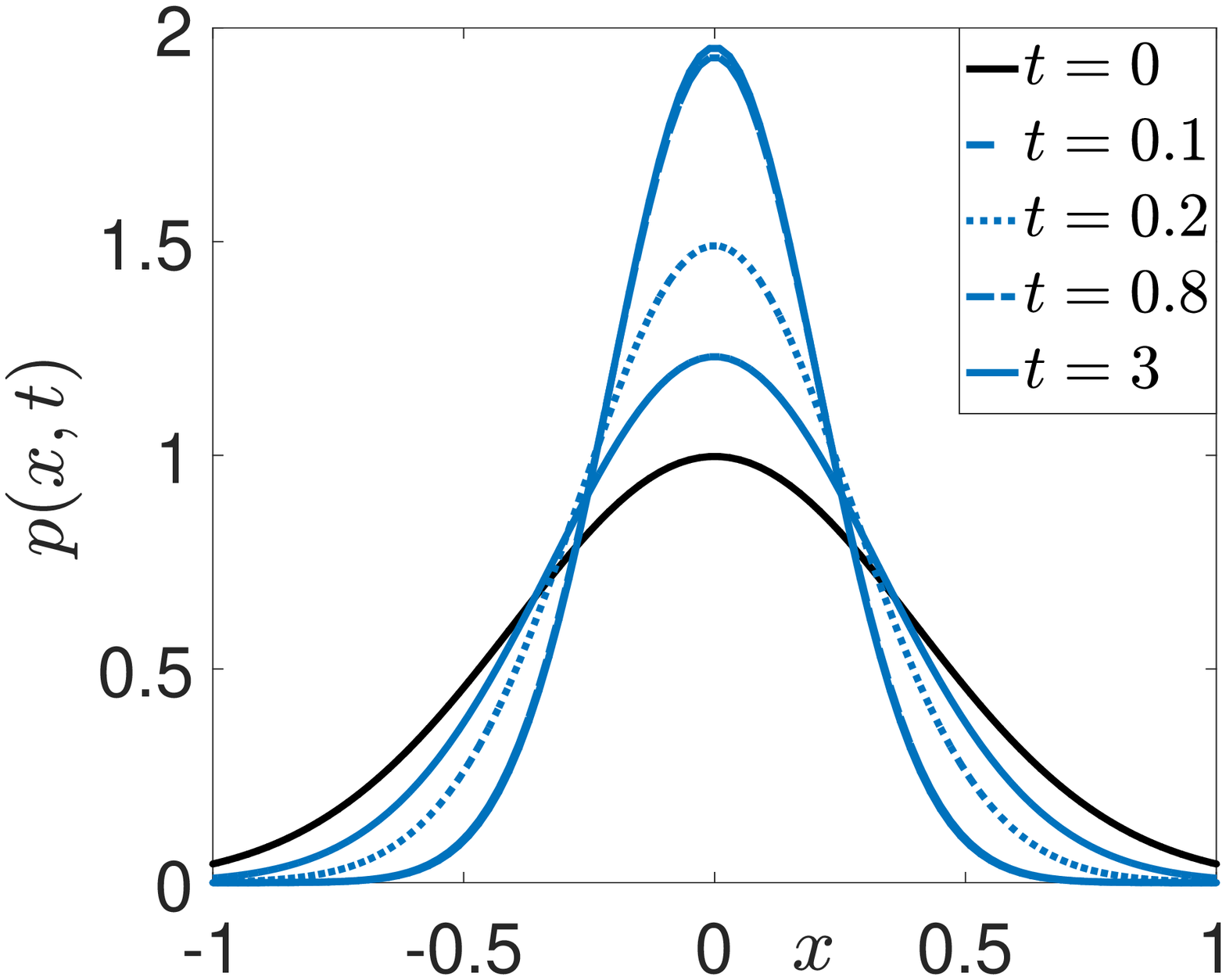}
	\caption{Time evolution of the McKean-Vlasov equation for $V^\eps(x) = \frac{x^2}{2} + \delta\cos\left(\frac{x}{\eps}\right)$with $\theta = 2$, $\beta =8$, $\delta = 1$. Left: \eqref{eq:FP_for_multiscale_s4} with $\epsilon = 0.1$. Right: homogenized equation~\eqref{eq:FP_for_homogenized_s4}. }
	 \label{fig:convex_add_FP}
 \end{figure}

As expected, the results obtained by solving the time dependent McKean-Vlasov equation are in agreement with the results obtained from the Monte Carlo simulations and from solving the stationary McKean-Vlasov equation--i.e. the selfconsitency equation. We note that, similarly to what we observed in the solution of the system of interacting particles, the solution to the McKean-Vlasov equation converges to its steady state faster for the full dynamics than for the homogenized equation. This observation can be quantified by comparing the convergence rates in the weighted $L^2$ or relative entropy exponential estimates, in particular by comparing the constants in the Poincar\'{e} and logarithmic Sobolev inequalities for the full and for the homogenized dynamics. A preliminary study of this--for the Fokker-Planck operator of the finite dimensional dynamics--was presented in~\cite{DuncanPavliotis2016}. 

Finally, we present numerical results for the case of a bistable large-scale potential $V_0^b$ with nonseparable fluctuations -- the same case presented in Figures~\ref{fig:bistable_mult_position}-\ref{fig:bistable_mult_mean}. The parameters used were $\theta = 0.5$, $\beta \approx 5.6$, $\delta = 1$, $\epsilon = 0.1$.

\begin{figure}[h!]
\hskip-0.75cm
	\includegraphics[width=0.550\linewidth]{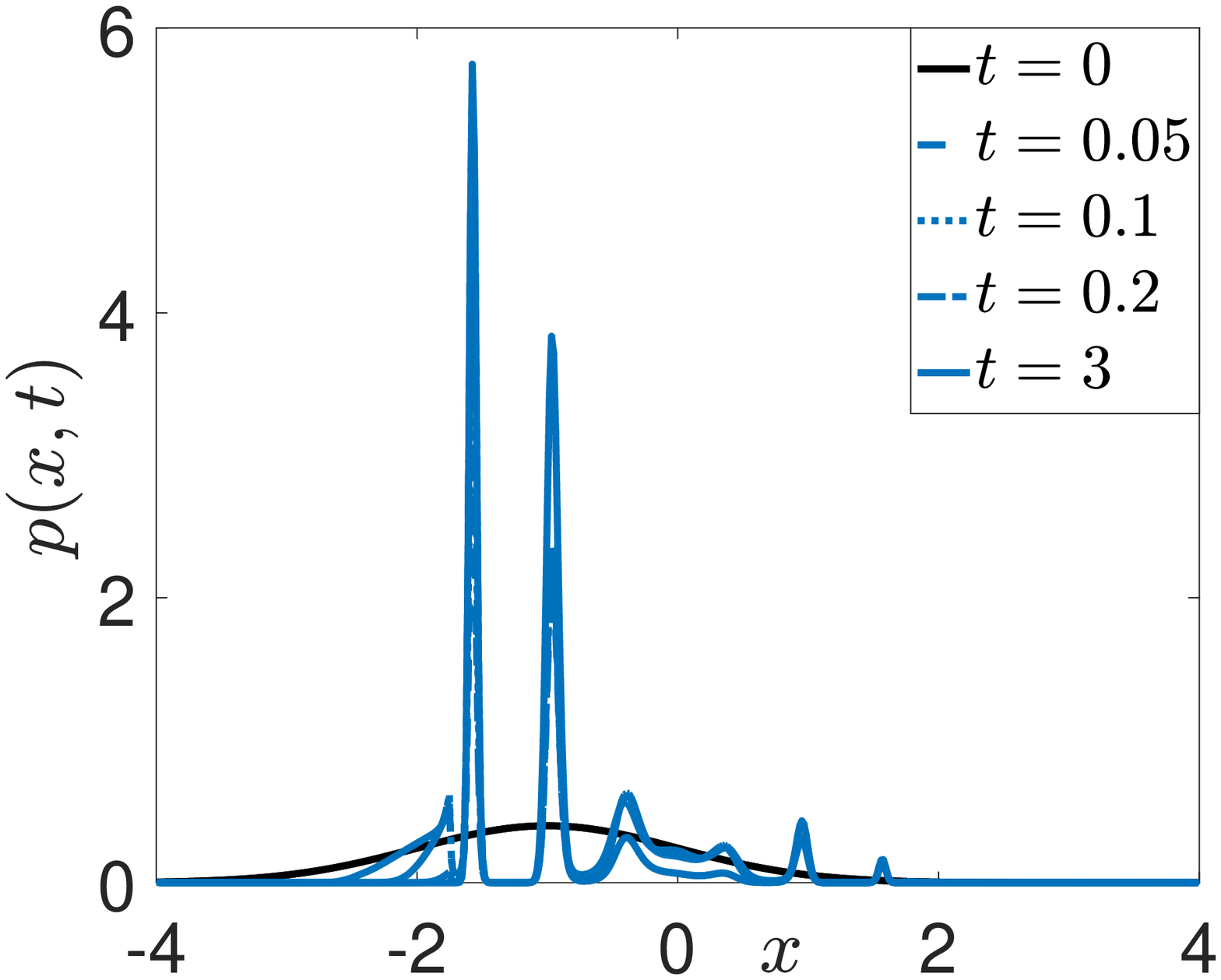}
	\includegraphics[width=0.550\linewidth]{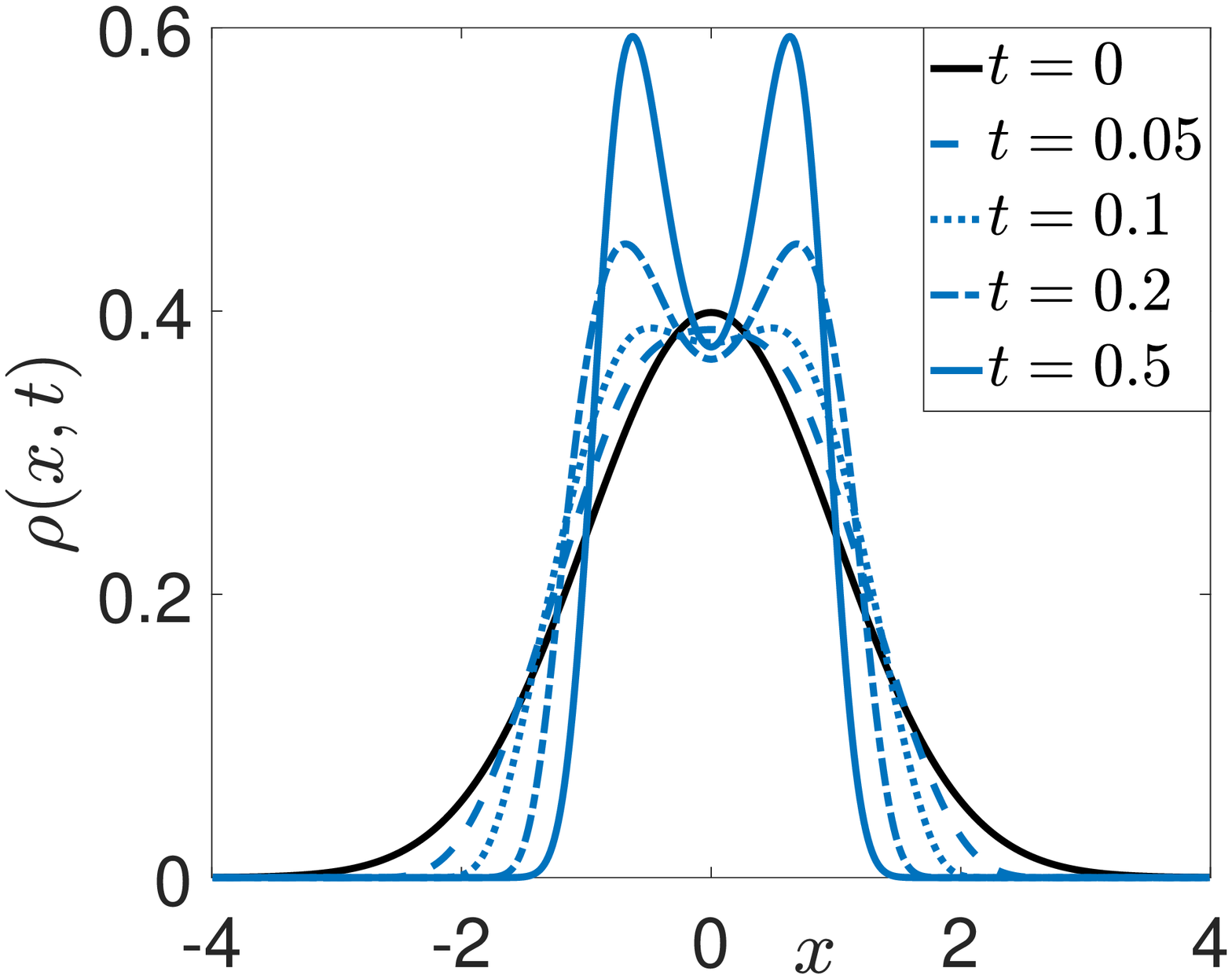}
	\caption{Time evolution of the McKean-Vlasov equation for $V^\eps(x) = \frac{x^4}{4} -\frac{x^2}{2}\left(1- \delta\chi_{[-5,5]}(x)\cos\left(\frac{x}{\eps}\right)\right)$ with $\theta = 0.5$, $\beta \approx 5.6$, $\delta = 1$. Left: \eqref{eq:FP_for_multiscale_s4} with $\epsilon = 0.1$. Right: homogenized equation~\eqref{eq:FP_for_homogenized_s4}. }
	 \label{fig:bistable_mult_FP}
 \end{figure}

As expected, the solutions converge to those computed by solving the stationary McKean-Vlasov equation and are qualitatively similar to those obtained from the particle system simulations, see Figure~\ref{fig:bistable_mult_histogram}. In this case, the solution of the time dependent McKean-Vlasov PDE converges to a steady state slower for the full dynamics, in comparison to the homogenized dynamics. We believe that this is related to the phenomenon of critical slowing down~\cite{shiino1987} when the dynamics is close to a bifurcation, since the inverse temperature $\beta^{-1}$ that we use for the simulations is close to the critical temperature $\beta_C^{-1}$ for the full dynamics.

\section{Conclusions and Further Work}
\label{sec:conclusions}

The combined mean field and homogenization limit for a system of interacting diffusions in a two-scale confining potential was studied in this paper. In particular, the homogenized McKean-Vlasov equation was obtained and studied and the bifurcation diagram for the stationary states was considered.  It was shown, by means of analysis and extensive numerical simulations, that the homogenization and mean field limits, at the level of the bifurcation diagram (i.e. when combined with the long time limit) do not commute for nonseparable two-scale potentials. Furthermore, it was shown that the bifurcation diagrams can be completely different for small but finite $\epsilon$ and for the homogenized McKean-Vlasov equation. 

It should be emphasized, as is clearly explained in~\cite{ChayesPanferov2010}, see in particular the remarks at the end of Sec. 2 of this paper, that the connection between  bifurcations and phase transitions for the McKean-Vlasov dynamics is not entirely straightforward. In particular, in order for a bifurcation point to correspond to a genuine phase transition, it is not sufficient to have the emergence of a new branch of solutions, but these emergent solutions should have a lower free energy. More precisely, it was shown in~\cite{ChayesPanferov2010} for the McKean-Vlasov dynamics on the torus and with a finite-range interaction potential, that the loss of linear stability of the uniform state -- which corresponds to the mean-zero Gibbs state in our setting -- does not imply a second order phase transition. Furthermore, the critical temperature (or, equivalently, critical interaction strength) at which first order phase transitions occur, is lower than the temperature at which the pitchfork bifurcation happens. For the problem that we studied, supercritical pitchfork bifurcations occur which correspond to second order (continuous) phase transitions. On the other hand, when only saddle node bifurcations are present, e.g. in Figure~\ref{fig:conv_mult_bif}, then the mean-zero solution is still the global minimizer of the free energy, see Figure~\ref{fig:conv_mult_FEf}. In particular, no first order phase transitions seem to appear in the McKean-Vlasov model that we studied in this work.

There are many open questions that are not addressed in this work. First, the rigorous multiscale analysis for the McKean-Vlasov equation in locally periodic potentials needs to be carried out. Perhaps more importantly, the rigorous construction of the bifurcation diagram in the presence of infinitely many local minima in the confining potential, thus extending the results presented in e.g.~\cite{Dawson1983, Tamura1984, Tugaut2014} appears to be completely open.   Furthermore, the study of the stability of stationary  solutions to the McKean-Vlasov equation in the presence of a multiscale structure, as well as the analysis of the problem of convergence to equilibrium in this setting is an intriguing question. Finally, the extension of the work presented in this paper to higher dimensions presents additional challenges. We mention, for example, that the corresponding nonlinear diffusion process does not have to be reversible~\cite{LelievreNierPavliotis2013, DuncanLelievrePavliotis2016}.  We believe that the results reported in this work open up a new exciting avenue of research in the study of mean field limits for interacting diffusions in the presence of many local minima, with potentially interesting applications to the study of McKean-Vlasov based mathematical models in the social sciences.

\section{Acknowledgements}
The authors are grateful to S. Kalliadasis, J.A. Carrillo, A.O. Parry and Ch. Kuehn  for useful discussions. They are particularly grateful to J.A. Carrillo for making available the code used for the solution of the McKean-Vlasov-Fokker-Planck dynamics presented in Section~\ref{sec:evolution} and to P. Yatsyshin for help with the arclength continuation methodology. SG is currently supported by the EPSRC under grant No. EP/K034154/1. Part of the work was done while she was supported by the EPSRC under grant No. EP/L020564/1. GP is partially supported by the EPSRC under grants No. EP/P031587/1, EP/L024926/1, EP/L020564/1 and EP/L025159/1. 


\begin{thebibliography}{10}

\bibitem{Abdulle2017}
A.~Abdulle, G.A. Pavliotis, and U.~Vaes.
\newblock Spectral methods for multiscale stochastic differential equations.
\newblock {\em arxiv:1609.05097v1}, 2017.

\bibitem{ContAllgower}
E.~L. Allgower and K.~Georg.
\newblock {\em Introduction to Numerical Continuation Methods}.
\newblock Colorado State University, 1990.

\bibitem{Arnold1996}
A~Arnold, L.L. Bonilla, and P.A. Markowich.
\newblock Liapunov functionals and large-time-asymptotics of mean-field
  nonlinear {F}okker-{P}lanck equations.
\newblock {\em Transport Theory and Statistical Physics}, 25(7):733--751, 1996.

\bibitem{balescu97}
R.~Balescu.
\newblock {\em Statistical dynamics. Matter out of equilibrium}.
\newblock Imperial College Press, London, 1997.

\bibitem{BinneyTremaine2008}
J.~Binney and S.~Tremaine.
\newblock {\em Galactic Dynamics}.
\newblock Princeton University Press, Princeton, second edition, 2008.

\bibitem{BKRS2015}
V.~I. Bogachev, N.~V. Krylov, M.~R\"ockner, and S.~V. Shaposhnikov.
\newblock {\em Fokker-{P}lanck-{K}olmogorov equations}, volume 207 of {\em
  Mathematical Surveys and Monographs}.
\newblock American Mathematical Society, Providence, RI, 2015.

\bibitem{CMV2006}
J.~A. Carrillo, R.~J. McCann, and C.~Villani.
\newblock Contractions in the 2-{W}asserstein length space and thermalization
  of granular media.
\newblock {\em Arch. Ration. Mech. Anal.}, 179(2):217--263, 2006.

\bibitem{Carrillo2015}
J.A. Carrillo, A.~Chertock, and Y.~Huang.
\newblock A finite-volume method for nonlinear nonlocal equations with a
  gradient flow structure.
\newblock {\em Commun. Comput. Phys.}, 17(1):233--258, 2015.

\bibitem{ChayesPanferov2010}
L.~Chayes and V.~Panferov.
\newblock The {M}c{K}ean-{V}lasov equation in finite volume.
\newblock {\em J. Stat. Phys.}, 138(1-3):351--380, 2010.

\bibitem{Dawson1983}
D.~A. Dawson.
\newblock Critical dynamics and fluctuations for a mean-field model of
  cooperative behavior.
\newblock {\em J. Statist. Phys.}, 31(1):29--85, 1983.

\bibitem{matcont}
A.~Dhooge, W.~Govaerts, Yu.~A. Kuznetsov, W.~Mestrom, A.M. Riet, and
  B.~Sautois.
\newblock {\em MATCONT and CL MATCONT: Continuation toolboxes in matlab}.
\newblock Utrecht University, Netherlands and Universiteit Gent, Belgium, 2006.

\bibitem{Pavliotis_al2016}
A.~B. Duncan, S.~Kalliadasis, G.~A. Pavliotis, and M.~Pradas.
\newblock {Noise-induced transitions in rugged energy landscapes}.
\newblock {\em {Phys. Rev. E}}, {94}({3}), {SEP 6} {2016}.

\bibitem{DuncanPavliotis2016}
A.~B. Duncan and G.~A. Pavliotis.
\newblock Brownian motion in an {N}-scale periodic potential.
\newblock {\em arXiv:1605.05854}, 2016.

\bibitem{DuncanLelievrePavliotis2016}
A.B. Duncan, T.~Lelievre, and G.~A. Pavliotis.
\newblock Variance reduction using nonreversible langevin samplers.
\newblock {\em J. Stat. Phys.--to appear}, {152}(2):{ 237--274 }, {2016}.

\bibitem{Farkhooi2017}
F.~Farkhooi and W.~Stannat.
\newblock A complete mean-field theory for dynamics of binary recurrent neural
  networks.
\newblock {\em arXiv:1701.07128v1}, 2017.

\bibitem{frank04}
T.~D. Frank.
\newblock {\em Nonlinear {F}okker-{P}lanck equations}.
\newblock Springer Series in Synergetics. Springer-Verlag, Berlin, 2005.

\bibitem{GPY2012}
J.~Garnier, G.~Papanicolaou, and T.-W. Yang.
\newblock Large deviations for a mean field model of systemic risk.
\newblock {\em SIAM Journal of Financial Mathematics}, {4}({1}):{151--184},
  2013.

\bibitem{GPY2017}
J.~Garnier, G.~Papanicolaou, and T.-W. Yang.
\newblock Consensus convergence with stochastic effects.
\newblock {\em Vietnam J. Math.}, 45(1-2):51--75, 2017.

\bibitem{Gartner1988}
J.~G\"{a}rtner.
\newblock On the {M}c{K}ean -{V}lasov limit for interacting diffusions.
\newblock {\em Math. Nachr.}, 137:197--248, 1988.

\bibitem{GodPavlal-2012}
B.~D. Goddard, A.~Nold, N.~Savva, G.~A. Pavliotis, and S.~Kalliadasis.
\newblock {General Dynamical Density Functional Theory for Classical Fluids}.
\newblock {\em {Phys. Rev. Lett.}}, {109}({12}), {SEP 18} {2012}.

\bibitem{GodPavlKall11}
B.~D. Goddard, G.~A. Pavliotis, and S.~Kalliadasis.
\newblock The overdamped limit of dynamic density functional theory: rigorous
  results.
\newblock {\em Multiscale Model. Simul.}, 10(2):633--663, 2012.

\bibitem{HLZP2014}
C.~Hartmann, J.~C. Latorre, W.~Zhang, and G.~A. Pavliotis.
\newblock Optimal control of multiscale systems using reduced-order models.
\newblock {\em J. Comput. Dyn.}, 1(2):279--306, 2014.

\bibitem{HorsLef84}
W.~Horsthemke and R.~Lefever.
\newblock {\em Noise-induced transitions}, volume~15 of {\em Springer Series in
  Synergetics}.
\newblock Springer-Verlag, Berlin, 1984.
\newblock Theory and applications in physics, chemistry, and biology.

\bibitem{Imkeller2012}
P.~Imkeller, N.~S. Namachchivaya, N.~Perkowski, and H.~C. Yeong.
\newblock Dimensional reduction in nonlinear filtering: a homogenization
  approach.
\newblock {\em Ann. Appl. Probab.}, 23(6):2290--2326, 2013.

\bibitem{Krauskopf}
B.~Krauskopf.
\newblock {\em Numerical Continuation Methods for Dynamical Systems}.
\newblock Springer, 2007.

\bibitem{LelievreNierPavliotis2013}
T.~Lelievre, F.~Nier, and G.~A. Pavliotis.
\newblock Optimal non-reversible linear drift for the convergence to
  equilibrium of a diffusion.
\newblock {\em J. Stat. Phys.}, {152}(2):{ 237--274 }, {2013}.

\bibitem{Lucon2016}
E.~Lu\'{c}on and W.~Stannat.
\newblock Transition from gaussian to non-gaussian fluctuations for mean-field
  diffusions in spatial interaction.
\newblock {\em The Annals of Probability}, 26(6):3840--3909, 2016.

\bibitem{MartzelAslangul2001}
N.~Martzel and C.~Aslangul.
\newblock Mean-field treatment of the many-body {F}okker-{P}lanck equation.
\newblock {\em J. Phys. A}, 34(50):11225--11240, 2001.

\bibitem{McKean1966}
H.~P. McKean, Jr.
\newblock A class of {M}arkov processes associated with nonlinear parabolic
  equations.
\newblock {\em Proc. Nat. Acad. Sci. U.S.A.}, 56:1907--1911, 1966.

\bibitem{McKean1967}
H.~P. McKean, Jr.
\newblock Propagation of chaos for a class of non-linear parabolic equations.
\newblock In {\em Stochastic {D}ifferential {E}quations ({L}ecture {S}eries in
  {D}ifferential {E}quations, {S}ession 7, {C}atholic {U}niv., 1967)}, pages
  41--57. Air Force Office Sci. Res., Arlington, Va., 1967.

\bibitem{Motsch2014}
S.~Motsch and E.~Tadmor.
\newblock Heterophilious dynamics enhances consensus.
\newblock {\em SIAM Review}, 56(4):577--621, 2014.

\bibitem{Oelschlager1984}
K.~Oelschl\"{a}ger.
\newblock A martingale approach to the law of large numbers for weweak
  interacting stochastic processes.
\newblock {\em The Annals of Probability}, 12(2):458--479, 1984.

\bibitem{papav-2007}
A.~Papavasiliou.
\newblock Particle filters for multiscale diffusions.
\newblock In {\em Conference {O}xford sur les m\'ethodes de {M}onte {C}arlo
  s\'equentielles}, volume~19 of {\em ESAIM Proc.}, pages 108--114. EDP Sci.,
  Les Ulis, 2007.

\bibitem{PapPavSt08}
A.~Papavasiliou, G.~A. Pavliotis, and A.~M. Stuart.
\newblock Maximum likelihood drift estimation for multiscale diffusions.
\newblock {\em Stochastic Process. Appl.}, 119(10):3173--3210, 2009.

\bibitem{Pavl2014}
G.~A. Pavliotis.
\newblock {\em Stochastic processes and applications}, volume~60 of {\em Texts
  in Applied Mathematics}.
\newblock Springer, New York, 2014.
\newblock Diffusion processes, the Fokker-Planck and Langevin equations.

\bibitem{PavlSt06}
G.~A. Pavliotis and A.~M. Stuart.
\newblock Parameter estimation for multiscale diffusions.
\newblock {\em J. Stat. Phys.}, 127(4):741--781, 2007.

\bibitem{PavlSt08}
G.A. Pavliotis and A.M. Stuart.
\newblock {\em Multiscale Methods}, volume~53 of {\em Texts in Applied
  Mathematics}.
\newblock Springer, New York, 2008.
\newblock Averaging and Homogenization.

\bibitem{Pinnau_al2017}
R.~Pinnau, C.~Totzeck, O.~Tse, and S.~Martin.
\newblock A consensus-based model for global optimization and its mean-field
  limit.
\newblock {\em Math. Models Methods Appl. Sci.}, 27(1):183--204, 2017.

\bibitem{shiino1987}
M.~Shiino.
\newblock Dynamical behavior of stochastic systems of infinitely many coupled
  nonlinear oscillators exhibiting phase transitions of mean-field type: H
  theorem on asymptotic approach to equilibrium and critical slowing down of
  order-parameter fiuctuations.
\newblock {\em Physical Review A}, 36(5):2393--2412, 1987.

\bibitem{spiliop2013}
K.~Spiliopoulos.
\newblock Large deviations and importance sampling for systems of slow-fast
  motion.
\newblock {\em Appl. Math. Optim.}, 67(1):123--161, 2013.

\bibitem{Tamura1984}
Y.~Tamura.
\newblock On asymptotic behaviors of the solution of a non-linear diffusion
  equation.
\newblock {\em J. Fac. Sci. Univ. Tokyo}, 31:195--221, 1984.

\bibitem{Tamura1987}
Y.~Tamura.
\newblock Free energy and the convergence of distributions of diffusion
  processes of {M}c{K}ean type.
\newblock {\em J. Fac. Sci. Univ. Tokyo Sect. IA Math.}, 34(2):443--484, 1987.

\bibitem{Tugaut2014}
J.~Tugaut.
\newblock Phase transitions of {M}c{K}ean-{V}lasov processes in double-wells
  landscape.
\newblock {\em Stochastics}, 86(2):257--284, 2014.

\bibitem{Villani2003}
C.~Villani.
\newblock {\em Topics in optimal transportation}, volume~58 of {\em Graduate
  Studies in Mathematics}.
\newblock American Mathematical Society, Providence, RI, 2003.

\end{thebibliography}


\def\cprime{$'$} \def\cprime{$'$} \def\cprime{$'$} \def\cprime{$'$}
  \def\cprime{$'$} \def\cprime{$'$} \def\cprime{$'$} \def\cprime{$'$}
  \def\cprime{$'$} \def\Rom#1{\uppercase\expandafter{\romannumeral
  #1}}\def\u#1{{\accent"15 #1}}\def\Rom#1{\uppercase\expandafter{\romannumeral
  #1}}\def\u#1{{\accent"15 #1}}\def\cprime{$'$} \def\cprime{$'$}
  \def\cprime{$'$} \def\cprime{$'$} \def\cprime{$'$} \def\cprime{$'$}
  \def\cprime{$'$} \def\polhk#1{\setbox0=\hbox{#1}{\ooalign{\hidewidth
  \lower1.5ex\hbox{`}\hidewidth\crcr\unhbox0}}} \def\cprime{$'$}
  \def\cprime{$'$} \def\cprime{$'$}

\end{document}